\documentclass[12pt]{article}
\usepackage{amsmath}
\usepackage{graphicx}
\usepackage{enumerate}
\usepackage{natbib}
\usepackage{url} 
\usepackage{booktabs}
\usepackage{amsthm} 
\usepackage{nicefrac}

\usepackage{etoolbox}

\newcommand{\blind}{1}

\usepackage{amssymb}
\usepackage{xfrac} 
\usepackage{physics}
\usepackage{subcaption}
\usepackage{hyperref}
\usepackage{color}

\newtheorem{lemma}{Lemma}[section]
\newtheorem{prop}[lemma]{Proposition}  
\counterwithin{lemma}{section}
\newtheorem{assum}{Assumption}
\newtheorem{proposition}{Proposition}%
\newtheorem{remark}{Remark}%

\newtheorem{remarkapp}[lemma]{Remark} 

\DeclareMathOperator*{\esssup}{ess\,sup}
\DeclareMathOperator*{\essinf}{ess\,inf}

\begin{document}

\setcounter{section}{0}
\renewcommand{\thesection}{\arabic{section}}

\def\spacingset#1{\renewcommand{\baselinestretch}%
{#1}\small\normalsize} \spacingset{1}


\date{}
\if1\blind
{
  \title{\bf A new class of non-stationary Gaussian fields with arbitrary smoothness on metric graphs}
  \author{
    David Bolin\thanks{The authors are listed alphabetically.}, 
    Lenin Riera-Segura\footnotemark[1], 
    and Alexandre B. Simas\footnotemark[1]\\
    Statistics Program, Computer, Electrical and Mathematical\\
    Sciences and Engineering (CEMSE) Division,\\
    King Abdullah University of Science and Technology (KAUST),\\
    Thuwal, 23955-6900, Kingdom of Saudi Arabia
  }
  \maketitle
} \fi

\if0\blind
{
  \bigskip
  \bigskip
  \bigskip
  \begin{center}
    {\LARGE\bf Computationally efficient inference for non-stationary Gaussian fields with general smoothness on  metric graphs}
\end{center}
  \medskip
} \fi

\bigskip
\begin{abstract}
    The increasing availability of network data has motivated statistical models on metric graphs, with Gaussian processes playing a central role. Existing models such as Whittle--Mat\'ern fields are limited in their ability to handle non-stationary covariance structures and arbitrary smoothness. We propose a new class of generalized Whittle--Mat\'ern fields on compact metric graphs allowing both non-stationarity and arbitrary smoothness, and establish new regularity results that also apply to the Whittle--Mat\'ern setting. We further develop an efficient approximation of the covariance operator using finite elements and rational approximations of fractional powers, enabling scalable Bayesian inference for large datasets. Explicit convergence rates are derived, and the approach is validated through simulations and an application to traffic speed data, demonstrating the flexibility and effectiveness of the proposed model class.
\end{abstract}

\noindent%

\spacingset{1} 

\section{Introduction}
\label{sec:intro}

In recent years, the ever-increasing availability of data collected from networks such as streets or water systems has encouraged researchers to propose and develop new statistical models for the analysis of network data \citep{Borovitskiy2021Matern, Bolin2023Statistical,Cressie2006Spatial, VerHoef2006Spatial}. A network of this type can be conveniently represented by a metric graph, which is a graph equipped with a notion of distance and where the edges are curves that connect the vertices \citep{Berkolaiko2013Introduction}. The difference from a combinatorial graph is thus that the edges are not only connecting the vertices, but are curves on which we want to define statistical models (see Figure~\ref{graph_and_sim} for a simple example).

As on Euclidean domains, the formulation of Gaussian processes is foundational for developing new statistical models on metric graphs. For practical applications, it is typically important that the class of Gaussian processes allows for

\begin{enumerate}[(i)]
    \item arbitrary smoothness, indexed explicitly by a parameter that can be estimated from data, which is particularly important for asymptotically optimal spatial prediction and interpolation \citep{Stein1999Interpolation, Kirchner2022Necessary};

    \item non-stationary covariance functions, where variances and practical correlation ranges can vary and be controlled, which is crucial when modeling on large or heterogeneous domains \citep{Paciorek2006Spatial, Xiong2024Covariance}; 

    \item computationally efficient inference and prediction, so that it can be applied to large data sets \citep{Lindgren2011AnExplicit, Heaton2019ACase}.  
\end{enumerate}

There have been several attempts to define Gaussian processes on metric graphs. An early example of these can be found in \cite{VerHoef2006Spatial} which studied Gaussian processes designed for river networks. More recently, \cite{Anderes2020Isotropic}  defined Gaussian processes with valid isotropic covariance functions for graphs with Euclidean edges. Their approach has later been extended to spatio-temporal models \citep{Porcu2023Stationary, Filosi2023Temporally} and log-Gaussian Cox processes \citep{Moller2024Cox}. However, the restriction to have Euclidean edges, for example, excludes graphs where multiple edges connect the same two vertices and can be rather restrictive for applications. An example where this is evident is the application we will present later, which is defined on a graph that does not have Euclidean edges (see Figure~\ref{replicate14}). Other drawbacks of these models are that they cannot be used to model differentiable processes (which thus is a restriction to our first requirement), and they are by construction isotropic and can therefore not model non-stationarity.  Other attempts to define Gaussian processes for data on metric graphs can be found in \citet{Borovitskiy2021Matern} and \citet{SanzAlonso2022TheSPDE}, where both used models based on the so-called graph Laplacian. However, contrary to the processes proposed by \citet{Anderes2020Isotropic}, these are defined only on the vertices. 

\begin{figure}[t]
    \centering
    \includegraphics[width=0.99\textwidth]{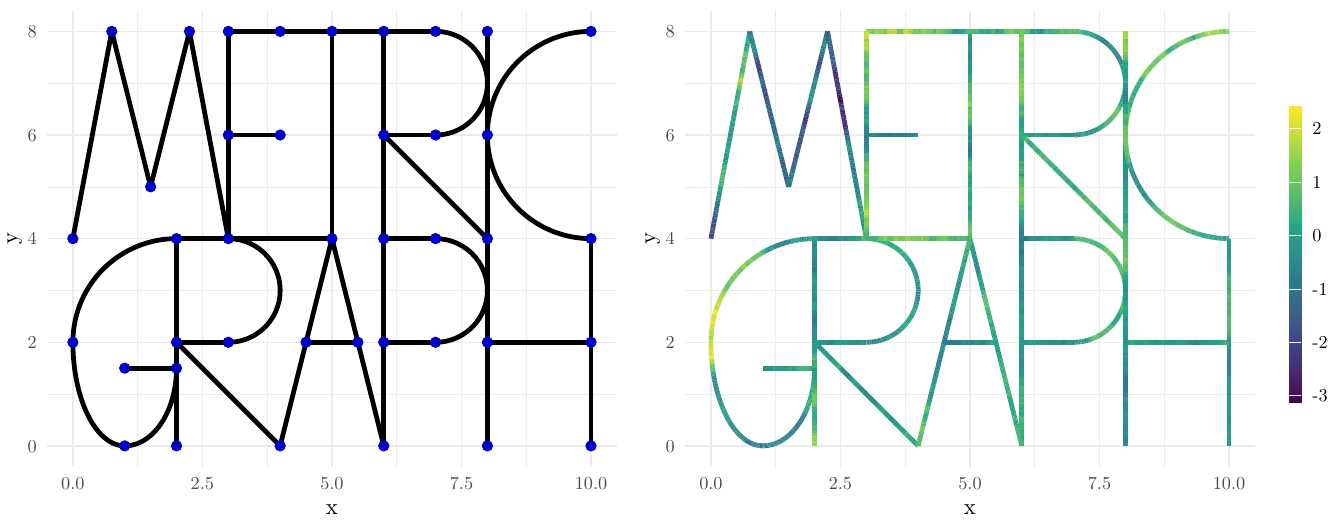}
    \caption{\texttt{MetricGraph} package's logo as a metric graph, with vertices in blue and edges in black (left), and a simulated non-stationary process on it (right).}
    \label{graph_and_sim}
\end{figure}

Recently, \citet{Bolin2024Gaussian} introduced Whittle--Mat\'ern fields as a class of Gaussian processes defined as solutions $u$ to the fractional order differential equation 
\begin{equation}\label{eq:spde}
(\kappa^2 - \Delta_{\Gamma})^{\sfrac{\alpha}{2}}(\tau u) = \mathcal{W}, \quad \text{on $\Gamma$}, 
\end{equation}
for a metric graph $\Gamma$, where $\Delta_{\Gamma}$ is the so-called Kirchhoff--Laplacian (see later sections for details), $\tau,\kappa>0$ control the marginal variance and practical correlation range, respectively, $\alpha>\sfrac{1}{2}$ controls the smoothness, and $\mathcal{W}$ is Gaussian white noise defined on a probability space $(\Omega,\mathcal{F},\mathbb{P})$. Defining the Whittle--Mat\'ern fields in this way is motivated by the fact that when this equation is considered on $\mathbb{R}^p$ (and $\Delta_{\Gamma}$ is replaced by the standard Laplacian), the solution $u$ is a centered Gaussian random field with an isotropic Mat\'ern covariance function \citep{Matern1960Spatial}.

Thus, the Whittle--Mat\'ern fields is a natural metric graph analogue to Gaussian processes with Mat\'ern covariance functions on $\mathbb{R}^d$ and they possess some desirable properties. For instance, they are specified on the entire metric graph (not only on the vertices) and are well-defined for any compact metric graph. Furthermore, they can model fields of arbitrary smoothness and, for example, have differentiable sample paths whenever $\alpha>\sfrac{3}{2}$. \cite{Bolin2023Statistical} showed that computationally efficient inference can be performed for these models whenever $\alpha\in\mathbb{N}$, and that having differentiable processes can be beneficial for spatial prediction. 

\cite{Bolin2024Regularity} later showed that these fields can be generalized to have non-stationary covariance functions by considering a more general differential operator in \eqref{eq:spde}. Specifically, they introduced the so-called Generalized Whittle--Mat\'ern fields as solutions to the differential equation $L^{\sfrac{\alpha}{2}} u = \mathcal{W}$ where $L$ is a second order elliptic operator given by $L = \kappa^2 - \nabla(H\nabla)$ for sufficiently smooth functions $\kappa$ and $H$. They showed that one can sample these models approximately by using a combination of the finite element method (FEM) and a quadrature approximation of the fractional power of the operator. Thus, these fields satisfy the first two requirements listed above. However, \cite{Bolin2024Regularity} did not consider statistical inference, and their numerical approximation is not suitable for statistical applications as it is computationally demanding and incompatible with popular inference tools such as \texttt{R-INLA} \citep{Lindgren2015Bayesian}. See \citet{Xiong2024Covariance} for a discussion about this for models on Euclidean domains. 

In this work, we propose a version of the Generalized Whittle--Mat\'ern fields, defined through equation \eqref{eq:spde} but where both $\kappa$ and $\tau$ are allowed to be spatially varying functions. An advantage of this model compared to those considered by \cite{Bolin2024Regularity} is that they provide more direct control over marginal variances and local practical correlation ranges, and they can be viewed as the metric graph analogue of the non-stationary SPDE-based models originally proposed by \cite{Lindgren2011AnExplicit} for Euclidean domains and manifolds. However, unlike the Euclidean setting, where boundary effects are confined to the exterior of the domain and do not influence interior regularity, metric graphs exhibit boundary-like behavior at internal vertices, which has a critical impact on pathwise regularity. In particular, the regularity of the solution may depend in a delicate way on the global regularity of $\tau$, a phenomenon that appears to be intrinsic to metric graphs and does not arise in classical Euclidean or manifold-based settings. To the best of our knowledge, this type of pathwise regularity analysis has not previously appeared in the literature.

As an important special case of the model class, we introduce the variance-stationary Whittle--Mat\'ern fields, which are defined so that they have constant variance throughout the domain. This modification preserves the flexibility and computational advantages of Whittle--Mat\'ern models on metric graphs while avoiding undesirable boundary and vertex effects in marginal variability.

For the proposed model class, we introduce a computationally efficient approximation which is suitable for Bayesian or likelihood-based inference. The method is theoretically justified by deriving explicit convergence rates for the covariance function of the approximation and we show how to use it for computationally efficient inference. The method is illustrated through simulation experiments and an application to traffic speed observations, where the ability to estimate the smoothness from data and model non-stationarity improves the predictive power. To complement the empirical application, we additionally perform a simulation study to quantify the effects of smoothness misspecification on parameter estimation and predictive accuracy.

The remainder of the paper is organized as follows. Section~\ref{sec:model} provides preliminary concepts and notation, defines the model class, and illustrates key theoretical properties. Section~\ref{num_approx_sec} introduces the numerical approximation method and its properties. In Section~\ref{numerical_experiments}, we perform numerical experiments to show the accuracy of the proposed approximation and to verify the theoretical results. Section \ref{application} is devoted to the application and some final conclusions are given in Section~\ref{sec:conc}. Technical details, proofs and further simulation results can be found in the appendices. The proposed model and the approximation methods are implemented in the \texttt{MetricGraph} package \citep{MetricGraphRpackage}, available on CRAN. This package facilitates using the models in general latent Gaussian models which can be fitted to data using \texttt{R-INLA}.

\section{A class of non-stationary Gaussian random fields}
\label{sec:model}

\subsection{Preliminaries}
\label{subsec:prelim}

First, let us provide the exact definition of a metric graph and introduce some notation. A compact metric graph is a pair $\Gamma = (\mathcal{V}, \mathcal{E})$, where $\mathcal{V}$ is a finite set of vertices and $\mathcal{E}$ is a finite set of undirected edges (each edge is a rectifiable curve). Elements of $\mathcal{V}$ and $\mathcal{E}$ will be denoted by $v\in\mathcal{V}$ and $e\in\mathcal{E}$, respectively. The graph is equipped with a metric $d(\cdot,\cdot)$, which we assume is the geodesic distance. We further assume that $\Gamma$ is connected so that there is a path between any two points in $\Gamma$. Each edge $e\in\mathcal{E}$ has a positive length $\ell_{e}\in(0,\infty)$ and connects two vertices $\underline{e}$ and $\overline{e}$ in $\mathcal{V}$. For each vertex $v\in\mathcal{V}$ let $\mathcal{E}_{v}$ denote the collection of edges incident to $v$, and set $L_v=\{e\in\mathcal{E}_v:\; e \text{ is a loop}\}$. The degree of $v$ is defined by $\deg(v)=|\mathcal{E}_v|+|L_v|$.

A location $s\in\Gamma$ is a point lying on some edge $e\in\mathcal{E}$ and may be expressed by the ordered pair $(e,t)$ with $t\in[0,\ell_e]$; we call $(e,t)$ a coordinate representation of $s$. A vertex $v\in\mathcal{V}$ admits $\deg(v)$ distinct coordinate representations, whereas any point $s\in\Gamma\setminus\mathcal{V}$ has a unique coordinate representation. We write $s\in v$ to indicate that $s$ is a coordinate representation of the vertex $v$.

Given a function $f$ on $\Gamma$ and an edge $e\in\mathcal{E}$, denote by $f_e=f|_e$ the restriction of $f$ to $e$, and for $s=(e,t)$ set $f(s)=f_e(t)$. Thus, a function $f$ on $\Gamma$ can be identified with the collection of its edgewise restrictions $\{f_e\}_{e\in\mathcal{E}}$. If for every $e\in\mathcal{E}$, $f_e\in L_2(e)$, we say that $f\in L_2(\Gamma)$, where $L_2(e) $ denotes the standard Lebesgue space on the interval $[0, \ell_e]$. Further, we let $(f,g)_{L_2(\Gamma)} = \int_\Gamma f(s)g(s)ds = \sum_{e\in\mathcal{E}}\int_{e}f_{e}(s)g_{e}(s)ds$ and its induced norm $\|f\|^2_{L_2(\Gamma)} = \sum_{e\in\mathcal{E}}\|f_{e}\|^2_{L_2(e)}$, where $f = \{f_e\}_{e \in \mathcal{E}}, g = \{g_e\}_{e \in \mathcal{E}} \in L_2(\Gamma)$. 

We let $C(\Gamma)$ be the space of continuous functions on $\Gamma$ with norm $ \|f\|_{C(\Gamma)} = \sup_{s \in \Gamma} |f(s)| $, and let $L_\infty(\Gamma)$ denote the space of essentially bounded functions, with norm $ \|f\|_{L_\infty(\Gamma)} = \esssup_{s \in \Gamma} |f(s)| $. Additionally, we consider the space of $\gamma$-H\"older continuous functions, $ C^{0,\gamma}(\Gamma) $, for $ 0 < \gamma \leq 1 $, with seminorm $[f]_{C^{0,\gamma}(\Gamma)} = \sup_{s, s' \in \Gamma} \sfrac{|f(s) - f(s')|}{d(s, s')^\gamma},$ and norm $\|f\|_{C^{0,\gamma}(\Gamma)} = \|f\|_{C(\Gamma)} + [f]_{C^{0,\gamma}(\Gamma)}$. On individual edges, these function spaces coincide with their standard definitions. Specifically, $ C(e) $ denotes the space of continuous functions, while for integer $k\ge 0$ and $0<\gamma\le 1$, $C^{k,\gamma}(e)$ denotes the H\"older space of functions with $k$ continuous derivatives whose $k$-th derivative is $\gamma$-H\"older continuous on $e$. Notably, a $\gamma$-H\"older continuous function with $\gamma = 1$ is referred to as a Lipschitz function. For the entire graph $\Gamma$, and integer $k\geq0$ and $0 < \gamma \leq 1$, we let $C^{k,\gamma}(\Gamma)$ denote the H\"older space of functions $f\in \bigoplus_{e\in\mathcal{E}} C^{k,\gamma}(e)$ such that $f^{(2j)}\in C^{0,\gamma}(\Gamma)$ for every $j\in\{0,1,2,\dots,\lfloor\sfrac{k}{2}\rfloor\}$. Thus, functions in $C^{k,\gamma}(\Gamma)$ possess $C^{k,\gamma}(e)$-regularity and, in addition, their even-order derivatives up to order $k$ extend to $\gamma$-H\"older continuous on $\Gamma$. For instance, if $f\in C^{2,\gamma}(\Gamma)$, then $f,f''\in C^{0,\gamma}(\Gamma)$. We also let $C^{k,0}(e)$ denote the space of $k$-times continuously differentiable real functions on $e$. Lastly, we define the subspace of continuous functions that satisfy Kirchhoff vertex conditions as 

\begin{equation}
\label{eq:kirchhoff_cond}
   \mathcal{K}(\Gamma) =  \left\{f\in C(\Gamma)\cap \bigoplus_{e\in\mathcal{E}} C^{1,0}(e)\;\middle|\; \forall v\in \mathcal{V}:\; \sum_{s\in v}\partial f(s)=0 \right\},
\end{equation}
where, for a coordinate representation $s\in v$, $\partial u(s)$ denotes the directional derivative of $u$ along the incident edge taken in the direction away from the vertex $v$. Concretely, if $e=[0,\ell_e]$ and $s=(e,0)\in v$ then $\partial u(s)=u_e'(0)$, whereas if $s=(e,\ell_e)\in v$ then $\partial u(s)=-u_e'(\ell_e)$. Let $f\in C^{1,\gamma}(\Gamma)$ with $\gamma\in [0,1]$. We say that $f'\in C^{0,\gamma}(\Gamma)$ if for every $s\in\Gamma$ and every simple path $p\subset\Gamma$ containing $s$ in its interior, the derivative $(f|_p)'$ of the restriction belongs to $C^{0,\gamma}(p)$. Here $p$ is parameterized by arc length as $[0,\ell_p]$. If $s$ is a vertex, the edge directions are chosen compatibly with the orientation of $p$ so that this parameterization is well-defined. If $\gamma=0$, we simply write $f'\in C(\Gamma)$. Note that care must be taken when interpreting odd-order derivatives as they depend on the orientation of edge parameterizations, which is not the case for even-order derivatives, see Appendix~\ref{app:theoretical_details} and \cite{Awadelkarim2025Fractional} for theoretical details.

\subsection{The model class}
\label{sec:model_class}
As previously mentioned, we consider Gaussian fields specified as solutions to \eqref{eq:spde} when $\kappa$ and $\tau$ are functions which will determine the marginal variances and practical correlation ranges of the process. To guarantee the existence of the these processes, we make the following assumptions on $\kappa$ and $\tau$. 
\begin{assum}
\label{assumption1}
    Let $\alpha>\sfrac{1}{2}$. If $\sfrac{1}{2}<\alpha\leq2$,
    we assume that  $\kappa, \tau \in L_\infty(\Gamma)$ and that there exist $\kappa_0,\tau_0>0$ such that $\essinf_{s\in\Gamma}\kappa(s)\geq \kappa_0$ and $\essinf_{s\in\Gamma}\tau(s)\geq \tau_0$. If $\alpha>2$, we additionally require that $\kappa\in \mathcal{K}(\Gamma)$  and $\kappa_e\in C^{\lceil\alpha\rceil-3,1}(e)$ for all $e\in\mathcal{E}$.
\end{assum}

We now formally define the model. The Kirchhoff--Laplacian $\Delta_{\Gamma}$ in \eqref{eq:spde} is a differential operator that acts on a function $f$ by taking its second derivative on each edge, provided that $ f_e $ has a well-defined second derivative on every edge $e$. This operator is coupled with the Kirchhoff vertex conditions \eqref{eq:kirchhoff_cond}. Given $ L = \kappa^2 - \Delta_{\Gamma} $, the fractional power $ L^\alpha $, for $ \alpha > 0 $, is defined in the spectral sense. This definition is valid because $ L $ is a densely defined, self-adjoint operator with a compact resolvent, ensuring a well-posed spectral decomposition. In particular, the inverse fractional power $ L^{-\alpha} $ is also well-defined. For further details, see Appendix~\ref{app:uniqueness}.

The final component needed to understand \eqref{eq:spde} is the definition of Gaussian white noise. It is represented as a family of centered Gaussian random variables $\{\mathcal{W}(h) : h \in L_2(\Gamma)\}$ satisfying $\mathbb{E}[\mathcal{W}(h) \mathcal{W}(g)] = (h, g)_{L_2(\Gamma)}$ for all $h, g \in L_2(\Gamma)$. A solution to \eqref{eq:spde} is a centered Gaussian random field $ u $ in $ L_2(\Gamma) $ satisfying, for all $ h \in L_2(\Gamma)$, $(u, h)_{L_2(\Gamma)} = \mathcal{W}(L^{-\sfrac{\alpha}{2}}(\tau^{-1} h))$, where $ \tau^{-1} h $ is the function defined by $ (\tau^{-1} h)(s) = \tau^{-1}(s) h(s) $ for all $ s \in \Gamma $. For further details, refer to Appendix~\ref{app:uniqueness}. In the next proposition, we need the following compatibility assumption for strict positive-definiteness of the corresponding covariance function.

\begin{assum}
\label{assumption2}
    Let Assumption \ref{assumption1} hold. If $\alpha > 2.5$, define $m_0=\left\lfloor\frac{\alpha}{2}-\frac{1}{4}\right\rfloor$ and assume that $L^{m_0-1}\kappa^2 \in C(\Gamma)$. Furthermore, if $\alpha > 3.5$, define $m_1=\left\lfloor\frac{\alpha}{2}-\frac{3}{4}\right\rfloor$ and additionally assume that $L^{m}\kappa^2 \in \mathcal{K}(\Gamma)$ for $m=0,\ldots,m_1-1$.
\end{assum}

\begin{proposition}
\label{prp:existence}
If Assumption~\ref{assumption1} holds, \eqref{eq:spde} has a unique solution $u$ which is a centered Gaussian process with covariance function $\varrho^{\alpha}(s,t) = \text{Cov}(u(s), u(t))$. If, additionally, Assumption~\ref{assumption2} holds, this covariance function is strictly positive definite.
\end{proposition}

The main difficulty in proving this novel result is to show the strict positive definiteness of $\varrho^{\alpha}$, a property that is highly useful for statistical applications.

\subsection{Regularity of the field}\label{subsec:regularity}
As mentioned in the introduction, a key property of Gaussian random fields is their ability to exhibit varying degrees of regularity. While Proposition \ref{prp:existence} establishes the existence of solutions, it does not address their smoothness. Our goal here is to investigate the regularity of these solutions. The parameter $\alpha$ plays a critical role in determining regularity. However, due to the way $\tau$ appears in \eqref{eq:spde}, additional assumptions on $\tau$ are required to achieve a desired level of regularity. This observation is particularly relevant for practical applications, as we will see later. 

\begin{remark}
Throughout this work, smoothness is understood in terms of sample path regularity. Mean-square regularity is not discussed separately, as it is a strictly weaker notion for Gaussian processes and therefore follows automatically. Classical covariance-based characterizations of smoothness, which are typically derived under stationarity assumptions, are not directly applicable in this setting due to the non-stationary nature of the proposed models. Instead, smoothness is controlled through the fractional power of the defining differential operator, as is standard in SPDE formulations of non-stationary Gaussian processes.
\end{remark}

We distinguish between two types of regularity: \emph{local regularity}, which examines the regularity of the field when restricted to individual edges, and \emph{global regularity}, which considers the regularity on the entire metric graph. The main result concerning local regularity is presented below. In the result and later, when we state that a Gaussian process $u$ belongs to a function space $F$ we mean that there exists a modification of $ u $ such that the sample paths of the modified field belong to $F$.

\begin{proposition}
\label{prp:regularity}
    Let $u$ be the solution to \eqref{eq:spde} under Assumption~\ref{assumption1}. Then:
    \begin{enumerate}
        \item[(i)] If $\alpha>\sfrac{1}{2}$, then for any $\gamma$ such that $0 < \gamma < \min\{\alpha-\sfrac{1}{2}, 1\}$ and $\tau_e \in C^{0, \gamma}(e)$ for all $e \in \mathcal{E}$, we have $u_e \in C^{0, \gamma}(e)$ for every $e \in \mathcal{E}$.

        \item[(ii)] If $\alpha>\sfrac{3}{2}$, then for any $\gamma$ such that $0 < \gamma < \min\{\alpha-\sfrac{3}{2}, 1\}$ and $\tau_e \in C^{1, \gamma}(e)$ for all $e \in \mathcal{E}$, we have $u_e \in C^{1, \gamma}(e)$ for every $e \in \mathcal{E}$.

        \item[(iii)] If $\alpha>\sfrac{5}{2}$, then for any $\gamma$ such that $0 < \gamma < \min\{\alpha-\sfrac{5}{2}, 1\}$ and $\tau_e \in C^{2, \gamma}(e)$ for all $e \in \mathcal{E}$, we have $u_e \in C^{2, \gamma}(e)$ for every $e \in \mathcal{E}$.
    \end{enumerate}
\end{proposition}
Proposition~\ref{prp:regularity} establishes edgewise regularity, where the smoothness of $u$ on each edge depends only on $\alpha$ and the behavior of $\tau_e$. Global regularity follows by imposing additional assumptions on $\tau$ and $\kappa$ that ensure compatibility at the vertices. These conditions preserve the edgewise regularity from Proposition~\ref{prp:regularity} and, in addition, yield continuity and higher-order smoothness of $u$ across the entire graph. Thus, stronger assumptions on the coefficients lead to correspondingly stronger global regularity on the graph. Note that, in the next proposition, the condition $\tau' \in C^{0,\gamma}(\Gamma)$ appearing in part (iii) below refers to the pathwise definition from Section~\ref{subsec:prelim}, requiring regularity of the derivative along every simple path through each point (with compatible orientations at vertices).

\begin{proposition}
\label{prp:regularity-global}
    Under the same setting as Proposition~\ref{prp:regularity}:
    \begin{enumerate}
        \item[(i)] If $\tau \in C^{0,\gamma}(\Gamma)$, then $u \in C^{0,\gamma}(\Gamma)$. Furthermore,  $\tau \in C(\Gamma)$ if and only if $u \in C(\Gamma)$.
        \item[(ii)] We have $u \in \mathcal{K}(\Gamma)$ if and only if $\tau \in \mathcal{K}(\Gamma)$.
        \item[(iii)] Let $\kappa \in C^{0,\gamma}(\Gamma)$. If $\tau\in C^{2,\gamma}(\Gamma)$ and $\tau' \in C^{0,\gamma}(\Gamma)$, then $u''\in C^{0,\gamma}(\Gamma)$. Furthermore, if $\tau,\tau'' \in C(\Gamma)$, then  $\tau' \in C(\Gamma)$ if and only if $u''\in C(\Gamma)$.
    \end{enumerate}
\end{proposition}
We emphasize that the results presented in Propositions \ref{prp:regularity} and \ref{prp:regularity-global} are novel. In particular, Proposition~\ref{prp:regularity-global} (ii) is new even in the context of standard Whittle--Mat\'ern fields as it improves upon \citet[Proposition 11]{Bolin2024Gaussian}, where $\alpha \geq 2$ was required for standard Whittle--Mat\'ern fields. Here, we relax this condition to $\alpha > \sfrac{3}{2}$.

The dependence of the global regularity of the solution on the global regularity of $\tau$ and $\kappa$ is, to the best of our knowledge, a completely novel phenomenon, and it plays a particularly critical role in the context of metric graphs. Unlike in the Euclidean setting, where boundary extensions can often mitigate boundary effects, vertex conditions on metric graphs are unavoidable. Moreover, as $\alpha$ increases, achieving the desired regularity requires increasingly stringent conditions on $\tau$. While $\kappa$ is also required to satisfy regularity assumptions, these are comparatively milder, highlighting the importance of careful modeling of $\tau$ in practical applications. To illustrate this, consider the graph in Figure~\ref{kappa_disc}. Let $f(s) = \mathrm{edge.number}(s)/4$ and $g(s) = 0.5 \cdot (x^2(s) - y^2(s)) + 0.5$, where $(x(s), y(s))$ are Euclidean coordinates on the plane. A simulation of the Gaussian process $u$ with $\alpha = 3$ is shown in the top row of the figure, with $\tau(s) = \exp(g(s))$ and $\kappa(s) = \exp(f(s))$. In this case, the field is continuous. However, when $\tau(s) = \exp(f(s))$ and $\kappa(s) = \exp(g(s))$, the simulated field is discontinuous, as illustrated in the bottom row of Figure~\ref{kappa_disc}.

\begin{figure}[t]
\centering
\includegraphics[width=0.99\textwidth]{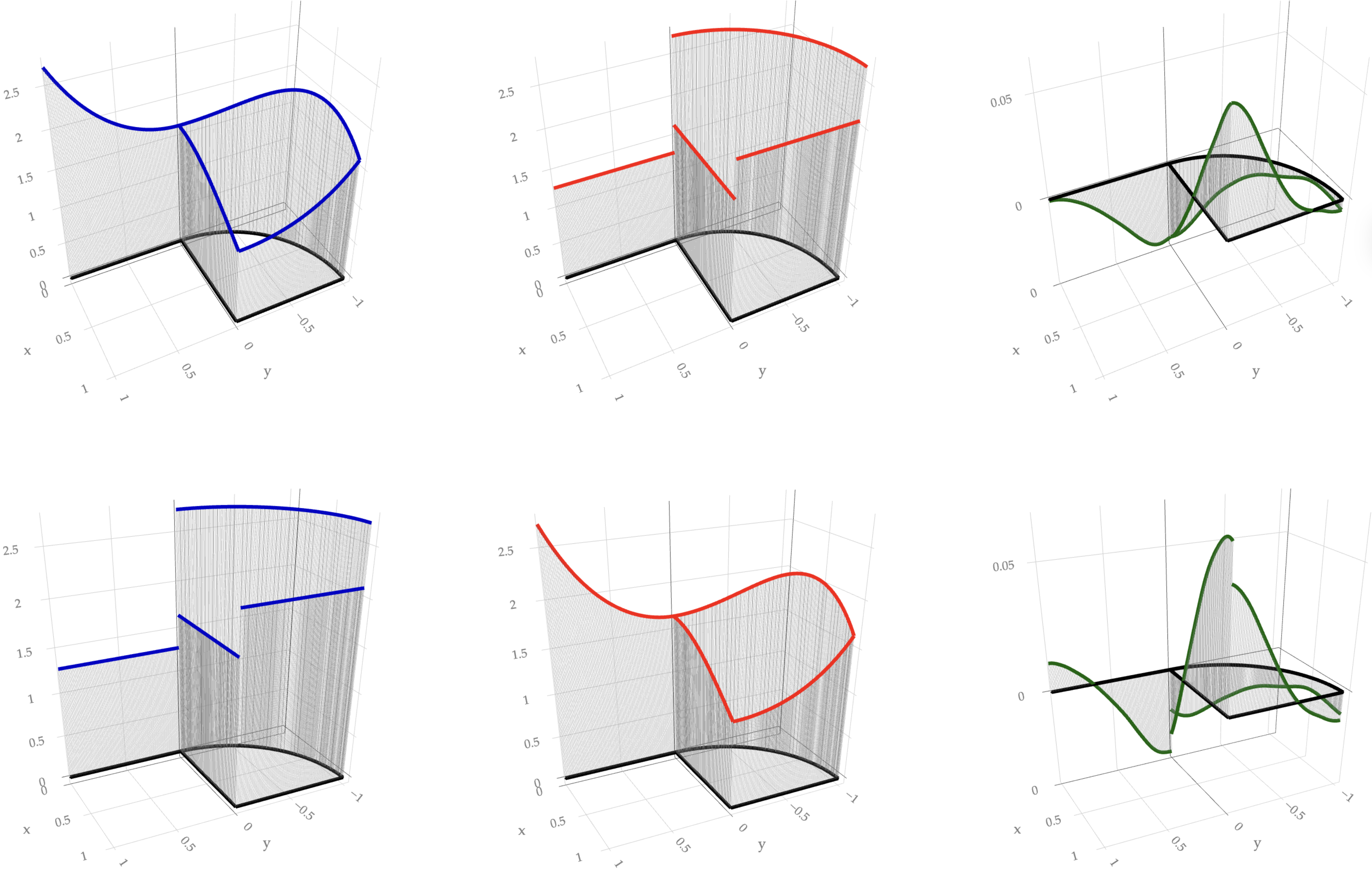}
\caption{Models for $\tau(\cdot)$ (left), $\kappa(\cdot)$ (center), and their simulated fields $u(\cdot)$ (right).}
\label{kappa_disc}
\end{figure}

\subsection{Illustrations of the model properties}
\label{subsec:model_properties}

We illustrate the flexibility of the proposed model using the \texttt{MetricGraph} package's logo in Figure~\ref{graph_and_sim}. To visualize the versatility of our method in terms of non-stationarity, let $\alpha=0.9$, $\tau(s) = e^{0.05\cdot(x(s)-y(s))}$ and $\kappa(s) = e^{0.1\cdot(x(s)-y(s))}$, where $(x(s),y(s))$ are Euclidean coordinates on the plane. The left panel of Figure \ref{cov_diff_range} shows three locations, $s_1$, $s_2$, and $s_3$, for which we plot $\varrho^{\alpha}(s_i,\cdot)$, $i=1,2,3$. We can see that $\sigma(s_i)$ controls the magnitude of $\varrho^{\alpha}(s_i,\cdot)$ and $\rho(s_i)$ controls how quickly the function decays away from $s_i$. Here $\rho(s)$ is the smallest distance $d(s,t)$ so that the correlation between $u(s)$ and $u(t)$ is below $0.1$. See Appendix \ref{app:model_prop} for more details. The right panel in Figure \ref{graph_and_sim} shows a simulation of a non-stationary field using the previous choices of $\alpha$, $\tau(\cdot)$, and $\kappa(\cdot)$. To exhibit the flexibility of our method concerning  arbitrary smoothness, let $\tau(\cdot)$ and $\kappa(\cdot)$ as before, and with these fixed, consider the cases $\alpha_i \in \{0.9,1.3,2.1\}$. The right panel of Figure \ref{cov_diff_range} shows an illustration of $\varrho^{\alpha_i}(s_0,s)$ for a fixed location $s_0$ for these choices of $\alpha_i$.

\begin{figure}[t]
\centering
\includegraphics[width=0.99\textwidth]{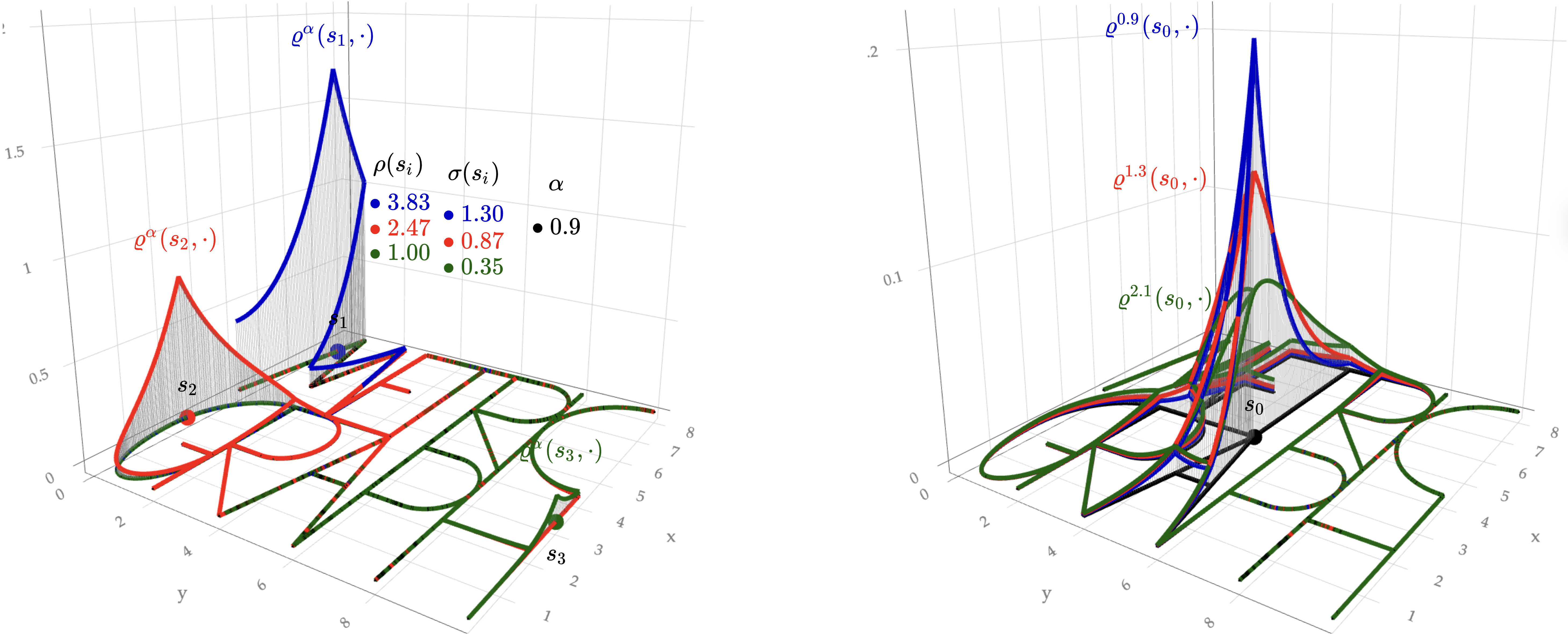}
\caption{Covariance functions $\varrho^{\alpha}(s_i,\cdot)$, $i=1,2,3$ for the model with $\alpha = 0.9$,   $\tau(s) = e^{0.05\cdot(x(s)-y(s))}$ and $\kappa(s) = e^{0.1\cdot(x(s)-y(s))}$ (left), and covariance functions $\varrho^{\alpha_i}(s_0,s)$, $i=1,2,3$ for the same choices of $\tau(\cdot)$ and $\kappa(\cdot)$, with $\alpha_i\in \{0.9,1.3,2.1\}$ (right).}
\label{cov_diff_range}
\end{figure}

\subsection{Variance-stationary Whittle--Mat\'ern fields}
\label{tau_rol}
    
One feature of the Whittle--Mat\'ern fields on metric graphs is that they are inherently non-isotropic \citep{Bolin2024Gaussian}. Their marginal variances and practical correlation ranges change over the graph even if $\kappa$ and $\tau$ are constants. That the practical correlation range is non-stationary is often realistic for applications \citep{Bolin2024Gaussian}. However, a non-constant variance might be less ideal in some applications. An example of the marginal standard deviations of a Whittle--Mat\'ern field with $\alpha=1$, $\kappa=2$, and $\tau=1$, on a simple graph are shown in Figure~\ref{var1}. One can note that the variance is higher close to vertices of degree 1 and lower close to vertices of degree larger than two. 

An interesting option to mitigate this effect is a nonstationary model defined via 

\begin{equation}
\label{eq:spde_variance_stat}
    (\kappa^2 - \Delta_{\Gamma})^{\sfrac{\alpha}{2}} (\sigma_{\kappa} u) = \sigma_0 \mathcal{W}, \quad \text{on } \Gamma,
\end{equation}
where $\sigma_0 > 0$ and $\alpha > \sfrac{1}{2}$ are constants, $\kappa$ satisfies Assumption~\ref{assumption1}, and $\sigma_{\kappa}(s)$ represents the marginal standard deviations of a generalized Whittle--Mat\'ern field $w$ with the same $\kappa$ and $\alpha$, but with $\tau = 1$. This is a special case of \eqref{eq:spde} with  $\tau(s) = \sigma_0^{-1} \sigma_{\kappa}(s)$. Under this specification, the model parameters are $(\kappa, \sigma_0, \alpha)$, and it follows directly that the variance of $ u(s) $ is $\sigma_0^2$ for all $ s \in \Gamma $. Indeed, by \eqref{eq:spde_variance_stat}, we have that $w = \sigma^{-1}_0\sigma_\kappa u$. Since $\mathbb{V}(w(s)) = \sigma_\kappa^2(s)$, it follows that $\mathbb{V}(u(s)) = \mathbb{V}(\sigma_0\sigma_\kappa^{-1}(s)w(s)) = \sigma_0^2\sigma_\kappa^{-2}(s)\mathbb{V}(w(s)) = \sigma_0^2\sigma_\kappa^{-2}(s)\sigma_\kappa^{2}(s) = \sigma_0^2$.

For this reason, we refer to this model as a \emph{variance-stationary Whittle--Mat\'ern field}. The right panel of Figure~\ref{var1} illustrates the marginal variance of the field for $\sigma_0 = 1$, showing that the variance is indeed 1 across all locations. The following result demonstrates that this choice of $\tau$ satisfies a global regularity assumption.

\begin{proposition}
\label{prp:variance_stationary_regularity}
    Let $\alpha>\sfrac{1}{2}$, $\kappa$ satisfy Assumption~\ref{assumption1}, and $\tau(s) = \sigma_0^{-1} \sigma_{\kappa}(s)$ for some $\sigma_0 > 0$. Then, $\tau(\cdot)$ satisfies Assumption~\ref{assumption1}. 
    Also, if $u$ is the solution to \eqref{eq:spde_variance_stat}, then
    \begin{enumerate}[(i)]
        \item $\tau \in C^{0, \tilde{\alpha}}(\Gamma)$, where $\tilde{\alpha} = \min\{\alpha - \sfrac{1}{2}, 1\}$, and $u \in C^{0, \gamma}(\Gamma)$ for every $ 0 < \gamma < \tilde{\alpha} $.
        
        \item If $\alpha > \sfrac{3}{2}$, then for every $ e \in \mathcal{E} $, $\tau_e \in C^{1, \alpha - \sfrac{3}{2}}(e)$, and $\tau\in\mathcal{K}(\Gamma)$. Consequently, the solution $u\in\mathcal{K}(\Gamma)$.
    \end{enumerate}
\end{proposition}

Implementing the variance-stationary Whittle--Mat\'ern field requires computing $\sigma_{\kappa}(\cdot)$, which for integer values of $\alpha$ and constant $\kappa$ can be done computationally efficiently using  \citet[Theorem 4]{Bolin2023Statistical}. For non-integer $\alpha$, we can instead compute $\sigma_\kappa(\cdot)$ based on the numerical approximation method developed in the next section. Based on that approximation, $\sigma_\kappa(\cdot)$ can be obtained through a partial inverse of the precision matrix of the field evaluated at the finite element mesh. The partial inverse can be computed efficiently using the Takahashi equations \citep{Takahashi1973Formation}. Thus, these variance-stationary fields could serve as flexible alternatives to the isotropic models of \citet{Anderes2020Isotropic}, as they are well-defined for general metric graphs and facilitate efficient inference for fields with arbitrary smoothness. Moreover, by combining this construction with spatially varying coefficients in the general model, one can obtain fields with arbitrary non-stationary variance, providing direct control over the marginal variance across the graph.

\begin{figure}[t]
\centering
\includegraphics[width=0.66\textwidth]{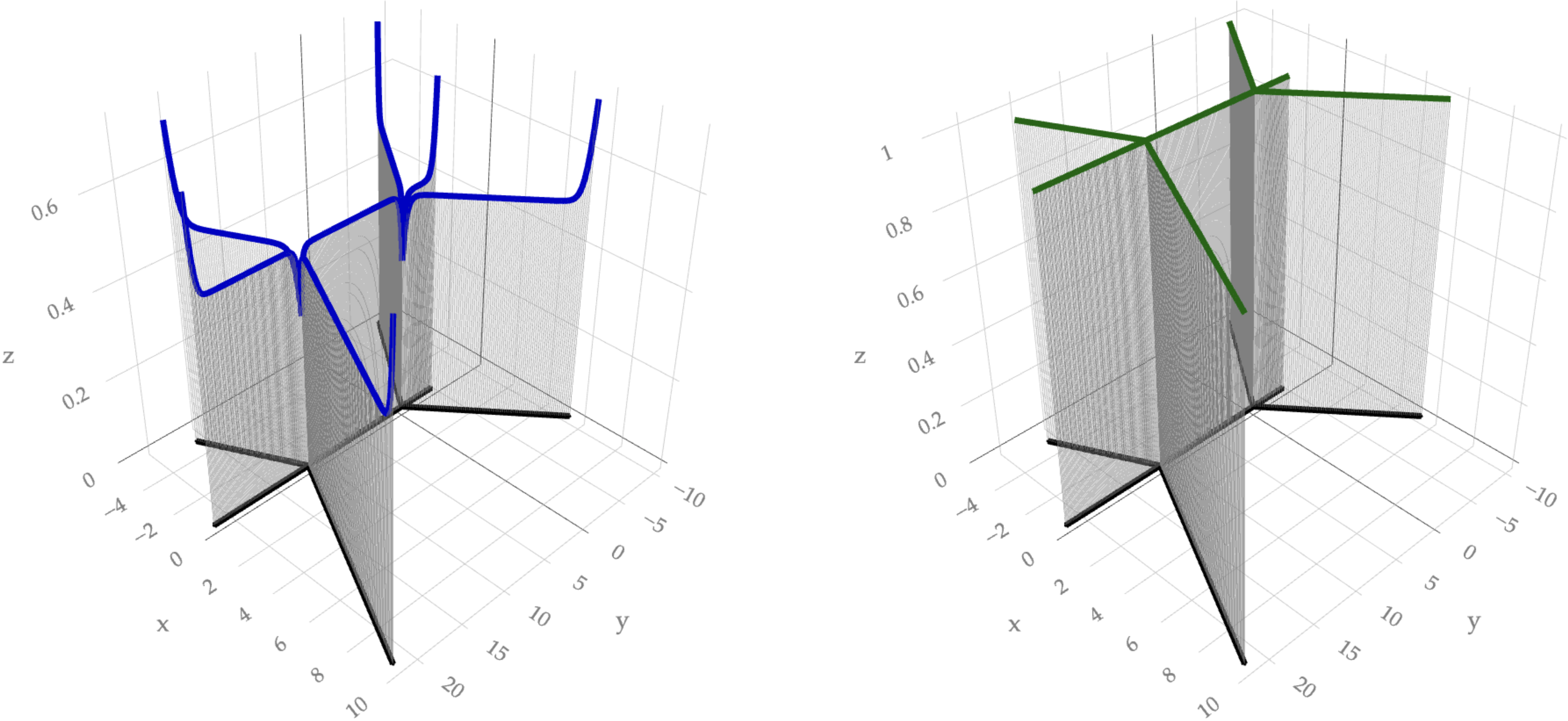}
\caption{Marginal standard deviation $\sigma_\kappa(\cdot)$ of a Whittle--Mat\'ern field with $\tau = 1$ (left), and of a Generalized Whittle--Mat\'ern field with $\tau(\cdot) = \sigma_\kappa(\cdot)$ (right).}
\label{var1}
\end{figure}

\subsection{Non-stationary models with log-regressions for $\kappa$ and $\tau$ }
\label{sec:regression}

\citet{Lindgren2011AnExplicit} proposed modeling $\kappa$ and $\tau$ in the SPDE approach for Euclidean domains using log-regressions. Naturally, we can use the same idea for the metric graph models. Specifically, we use covariates $g_1, \ldots, g_m \in L_2(\Gamma)$, with $m \in \mathbb{N}$, and define 

\begin{equation}
\label{eq:log_regression_tau}
    \log \tau(s) = \theta_{\tau,0} + \sum_{j=1}^m \theta_{\tau,j} g_j(s),\quad 
    \log \kappa(s) = \theta_{\kappa,0} + \sum_{j=1}^m \theta_{\kappa,j} g_j(s),
    \quad s \in \Gamma,
\end{equation}
where $\theta_{\tau,0}, \ldots, \theta_{\tau,m} \in \mathbb{R}$ and $\theta_{\kappa,0}, \ldots, \theta_{\kappa,m} \in \mathbb{R}$ are unknown parameters to be estimated. A natural question is then what assumptions we have to make on the covariates to ensure the regularity requirements of Proposition~\ref{prp:regularity-global}. Because the main requirement there is on $\tau$, the following proposition shows how $\tau$ inherits the regularity of the covariates. In the next proposition, the notation $f' \in C^{0,\gamma}(\Gamma)$ (for $f$ any function, such as $\tau$ or $g_j$) indicates the pathwise definition from Section~\ref{subsec:prelim}.

\begin{proposition}
\label{prp:covariates_regularity}
    Let $\tau$ be defined by \eqref{eq:log_regression_tau}, where $m \in \mathbb{N}$ and $g_1, \ldots, g_m \in C^{0,\gamma}(\Gamma)$ for some $0 < \gamma \leq 1$. Then, $\tau \in C^{0,\gamma}(\Gamma)$. Further, if $g_{1,e}, \ldots, g_{m,e} \in C^{1,\gamma}(e)$ for every $e \in \mathcal{E}$, and $g_1,\ldots, g_m\in\mathcal{K}(\Gamma)$, then $\tau_e \in C^{1,\gamma}(e)$ and $\tau\in\mathcal{K}(\Gamma)$. Finally, if $g_1,\dots,g_m\in C^{2,\gamma}(\Gamma)$ and $g_1',\dots,g'_m \in C^{0,\gamma}(\Gamma)$, then $\tau\in C^{2,\gamma}(\Gamma)$ and $\tau'
    \in C^{0,\gamma}(\Gamma)$.
\end{proposition}

A common scenario is that some covariate of interest, $Z(s)$, is only available at some fixed locations $s_1, \ldots, s_n \in \Gamma$. To obtain a continuously indexed covariate, we can use kriging prediction to interpolate the observed values $z_1,\ldots, z_n$. We can do this by either assuming that the observed values are direct observations of the covariate, $z_i = Z(s_i)$, or by assuming that the they are noisy observations

\begin{equation}
\label{eq:z_1_z_n_data}
z_i|Z(\cdot) \sim N(Z(s_i), \sigma_\epsilon^2), \quad i = 1, \ldots, n, \quad \sigma_\epsilon^2 > 0.
\end{equation}
In either case, to perform the interpolation, we assume that $Z(s) = \beta_0 + w(s)$, where $\beta_0 \in \mathbb{R}$ is an intercept, and $w(\cdot)$ is a centered Gaussian process obtained as the solution to

\begin{equation}
\label{eq:auxiliary_problem_2}
    (\kappa^2 - \Delta_\Gamma)^{\alpha/2} w = \sigma_0 \mathcal{W}, \quad \sigma_0 > 0.
\end{equation}
We then fit this model to the observed values and use  $z(s) = \beta_0 + \mathbb{E}(w(s)|z_1, \ldots, z_n),$ $s\in\Gamma$ as the covariate. The next proposition establishes the regularity of this covariate.

\begin{proposition}
\label{prp:kriging_predictor_regularity}
    Let $\alpha>\sfrac{1}{2}$, and let $w$ be the solution to \eqref{eq:auxiliary_problem_2}. Assume that $z_1, \ldots, z_n$ either are direct observations of $Z(s)$ or follow \eqref{eq:z_1_z_n_data}, where $\beta_0 > 0$, $n \in \mathbb{N}$, $\sigma_\epsilon > 0$, and $s_1, \ldots, s_n \in \Gamma$ are distinct locations. Finally, let $z(s) = \beta_0 + \mathbb{E}(w(s) \mid z_1, \ldots, z_n)$.
    \begin{enumerate}[(i)]
        \item If $\alpha > \sfrac{1}{2}$, for any $0 < \gamma \leq \tilde{\alpha}$, with $\tilde{\alpha} = \min\{\alpha - \sfrac{1}{2}, 1\}$, we have $z \in C^{0, \gamma}(\Gamma)$.
        
        \item If $\alpha > \sfrac{3}{2}$, for any $0 < \gamma \leq \alpha - \sfrac{1}{2}$, we have $z_{e} \in C^{1, \gamma}(e)$ for every $e \in \mathcal{E}$, and $z\in\mathcal{K}(\Gamma)$.
    \end{enumerate}
\end{proposition}

By combining Propositions~\ref{prp:covariates_regularity} and \ref{prp:kriging_predictor_regularity}, we conclude that a log-regression model for $\tau$ using a kriging predictor as a covariate ensures the regularity conditions required in Proposition~\ref{prp:regularity-global}.

\section{The numerical approximation method}
\label{num_approx_sec}

Having defined the models and illustrated their flexibility, it remains to obtain a method for using them in statistical inference. In this section, we outline the main idea of this method, show the main theoretical results, and finally illustrate how the approximation is used for computationally efficient inference.

\subsection{Main idea}

\cite{Bolin2023Statistical} showed that one can evaluate finite-dimensional distributions exactly and computationally efficiently of solutions to \eqref{eq:spde} whenever $\kappa,\tau>0$ are constants and $\alpha\in\mathbb{N}$. Their approach relies heavily on the fact that the corresponding process has Markov properties \citep{Bolin2026Markov} if $\alpha\in\mathbb{N}$. As we want a method that works for general $\alpha>\sfrac{1}{2}$, there is little hope in extending their method to the non-stationary fields that we are considering. \citet{Bolin2024Regularity} proposed approximating Generalized Whittle--Mat\'ern fields using a FEM approximation  combined with a quadrature approximation of the fractional power of the operator. That is essentially the metric graph version of the rational SPDE approach proposed in \cite{Bolin2020Rational} for Generalized Whittle--Mat\'ern fields on Euclidean domains. The disadvantage of this approach is that the corresponding approximation is not a Gaussian Markov random field (GMRF), which means that it cannot be implemented in software such as \texttt{R-INLA} \citep{Xiong2024Covariance}. This limits applicability, as the vast majority of applications of the SPDE approach have been done using \texttt{R-INLA}, because of its flexibility and computational efficiency. For Euclidean domains, \cite{Xiong2024Covariance} solved this issue by proposing a covariance-based rational SPDE approach, where a rational approximation of the covariance operator of the process was performed to obtain a GMRF approximation. The method we propose for the Gaussian processes from Section~\ref{sec:model} is the metric graph version of this approach. Specifically, we combine a FEM approximation with a rational approximation of the corresponding covariance operator.

\subsection{Finite element approximation}
\label{fem_approx}

To construct the finite element approximation, we introduce a space of piecewise linear and continuous functions, $V_h$, which is spanned by a set of hat functions $\{\psi_j\}_{j=1}^{N_h}$ on the graph. This set is obtained by introducing internal vertices on each edge, and for each vertex we define a basis function that is one at the vertex and decreases linearly to zero at the closest other vertices. We use $h$ to denote the largest distance between any two neighboring vertices, which is a parameter that will control the accuracy of the approximation. An example of the construction can be seen in Figure~\ref{app:basisfunctions} in Appendix \ref{app:fem_and_conv_cov_f}, where we illustrate that the hat functions can be divided into two types, one where the vertex is of degree 2, and one where the vertex is of higher degree. Further, note that for any function $\phi \in V_h$ and every edge $e \in \mathcal{E}$, the function $\phi_e$ is differentiable almost everywhere on $e$. The points where $\phi_e$ is not differentiable form a finite set, corresponding to the ``tips'' of the hat basis functions. More details on these functions are provided in Appendix~\ref{app:fem_and_conv_cov_f}.

We define the discrete version of $L$ as $L_h:V_h\longrightarrow V_h$ via $(L_h \phi,\psi)_{L_2(\Gamma)} = \mathfrak{h}(\phi,\psi)$ for $\phi,\psi\in V_h$. Here $\mathfrak{h}$ is the bilinear form corresponding to the operator $L$, given by 

\begin{equation*}
    \mathfrak{h}(f,g) = (\kappa^2 f,g)_{L_2(\Gamma)}+\sum_{e\in\mathcal{E}}\int_{e}f_e'(s)g_e'(s)ds,
\end{equation*}
where the derivatives are the almost everywhere derivatives of $f_e$ and $g_e$. The corresponding discrete problem for \eqref{eq:spde}, consists of finding $u_h \in V_h$ such that  

\begin{equation}
\label{discretemodel}
    L_h^{\sfrac{\alpha}{2}} (\tau u_h) = \mathcal{W}_h,
\end{equation}
where $\mathcal{W}_h$ represents Gaussian white noise on $V_h$. Specifically, $\mathcal{W}_h$ can be expressed as a collection of centered Gaussian random variables $\{\mathcal{W}_h(\psi_j)\}_{j=1}^{N_h}$, where $\psi_j$ is the $j$th hat basis function of $V_h$. These variables satisfy $\mathbb{E}[\mathcal{W}_h(\psi_i)\mathcal{W}_h(\psi_j)] = (\psi_i, \psi_j)_{L_2(\Gamma)}$,  for $i, j = 1, \dots, N_h$. The fractional power of $L_h$ is defined in the spectral sense (see Appendix~\ref{app:fem_and_conv_cov_f} for details), and the solution of \eqref{discretemodel} can be represented as

\begin{equation}
\label{eq:basisexp}
    u_h(s) =  \sum_{j=1}^{N_h} u_j\psi_j(s),
\end{equation}

To illustrate how this approximation works, let $\alpha = 2$ and $\tau \equiv 1$ in \eqref{discretemodel}. In this case, the stochastic coefficients $\mathbf{u} =  [u_1,\dots,u_{N_h}]^\top$ in \eqref{eq:basisexp} can be determined by solving the linear system $\mathfrak{h}(u_h, \psi_j) = \mathcal{W}_h(\psi_j)$, $j = 1,\dots, N_h$. In matrix form, this system can be written as $(\mathbf{G}+\mathbf{C}^\kappa)\mathbf{u} = \mathbf{W}$, where $\mathbf{G}$ is the stiffness matrix with entries $\mathbf{G}_{ij} = (\psi'_i,\psi'_j)_{L_2(\Gamma)}$,  $\mathbf{C}^\kappa$ has entries $\mathbf{C}^\kappa_{ij} = (\kappa^2\psi_i,\psi_j)_{L_2(\Gamma)}$, and $\mathbf{W}=[\mathcal{W}_h(\psi_1),\dots,\mathcal{W}_h(\psi_{N_h})]^\top\sim N(\mathbf{0},\mathbf{C})$, with $\mathbf{C}$ having entries $\mathbf{C}_{ij} = (\psi_i,\psi_j)_{L_2(\Gamma)}$. It follows that $\mathbf{u}\sim N(\mathbf{0}, \mathbf{Q}_2^{-1})$, where $\mathbf{Q}_2=\mathbf{L}^\top\mathbf{C}^{-1}\mathbf{L}$ with $\mathbf{L}= \mathbf{G}+\mathbf{C}^\kappa$. In the case that $\tau(\cdot)$ is a spatially varying function and a piecewise constant approximation of $\kappa(\cdot)$ is considered, the precision matrix of $\mathbf{u}$ becomes $\mathbf{Q}_2=\boldsymbol{\tau}\mathbf{L}^\top\mathbf{C}^{-1}\mathbf{L}\boldsymbol{\tau}$, where $\mathbf{L}= \mathbf{G}+\boldsymbol{\kappa}^2\mathbf{C}$ with $\boldsymbol{\tau} = \text{diag}(\tau(s_1), \dots, \tau(s_{N_h}))$ and $\boldsymbol{\kappa}^2 = \text{diag}(\kappa^2(s_1), \dots, \kappa^2(s_{N_h}))$. To obtain a sparse precision matrix, one can replace the mass matrix $\mathbf{C}$ by a lumped mass matrix $\Tilde{\mathbf{C}}$, which is a diagonal matrix with entries $\Tilde{\mathbf{C}}_{ii} = \sum_{j=1}^{N_h}\mathbf{C}_{ij}$.

\subsection{Handling the fractional power } 
\label{handling_fract_power}

To extend the approach to a general value of $\alpha > \sfrac{1}{2}$, we extend the covariance-based rational approximation of \citet{Xiong2024Covariance} to our metric graph setting. The solution $u_h$ to the discrete problem has covariance function $\varrho_h^\alpha(\cdot, \cdot)$, which is the kernel of the covariance operator $M_{\tau^{-1}}L_h^{-\alpha}M_{\tau^{-1}}$, where $M_{\tau^{-1}}$ denotes the operator that multiplies a function with $\tau^{-1}$. We approximate this by performing a rational approximation of $L_h^{-\alpha}$. Specifically, for the real-valued function $\hat{f}(x) = x^{\{\alpha\}}$, where $\{x\} = x - \lfloor x \rfloor \in [0,1)$ denotes the fractional part of $x$, consider the rational approximation of order $m$, $\hat{f}(x) \approx \hat{r}_m(x) = p(x)/q(x)$, defined on the interval $[0, \kappa_0^{-2}]$, which contains the spectrum of $L_h^{-1}$. Here, $\kappa_0$ is as specified in Assumption~\ref{assumption1}, $p(x) = \sum_{i=0}^m a_i x^i$, and $q(x) = \sum_{i=0}^m b_i x^i$, with the coefficients $\{a_i\}$ and $\{b_i\}$ determined as the best rational approximation in the $L_\infty$-norm. These coefficients can be efficiently computed using the BRASIL algorithm \citep{Hofreither2021AnAlgorithm}, see Appendix~\ref{app:fem_and_conv_cov_f} for details. Defining $r_m(x):=x^{\lfloor \alpha \rfloor}\hat{r}_m(x)$, we then approximate 

\begin{align*}
\label{covope2}
L_h^{-\alpha} &=  L_h^{-{\lfloor \alpha \rfloor}}L_h^{-\{\alpha\}}\approx L_h^{-{\lfloor \alpha \rfloor}} p(L_h^{-1})q(L_h^{-1})^{-1}  
= L_h^{-{\lfloor \alpha \rfloor}}\Bigl(\sum_{i=1}^mr_i(L_h-p_iI_{h})^{-1} + kI_{h} \Bigr),
\end{align*}
where the coefficients $r_i,k>0$ and $p_i<0$ are obtained by performing a partial fraction decomposition of $\hat{r}$ and $I_{h}$ is the identity operator on $V_h$, see \citet{Bolin2025Linear} for further details. By \citet[Appendix C]{Xiong2024Covariance}, the covariance matrix of the weights $\mathbf{u} =  [u_1,\dots,u_{N_h}]^\top$ in \eqref{eq:basisexp} can be written as 

\begin{equation}
\label{sigmau}
    \mathbf{\Sigma}_{\mathbf{u}} = \boldsymbol{\tau}^{-1}(\mathbf{L}^{-1}\mathbf{C})^{\lfloor \alpha \rfloor}\sum_{i=1}^mr_i(\mathbf{L}-p_i\mathbf{C})^{-1}\boldsymbol{\tau}^{-1} + \boldsymbol{\tau}^{-1}\mathbf{K}_{\lfloor \alpha \rfloor}\boldsymbol{\tau}^{-1},
\end{equation}
where $\mathbf{C}$ and $\mathbf{L}$ are as before, and $\mathbf{K}_n = k(\mathbf{L}^{-1}\mathbf{C})^{n-1}\mathbf{L}^{-1}$ for $n\in\mathbb{N}$ with $\mathbf{K}_0 = k\mathbf{C}$. Because $r_i,k>0$ and $p_i<0$, each term in \eqref{sigmau} is a valid covariance matrix. We can therefore express $\mathbf{u}$ as a sum of independent random vectors $\mathbf{x}_i\sim N(\mathbf{0}, \mathbf{Q}_i^{-1})$, $\mathbf{u} = \sum_{i=1}^{m+1} \mathbf{x}_i$, where 
\begin{equation*}
    \mathbf{Q}_i = \begin{cases}
        r_i^{-1}\boldsymbol{\tau}(\mathbf{L}-p_i\mathbf{C})(\mathbf{C}^{-1}\mathbf{L})^{\lfloor \alpha \rfloor}\boldsymbol{\tau}&\text{ if } i= 1,\dots,m,\\
        \boldsymbol{\tau}\mathbf{K}^{-1}_{\lfloor \alpha \rfloor}\boldsymbol{\tau}&\text{ if } i = m+1.
    \end{cases}
\end{equation*}
To ensure that each $\mathbf{Q}_i$ is sparse, we replace matrix $\mathbf{C}$ by matrix $\Tilde{\mathbf{C}}$, as before.

\subsection{Convergence rates}

A natural question is now how good the rational approximation is. To answer this, note that $M_{\tau^{-1}}r_m(L_h^{-1})M_{\tau^{-1}}$ is the covariance operator of the rational approximation, which is associated to a covariance function
$\varrho_{h,m}^\alpha(\cdot,\cdot)$ defined on $\Gamma\times\Gamma$ through the following relation:
$$
(M_{\tau^{-1}}r_m(L_h^{-1})M_{\tau^{-1}} f)(s) = \int_\Gamma \varrho_{h,m}^\alpha(s,s') f(s') ds'.
$$
The following result shows that this covariance function approximates the true covariance at a fixed rate depending on the mesh width $h$ and on the order of the rational approximation $m$. Here and for future reference, convergence in $L_2(\Gamma\times \Gamma)$ means convergence in the norm 
$$
\|f\|^2_{L_2(\Gamma\times \Gamma)}=\displaystyle\int_\Gamma\int_\Gamma f^2(s,r)dsdr = \sum_{e\in\mathcal{E}}\int_{e}\int_e f_e^2(s,r)dsdr.
$$
\begin{proposition}
\label{conv_cov_func}
Under Assumption \ref{assumption1}, if $\alpha>\sfrac{1}{2}$ then for any $\varepsilon > 0$
\begin{equation}
\label{mainrate}
\|\varrho^\alpha - \varrho_{h,m}^\alpha\|_{L_2(\Gamma\times \Gamma)}\leq C\left( h^{\min\{2\alpha-\sfrac{1}{2},2\}-\varepsilon} + 1_{\alpha\not\in\mathbb{N}}\cdot h^{-\sfrac{1}{2}}e^{-2\pi\sqrt{\{\alpha\}m}}\right),
\end{equation}
where $C>0$ is a constant independent of $h$ and $m$.
\end{proposition}

 We have two terms in the bound for the error, the first is due to the FEM approximation, and the second is due to the rational approximation. If we choose 
\begin{equation}\label{eq:calibration}
    m = c \lceil(\min\{2\alpha-\sfrac{1}{2},2\}+\sfrac{1}{2})^2\log^2(h)/(4\pi^2\{\alpha\})\rceil,
\end{equation}
we balance these two terms and the total convergence rate as a function of the mesh size is $\min\{2\alpha-\sfrac{1}{2},2\}$ if $\alpha>\sfrac{1}{2}$. As the error induced by the rational approximation decreases exponentially fast in $m$, we rarely have to use a large value of $m$, which is beneficial in practice as the computational cost increases as we increase $m$.

\subsection{Using the approximation for inference}
\label{using_the_approx_for_inference}
Suppose we have observations $\mathbf{y} =  [y_1,\dots,y_{n}]^\top$, obtained as    $y_i|u(\cdot)\sim N(u(s_i),\sigma_\epsilon^2)$, 
for $s_i\in\Gamma$, $i=1,\dots,n$, where $u(\cdot)$ is a solution to \eqref{eq:spde} on $\Gamma$. As previously discussed, the solution $u(\cdot)$ of \eqref{eq:spde} can be approximated as $u_h(s) =  \sum_{j=1}^{N_h} u_j\psi_j(s)$, where $\psi_j$ are FEM basis functions and the vector of stochastic weights $\mathbf{u} =  [u_1,\dots,u_{N_h}]^\top$ can be expressed as a sum $\mathbf{u} = \sum_{i=1}^{m+1} \mathbf{x}_i$ of independent GMRFs $\mathbf{x}_i$, each with sparse precision matrix $\mathbf{Q}_i$. With this formulation, the aforementioned model can be written hierarchically as $\mathbf{X}\sim N(\mathbf{0}, \mathbf{Q}^{-1})$ and $\mathbf{y}|\mathbf{X}\sim N(\overline{\mathbf{A}}\mathbf{X}, \mathbf{Q}_\epsilon^{-1})$, where $\mathbf{Q}_\epsilon =  \sigma_\epsilon^{-2}\mathbf{I}$, $\mathbf{Q}=\operatorname{diag}\left(\mathbf{Q}_1, \ldots, \mathbf{Q}_{m+1}\right)$, $\mathbf{X} = [\mathbf{x}_1^\top,\dots,\mathbf{x}_{m+1}^\top]^\top$, and $\overline{\mathbf{A}} = [\mathbf{A}\;\cdots\;\mathbf{A}]$ with $\mathbf{A}_{ij} = \psi_j(s_i)$. Because of the sparsity of $\mathbf{Q}$, statistical inference and spatial prediction can be performed computationally efficient. This computational advantage is well documented for models defined on Euclidean domains (see, e.g., \cite{Xiong2024Covariance}), and the same principles apply in the present metric graph setting. In particular, the posterior distribution of  $\mathbf{X}$ satisfies $\mathbf{X}|\mathbf{y}\sim N(\boldsymbol{\mu}_{\mathbf{X} \mid \mathbf{y}}, \mathbf{Q}_{\mathbf{X} \mid \mathbf{y}}^{-1})$, which allows the kriging predictor to be computed efficiently as $\boldsymbol{\mu}_{\mathbf{X} \mid \mathbf{y}}=\mathbf{Q}_{\mathbf{X} \mid \mathbf{y}}^{-1} \bar{\mathbf{A}}^{\top} \mathbf{Q}_\epsilon\mathbf{y}$, where $\mathbf{Q}_{\mathbf{X} \mid \mathbf{y}}=\overline{\mathbf{A}}^{\top} \mathbf{Q}_\epsilon\overline{\mathbf{A}}+\mathbf{Q}$. Additionally, the marginal log-likelihood $\ell(\mathbf{y})$ of $\mathbf{y}$ is given by
\begin{align*}
    2\ell(\mathbf{y}) = &\log |\mathbf{Q}| + \log |\mathbf{Q}_\epsilon|-\log |\mathbf{Q}_{\mathbf{X} \mid \mathbf{y}}| - \boldsymbol{\mu}_{\mathbf{X} \mid \mathbf{y}}^\top \mathbf{Q}\boldsymbol{\mu}_{\mathbf{X} \mid \mathbf{y}}\\& - (\mathbf{y} - \overline{\mathbf{A}}\boldsymbol{\mu}_{\mathbf{X} \mid \mathbf{y}})\top \mathbf{Q}_\epsilon (\mathbf{y} - \overline{\mathbf{A}}\boldsymbol{\mu}_{\mathbf{X} \mid \mathbf{y}})-n\log(2\pi),
\end{align*}
which can be evaluated efficiently due to the sparsity of $\mathbf{Q}$ and $\mathbf{Q}_{\mathbf{X}\mid\mathbf{y}}$. One can fit this model and more complex models by combining the \texttt{MetricGraph} package with either the \texttt{R-INLA} or 
\texttt{inlabru} \citep{inlabruRpackage} interfaces of the \texttt{rSPDE} package \citep{rSPDEpackage}, as shown in the vignettes by \citet{RINLAvignette, inlabruvignette}.

\begin{figure}[t]
    \centering
    \includegraphics[width=0.99\linewidth]{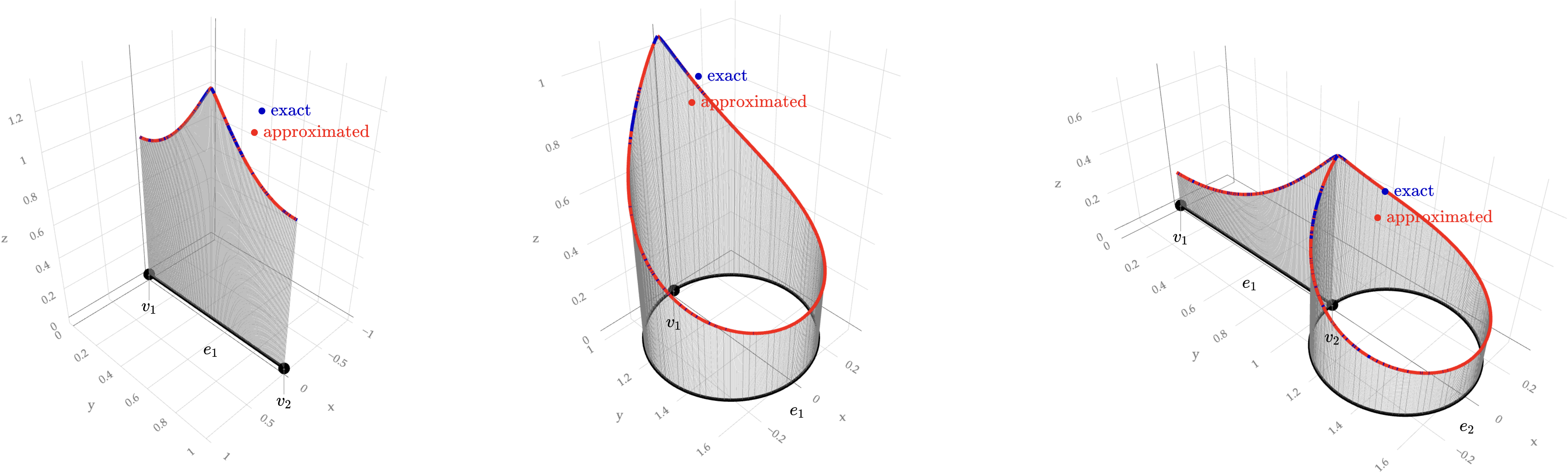}
\caption{For each graph, exact (blue) and approximate (red) covariances $\varrho^{\alpha}(s_0,s)$ are shown for a fixed point $s_0$. In all cases, $h=0.003$, $\alpha = 1.1$, $m = 4$, and $\rho=1$.}
\label{Interval.Circle.Tadpole}
\end{figure}

\section{Numerical Experiments}\label{numerical_experiments}
To demonstrate the performance of the proposed approximation, we consider the interval, circle, and tadpole graphs shown in Figure \ref{Interval.Circle.Tadpole}. These graphs are chosen because the covariance function of the solution $u$ is known in each case. See Appendix~\ref{app:cov_funct_def} for details. We consider a FEM mesh with step-size of $h=\sfrac{1}{1000}$ and for the orders of the rational approximation, we consider $m = 1,2,3,4,5$. For the smoothness parameter $\alpha$, we choose values ranging from 0.6 to 2.0.  We further choose $\tau$ such that the marginal variance is close to 1 and set $\kappa$ such that the practical correlation range $\rho$ is fixed to $0.1, 0.5,1$, or $2$. The approximation errors for the tadpole graph and all parameter combinations are shown in Figure \ref{tadpole_errors}. Corresponding illustrations for the interval and circle graphs are shown in Appendix~\ref{app:approx_error}.

\begin{figure}[t]
\centering
\includegraphics[width=0.99\textwidth]{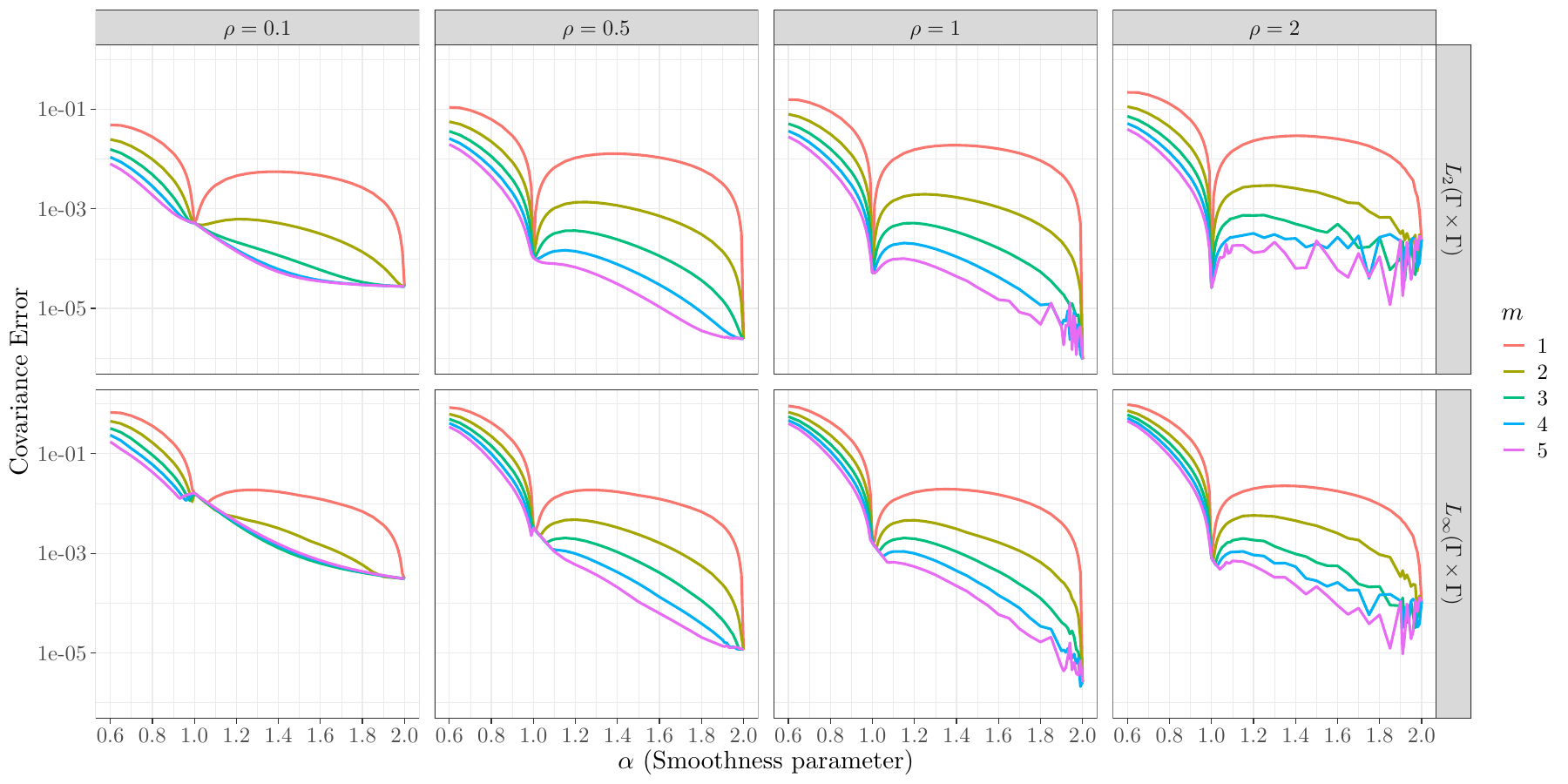}
\caption{Covariance errors in the $L_2(\Gamma\times\Gamma)$-norm (top) and supremum  norm (bottom) on the tadpole graph for different $\rho$, $\alpha$, and rational approximation orders $m$.}
\label{tadpole_errors}
\end{figure}

We observe that the error decreases rapidly in $m$ and that the errors for different $m$ collapse to a common point when $\alpha = 1$ and $\alpha=2$, since there is then no rational approximation and the only source of error is the FEM approximation. 
As $\rho$ increases, some numerical instabilities appear for large values of $m$, especially when $\alpha>2$. Similar trends are observed for the interval and circle cases.

To verify that we indeed obtain the convergence rate $\min\{2\alpha-\sfrac{1}{2},2\}$ as a function of $h$ if we calibrate the order of the rational approximation according to \eqref{eq:calibration}, we consider $\rho = 0.5$, $\sigma = 1$, $\alpha = \sfrac{n}{8}$ for $n =  6, 7, 8, 9, 12$, compute the error for $h = 2^{-\ell}$ for $ \ell = 4.5, 4.75, 5, 5.25, 5.5$. The empirical convergence rate is obtained by determining the slope $r$ in the regression $\log \text{error} = c + r \log h$, where $\text{error} \approx \|\varrho - \varrho_{h,m}^\alpha\|_{L_2(\Gamma\times \Gamma)}$ is approximated as 
    $\text{error}^2 = \mathbf{w}_{\text{ok}}^\top(\mathbf{\Sigma}_{\text{ok}}-\mathbf{A}\mathbf{\Sigma}_{h}\mathbf{A}^\top)^2\mathbf{w}_{\text{ok}}$.
Here, $\mathbf{\Sigma}_{h}$ is the covariance matrix of the approximation and  $\mathbf{\Sigma}_{\text{ok}}$ is the exact covariance matrix on a mesh with $h = h_{\text{ok}}$. Further, $\mathbf{A}$ is a matrix with entries $\psi_j^{h}(s_i)$, $\psi_j^{h}$ is the $j$th FEM basis function associated to the coarse mesh, $s_i$ is the $i$th node location in the fine mesh, and $\mathbf{w}_{\text{ok}}$ is a vector of weights $w_i = (\psi^{h_{\text{ok}}}_i,1)_{L_2(\Gamma)}$. The results, in Figure~\ref{Interval.Circle.Tadpole_rates} and Table \ref{Rate_of_convergence}, confirm that the empirical rates are close to the theoretical ones.

\begin{figure}[t]
\centering
\includegraphics[width=0.99\textwidth]{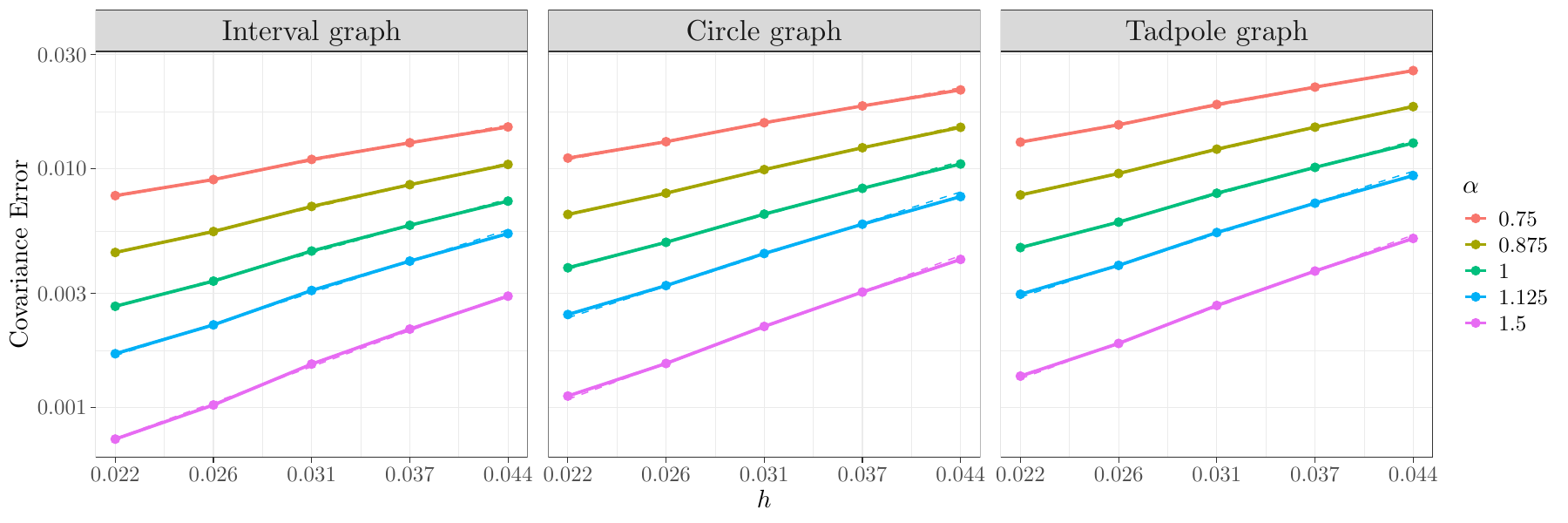}
\caption{Observed covariance error for different values of $\alpha$ as functions of the mesh size $h$.
The dashed lines show the theoretical rates for each case.}
\label{Interval.Circle.Tadpole_rates}
\end{figure}

\begin{table}[t]
    \caption{Observed rates of convergence for the covariance errors in Figure \ref{Interval.Circle.Tadpole_rates}.}
    \label{Rate_of_convergence}
    \centering
    \begin{tabular}{lccccc}
    \toprule
        $\alpha$ & $0.75$ & $0.875$ & $1$ & $1.125$ & $1.5$  \\ \hline   
        Theoretical rates &   $1$   &   $1.25$    & $1.5$ & $1.75$  & $2$ \\ 
        Interval graph &   $0.97$   &   $1.24$    & $1.48$ & $1.69$  & $2.01$ \\ 
        Circle graph &   $0.96$   &   $1.22$    & $1.46$ & $1.65$  & $1.92$ \\ 
        Tadpole graph &   $1.01$   &   $1.24$    & $1.47$ & $1.67$  & $1.93$ \\ 
        \bottomrule
    \end{tabular}
\end{table}

\begin{figure}[t]
    \centering
    \includegraphics[width=0.99\textwidth]{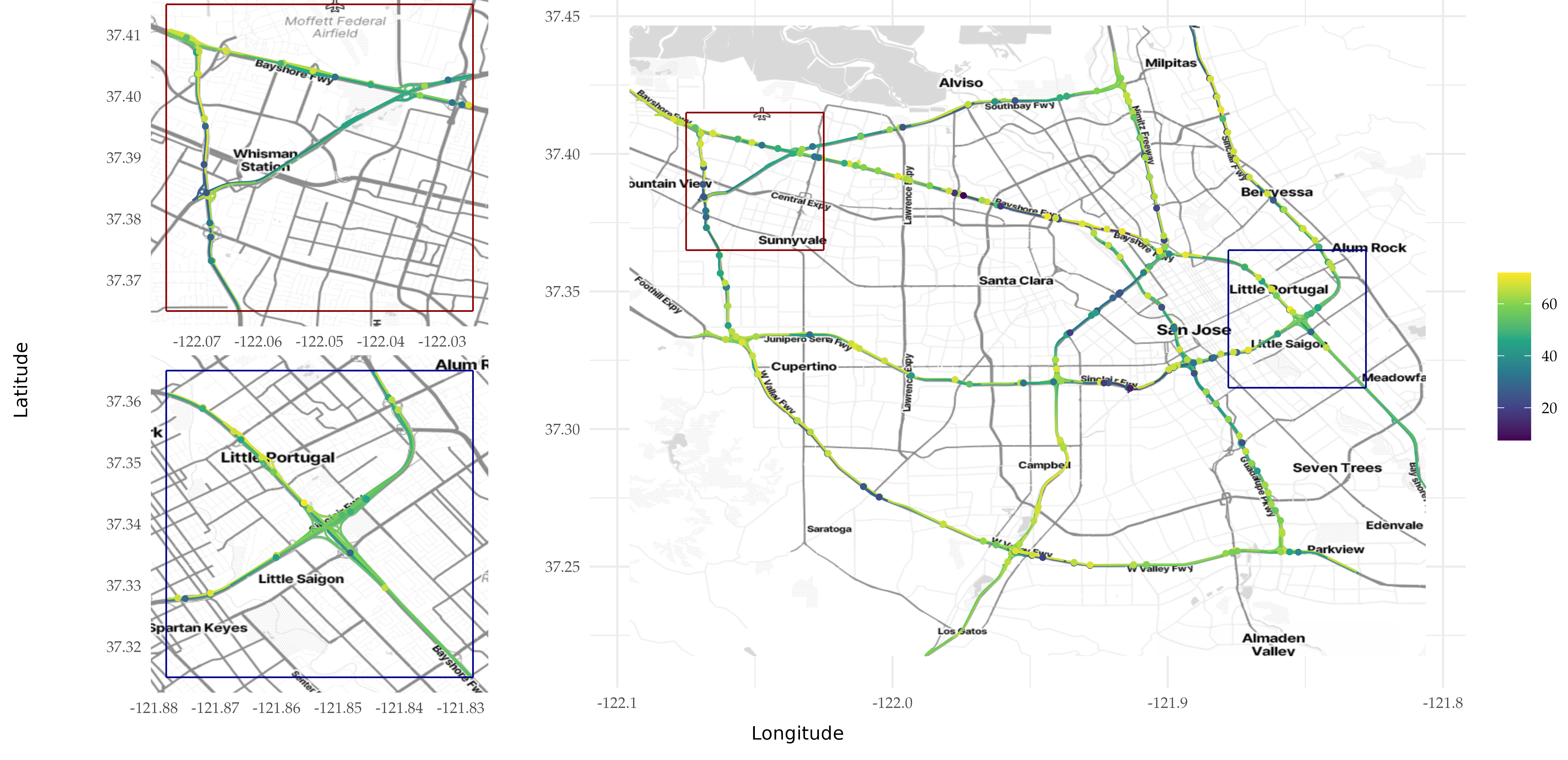}
    \caption{Speed observations (in mph) on the highway network of the city of San Jose, recorded on April 3, 2017. Points represent the observations, and the colored edges represent the model prediction. The left panels show two zoomed-in areas.}
    \label{replicate14}
\end{figure}

\section{Application}\label{application}

\subsection{The data}
To illustrate the capabilities of the proposed method, we consider traffic speed data from the highway system in San Jose, California, obtained from the California Performance Measurement System database \citep{Chen2001Freeway}. The data set consists of 26 replicates, recorded by automatic sensors weekly on Mondays at 17:30, from January 2, 2017, to June 26, 2017. Each replicate contains $n = 325$ speed records (in mph), observed at different locations which are fixed across the replicates. Figure \ref{replicate14} shows speed records for the 14th replicate. The highway network was obtained from OpenStreetMap \citep{Contributors2017Planet} and, when converted to a metric graph object, consists of 691 vertices and 848 edges.

\cite{Borovitskiy2021Matern} averaged this data across replicates and used graph Laplacian-based model restricted to the vertices. \cite{Bolin2023Statistical} used a Whittle--Mat\'ern field to model the same data. Here, we do not average the data across replicates.
The first 13 replicates are used to construct covariates, while the remaining 13 are used for modeling. Eleven locations were removed beforehand, as their recorded values were constant across the first 13 replicates. Figures \ref{app:mean_better_3} and \ref{stand_log_sigma} in Appendix \ref{app:covariates_for_applic} display the mean and standardized logarithm of the standard deviation for these replicates, showing that the data is not stationary.

\subsection{Constructing covariates with sufficient regularity}
\label{ssec:const_covariates}
Identifying meaningful covariates for this data is challenging. One natural candidate is the speed limit, but this is constant for these roads.
To derive a covariate for expected speed, we use an approach similar to that in Section~\ref{sec:regression}. Specifically, we calculate the average speed $\mathrm{m}_i$ at each observation location $s_i$, $i = 1, \ldots, 314$, using the first 13 replicates. We then fit the model $\mathrm{m}_i \mid u(\cdot) \sim N(\beta_0 + u(s_i), \sigma_\epsilon^2)$, where $u$ is a Whittle--Mat\'ern field with $\alpha$ estimated from the data. The covariate $\text{mean.cov}(s)$ is defined as $\beta_0+\mathbb{E}(u(s) \mid \mathrm{m}_1, \ldots, \mathrm{m}_{314})$ (see Figure \ref{app:mean_better_3} in Appendix \ref{app:covariates_for_applic}).

Similarly, we construct a covariate for $\tau$ and $\kappa$ in the non-stationary Whittle--Mat\'ern field by fitting the model 
$\log\mathrm{std}_i\sim N(\beta_0 + u(s_i), \sigma_\epsilon^2)$ to the standard deviations $\mathrm{std}_i$ of the data at the observation locations in the first 13 replicates. The covariate 
$\text{std.cov}(s)$ is then defined as the kriging predictor $\beta_0+\mathbb{E}(u(s)| \mathrm{log\;std}_1,\ldots, \mathrm{log\;std}_{314})$ standardized by subtracting its mean and dividing by its standard deviation (see Figure \ref{stand_log_sigma} in Appendix \ref{app:covariates_for_applic}).
As showed earlier, the regularity of $\tau$ is crucial for the regularity of the Gaussian field, and the advantage of this method for constructing the covariate is that if ${\log\tau(s) = \theta_1 + \theta_3 \cdot \text{std.cov}(s)}$, then, by Propositions \ref{prp:covariates_regularity} and \ref{prp:kriging_predictor_regularity}, $\tau$ satisfies the regularity conditions of Proposition~\ref{prp:regularity-global}.

\subsection{The models for the data}
\label{models_for_data}
We now model the second part of the data by assuming that the speed records $y_i$ are observations of 13 independent replicates satisfying
\begin{equation}
\label{applimodel}
    y_i|u(\cdot)\sim N(\beta_0 + \beta_1\text{mean.cov}(s_i) + u(s_i),\sigma_\epsilon^2),\;i = 1,\dots, 314,
\end{equation} 
where $u(\cdot)$ is a Gaussian process on the highway network. We consider Whittle--Mat\'ern fields with stationary parameters $\kappa,\tau>0$ and Generalized Whittle--Mat\'ern fields where $\tau$ and $\kappa$ are non-stationary following log-regressions
\begin{equation}
\label{logregressions}
    \log\tau(s) = \theta_1 + \theta_3\cdot \text{std.cov}(s)\quad\hbox{and}\quad \log\kappa(s) = \theta_2 + \theta_4 \cdot\text{std.cov}(s).
\end{equation}
For each of the two classes of models, we consider three scenarios for the smoothness parameter: when (1) $\alpha$ is fixed to 1 or (2) 2, and (3) $\alpha$ is estimated from the data.

The six models are fitted to the data using \texttt{MetricGraph} and \texttt{inlabru} (the code is available at \url{https://github.com/leninrafaelrierasegura/GWMFMG}). Posterior means of the parameter estimates can be seen in Tables \ref{table_stat} and \ref{table_nonstat}. All estimates in both tables are significantly different from zero, except for the one marked with~$^*$. From Table \ref{table_stat}, we observe that the estimates for $\beta_0$, $\beta_1$, and $\sigma_e$ are relatively stable regardless of the smoothness of the field, suggesting that the effect of the covariate and the residual variability of the model remain robust across different levels of spatial smoothness. The estimates for $\tau$ and $\kappa$ are quite similar for the cases where $\alpha=1$ and $\alpha$ is estimated. However, they differ significantly for the case $\alpha=2$, as the model needs to adjust to a smoother field,  requiring a larger $\kappa$ to reduce the spatial dependence range, allowing the correlation to decay more quickly over shorter distances. For the non-stationary models (Table \ref{table_nonstat}), we also observe some consistency in the estimates of $\beta_0$, $\beta_1$, and $\sigma_e$. However, for the parameters $\theta_i$, $i=1,2,3,4$ in \eqref{logregressions}, the estimates vary considerably between the three cases, suggesting that the parameters in \eqref{logregressions} are sensitive to changes in $\alpha$. 

\begin{table}[t]
    \caption{Posterior means of the parameters of stationary models.}
    \label{table_stat}
    \centering
\begin{tabular}{lccccc}
\toprule
$\alpha$              & $\beta_0$ & $\beta_1$ & $\tau$ & $\kappa$ & $\sigma_e$ \\ \hline 
$1$                    & $-14.463$ & $1.276$ & $0.224$ & $0.052$ & $8.573$ \\ 
$2$                    & $-16.096$ & $1.308$ & $0.333$ & $0.255$ & $8.782$ \\ 
$\texttt{est} = 1.003$ & $-14.452$ & $1.276$ & $0.224$ & $0.051$ & $8.559$ \\ 
\bottomrule
\end{tabular}
\end{table}

\begin{table}[t]
    \caption{Posterior means of the parameters of non-stationary models.}
    \label{table_nonstat}
    \centering
\begin{tabular}{lccccccc}
\toprule
$\alpha$ & $\beta_0$ & $\beta_1$ & $\theta_1$ & $\theta_2$ & $\theta_3$ & $\theta_4$ & $\sigma_e$ \\ \hline 
$1$                    & $-4.528$ & $1.126$ & $-1.184$ & $-6.287$ & $-0.614$ & $3.001$ & $8.246$ \\ 
$2$                    & $-4.456$ & $1.124$ & $-0.337$ & $-2.015$ & $-0.921$ & $0.689$ & $8.541$ \\ 
$\texttt{est} = 1.940$ & $-4.176$ & $1.122$ & $\;\;\;\;0.136^*$ & $-2.423$ & $-0.734$ & $0.664$ & $8.263$ \\ 
\bottomrule
\end{tabular}
\end{table}

\subsection{Comparison using cross-validation}
\label{sec:comparison_crossval}
We compare the predictive performance of the six models using leave-group-out pseudo cross-validation \citep{Liu2025Leave}, and specifically the strategy from \cite{Xiong2024Covariance}, where predictions at each location are made by excluding the most informative observations—those within a specified radius $R$ based on geodesic distance. By varying $R$, the model’s prediction accuracy is obtained as a function of $R$. We use the Mean Squared Errors (MSE) and the Negative Log-Score (NLS) \citep{Good1952Rational} as evaluation metrics, where a lower score is better for each. 

\begin{figure}[t]
    \centering
    \includegraphics[width=0.99\textwidth]{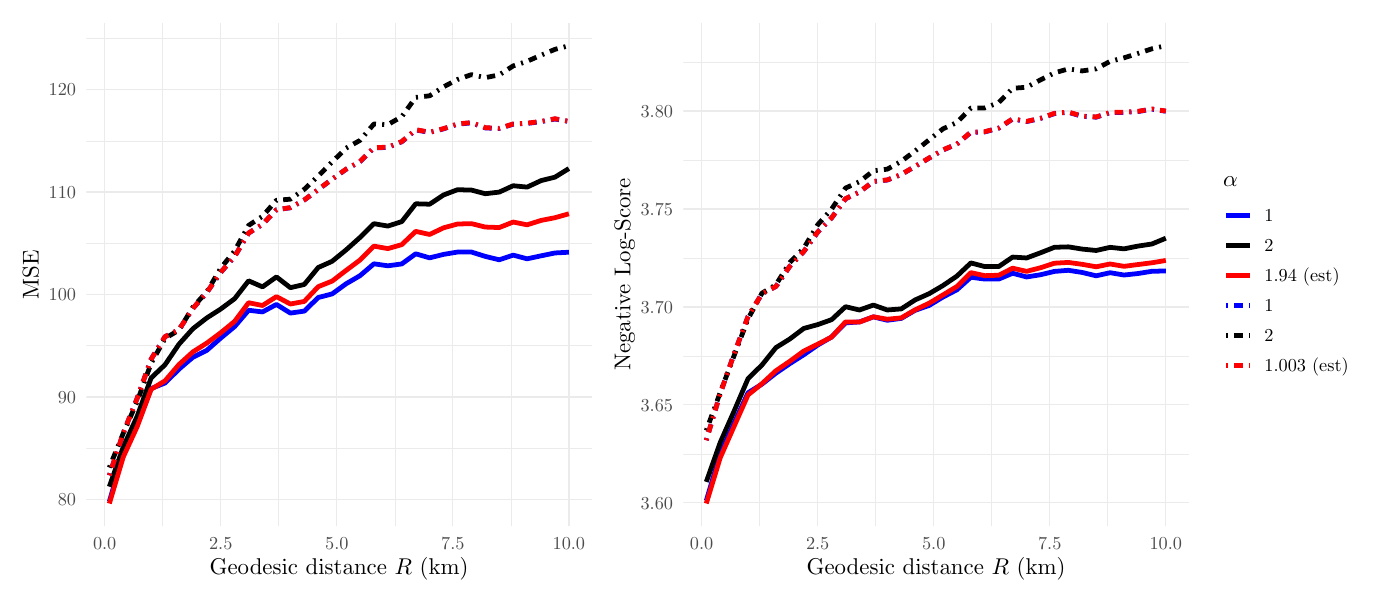}
    \caption{MSE and negative Log-Score as functions of distance for the stationary ($\boldsymbol{\cdot-}$) and non-stationary ($\boldsymbol{-\!\!\!-}$)  cases with $\alpha = 1$, $\alpha = 2$, and $\alpha$ estimated (est).}
    \label{cross-val}
    \end{figure}
    
Figure \ref{cross-val} shows that the non-stationary model outperforms its stationary counterpart, in both metrics and regardless of the distance and the value of $\alpha$ (fixed or estimated). This highlights the benefit of allowing models with non-stationary features, providing an advantage over alternative proposals. Regarding the smoothness parameter $\alpha$, we observe that small values of $\alpha$ produce better predictions, according to both metrics and regardless of the model class (stationary or not). For the stationary case, the estimated mean value of $\alpha$ is $1.003$, which is close to $1$, thus explaining the proximity of the blue and red dot-dashed curves in both panels of Figure \ref{cross-val}. Because there is no improvement in the model's performance when $\alpha$ is estimated, we might as well use the model with $\alpha = 1$ in this case. However, reaching this conclusion is only possible because of our ability to estimate the $\alpha$, which is evidence that allowing models with arbitrary smoothness is advantageous. Figure \ref{replicate14} shows the kriging predictor for the 14th replicate using the non-stationary model with $\alpha = 1$, which is the best model in terms of predictive performance.

To further isolate the effect of smoothness misspecification, we conduct a dedicated simulation study in Appendix \ref{app:effects_miss}, where data are generated from a model with known smoothness. The results show that fixing $\alpha$ at an incorrect value induces substantial bias in the SPDE parameters $\tau$ and $\kappa$ and leads to inferior predictive performance under both MSE and NLS. In contrast, jointly estimating $\alpha$ substantially improves recovery of the true dependence structure and yields consistently better predictions. These findings reinforce the practical importance of allowing flexible smoothness in spatial models.

\section{Conclusion}
\label{sec:conc}

We have proposed a flexible class of Gaussian random fields for metric graphs, which allow for non-stationary covariance functions and arbitrary smoothness. 
By building upon a chain of results inspired by the SPDE approach in Euclidean domains, and by combining the FEM-based numerical approximation of \citet{Bolin2024Regularity} with the rational approximation of \citet{Xiong2024Covariance}, we showed that the proposed class of models can be implemented in a computationally efficient way for general metric graphs, facilitating applications to large graphs and big data sets. The approach has been theoretically justified and numerical experiments verify the theoretical results. Further, the application presented demonstrated the advantages of the proposed approach for statistical modeling. 

Future work could explore the extension of the proposed framework to spatio-temporal models. This could allow, for example, to understand how traffic flow and congestion in road networks evolve over time due to changing traffic conditions or to study the diffusion and temporal dynamics of pollutants in river networks for environmental monitoring. Additionally, the framework could be extended to log-Gaussian Cox processes, enabling the modeling of point patterns constrained to graph structures. 
Future work from a theoretical perspective involves investigating identifiability of the model parameters, which could be studied following existing approaches for SPDEs, such as those in \citet{Bolin2023Equivalence}.

\setcounter{section}{0} 
\renewcommand{\thesection}{\Alph{section}} 
\renewcommand{\theequation}{\thesection.\arabic{equation}} 
\counterwithin{equation}{section} 

\renewcommand{\thefigure}{\thesection.\arabic{figure}} 
\counterwithin{figure}{section}   

\renewcommand{\thetable}{\thesection.\arabic{table}} 
\counterwithin{table}{section}   

\newpage

\section{Covariance functions}
\phantomsection 
\label{app:cov_funct_def}

To complement the numerical experiments in Section~\ref{numerical_experiments}, this appendix presents the exact covariance functions for Whittle--Mat\'ern fields defined on the interval, circle, and tadpole graphs. For the interval and circle graphs, their covariance functions can be written in terms of the folded Mat\'ern covariance function \citep{Khristenko2019Analysis}. To see this, let
\begin{equation*}
    C(h) = \dfrac{\sigma^2}{2^{\nu-1}\Gamma(\nu)}(\kappa h)^\nu K_\nu(\kappa h),\quad h\geq 0,
\end{equation*}
be the Mat\'ern covariance function, where $\Gamma(\cdot)$ is the Gamma function, $K_\nu(\cdot)$ is the modified Bessel function of the second kind, $\sigma^2 = \Gamma(\nu)/(\tau^2 \kappa^{2\nu}(4\pi)^{\sfrac{1}{2}} \Gamma(\nu+\sfrac{1}{2}))$ is the marginal variance, $\kappa>0$ is a parameter that controls the practical correlation range $\rho = \sqrt{8\nu}/\kappa$, and $0<\nu = \alpha-\sfrac{1}{2}$ is the smoothness parameter.
Then the covariance function on the interval $[0,L]$ is 
\begin{equation*}
    C_{\mathcal{N}}(s_1,s_2) = \sum_{k=-\infty}^\infty C(s_1-s_2+2kL)+C(s_1+s_2+2kL)
\end{equation*}
and the covariance function over a circle with arc-length parameterization on $[0,L]$ is 
\begin{equation*}
    C_{\mathcal{P}}(s_1,s_2) = \sum_{k=-\infty}^\infty C(s_1-s_2+2kL).
\end{equation*}

We remark that these infinite sums are not considered for numerical experiments but rather a truncated version of them. In the case of the tadpole graph, because the eigenpairs of the graph Laplacian $\Delta_\Gamma$ are known, the corresponding covariance function can be obtained using a truncated version of its Mercer representation \citep{Steinwart2012Mercer}, retaining only a finite number of terms in the summation.  Let $\Gamma_T = (\mathcal{V},\mathcal{E})$ characterize the tadpole graph with $\mathcal{V} = \{v_1,v_2\}$ and $\mathcal{E} = \{e_1,e_2\}$ as specified in Figure~\ref{Interval.Circle.Tadpole}c. The left edge $e_1$ has length 1 and the circular edge $e_2$ has length 2. As discussed in Subsection~\ref{subsec:prelim}, a point on $e_1$ is parameterized via $s=\left(e_1, t\right)$ for $t \in[0,1]$ and a point on $e_2$ via $s=\left(e_2, t\right)$ for $ t \in[0,2]$. One can verify that $-\Delta_\Gamma$ has eigenvalues $0,\left\{(i \pi / 2)^2\right\}_{i \in \mathbb{N}}$ and $\left\{(i \pi / 2)^2\right\}_{2 i \in \mathbb{N}}$ with corresponding eigenfunctions $\phi_0$, $\left\{\phi_i\right\}_{i \in \mathbb{N}}$, and $\left\{\psi_i\right\}_{2 i \in \mathbb{N}}$ given by $\phi_0(s)=1 / \sqrt{3}$ and 
\begin{equation*}
    \phi_i(s)=C_{\phi, i}\begin{cases}
        -2 \sin (\sfrac{i\pi}{2}) \cos (\sfrac{i \pi t}{2}), & s \in e_1, \\
\sin (\sfrac{i \pi t}{2}), & s \in e_2,
    \end{cases}
    \quad 
        \psi_i(s)=\frac{\sqrt{3}}{\sqrt{2}} \begin{cases}
    (-1)^{i / 2} \cos (\sfrac{i \pi t}{2}), & s \in e_1, \\
\cos (\sfrac{i \pi t}{2}), & s \in e_2,
\end{cases}
\end{equation*}
where $C_{\phi, i}=1$ if $i$ is even and $C_{\phi, i}=1 / \sqrt{3}$ otherwise. Moreover, these functions form an orthonormal basis for $L_2(\Gamma_T)$.

\section{Complementary illustrations}
\phantomsection 
\label{app:com_illustrations}

This appendix provides additional visual illustrations supporting the results in the main text.

\subsection{Non-stationary covariance features}
\phantomsection 
\label{app:model_prop}

Following Subsection~\ref{subsec:model_properties}, we specify $\tau(\cdot)$ and $\kappa(\cdot)$ via $\tau(s) = e^{0.05\cdot(x(s)-y(s))}$ and $\kappa(s) = e^{0.1\cdot(x(s)-y(s))}$, where $(x(s),y(s))$ are Euclidean coordinates on the plane. These choices induce large values of $\tau(\cdot)$ and $\kappa(\cdot)$ in the bottom-right region of the \texttt{MetricGraph} package's logo, as seen from the top row of Figure~\ref{four_plots}.  The bottom row of Figure~\ref{four_plots} shows the standard deviation $\sigma(\cdot)$ and practical correlation range $\rho(\cdot)$  for the case $\alpha = 0.9$. The resulting $\sigma(\cdot)$ and $\rho(\cdot)$ exhibit clear non-stationary behavior and display spatial patterns opposite to those of their counterparts $\tau(\cdot)$ and $\kappa(\cdot)$, illustrating that these parameters provide direct control over marginal variability and correlation structure. The panels in the bottom row of Figure~\ref{four_plots} also show three locations, $s_1$, $s_2$, and $s_3$, which are the same locations as in the left panel of Figure~\ref{cov_diff_range}.

\begin{figure}[t]
    \centering
    \includegraphics[width=0.95\textwidth]{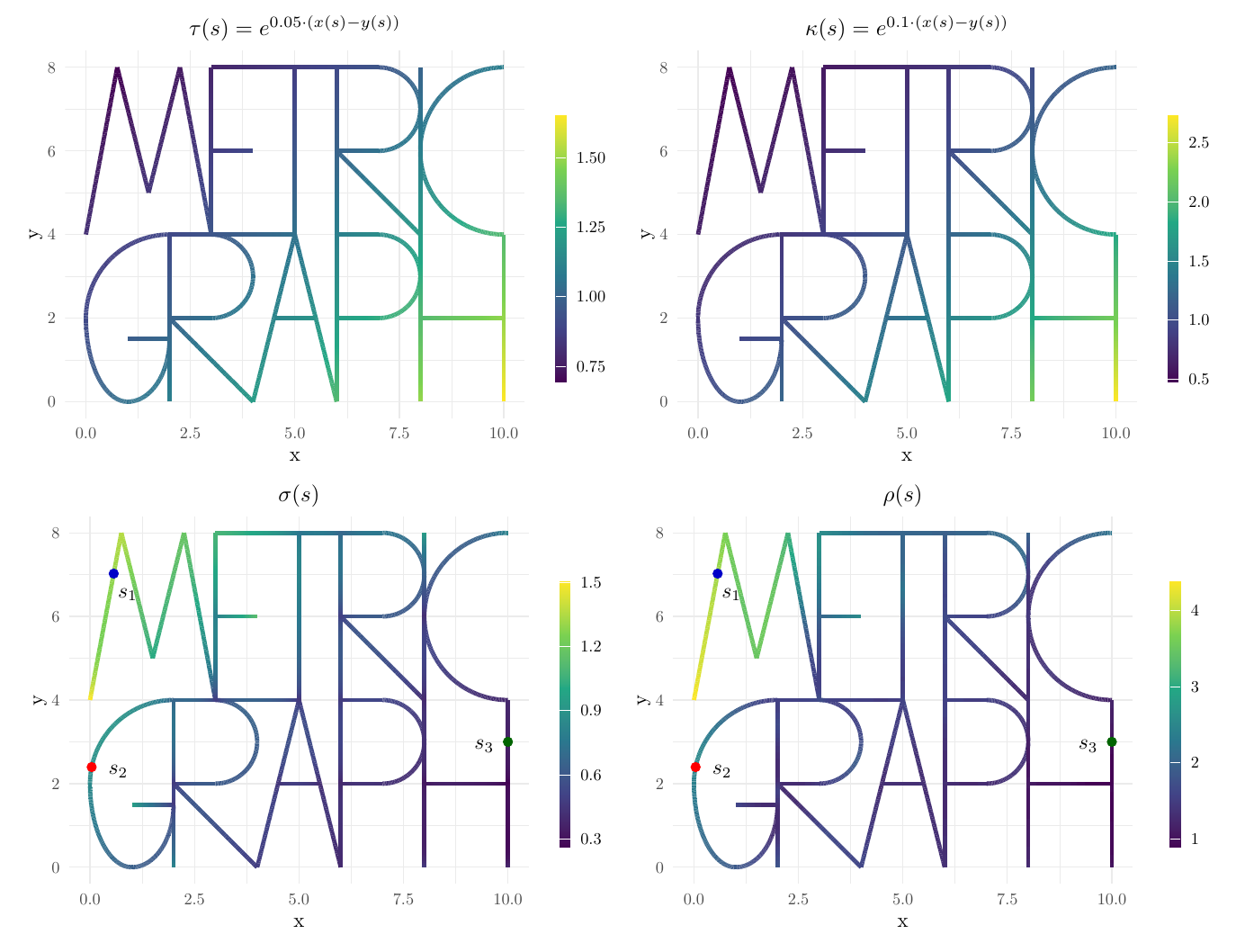}
    \caption{Non-stationary models for $\tau(\cdot)$ and $\kappa(\cdot)$ (top row). Non-stationary standard deviation $\sigma(\cdot)$ and practical correlation range $\rho(\cdot)$ (bottom row).}
    \label{four_plots}
\end{figure}

\subsection{Approximation errors}
\phantomsection 
\label{app:approx_error}

We illustrate the approximation errors for the interval and circle graphs, using the same parameter settings as in Section~\ref{numerical_experiments}. The errors for the interval graph are shown in Figure~\ref{interval_errors}, while those for the circle graph are shown in Figure~\ref{circle_errors}.

\begin{figure}[!t]
\centering
\includegraphics[width=1\textwidth]{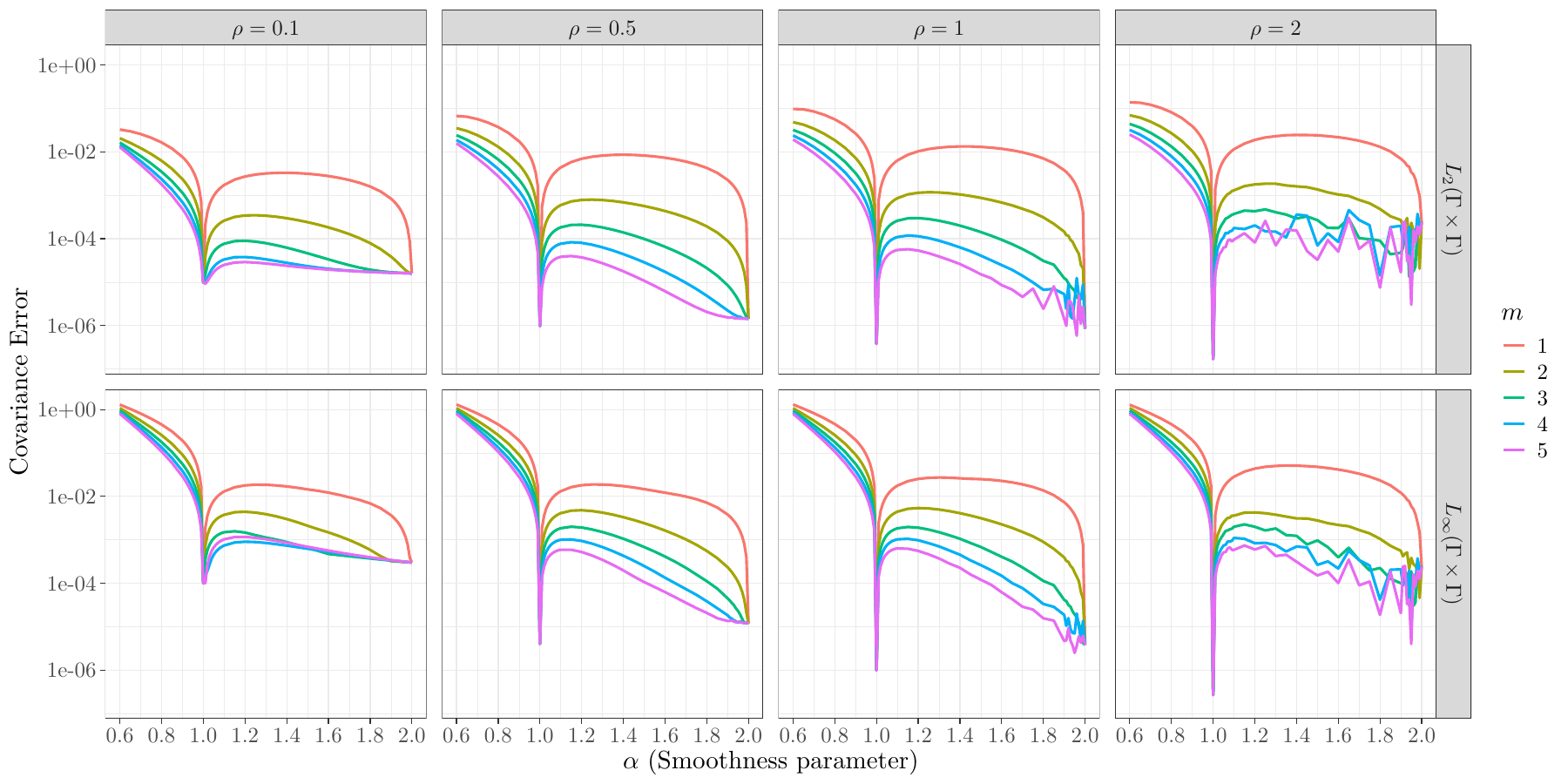}
\caption{Errors in $L_2(\Gamma\times\Gamma)$-norm (top row) and supremum norm (bottom row) on the interval graph for different combinations of parameters $\rho$ (practical range), $\nu$ (smoothness), and $m$ (rational order). In all cases, the FEM mesh contains 1000 equally spaced nodes.}
\label{interval_errors}
\end{figure}

\begin{figure}[!t]
\centering
\includegraphics[width=1\textwidth]{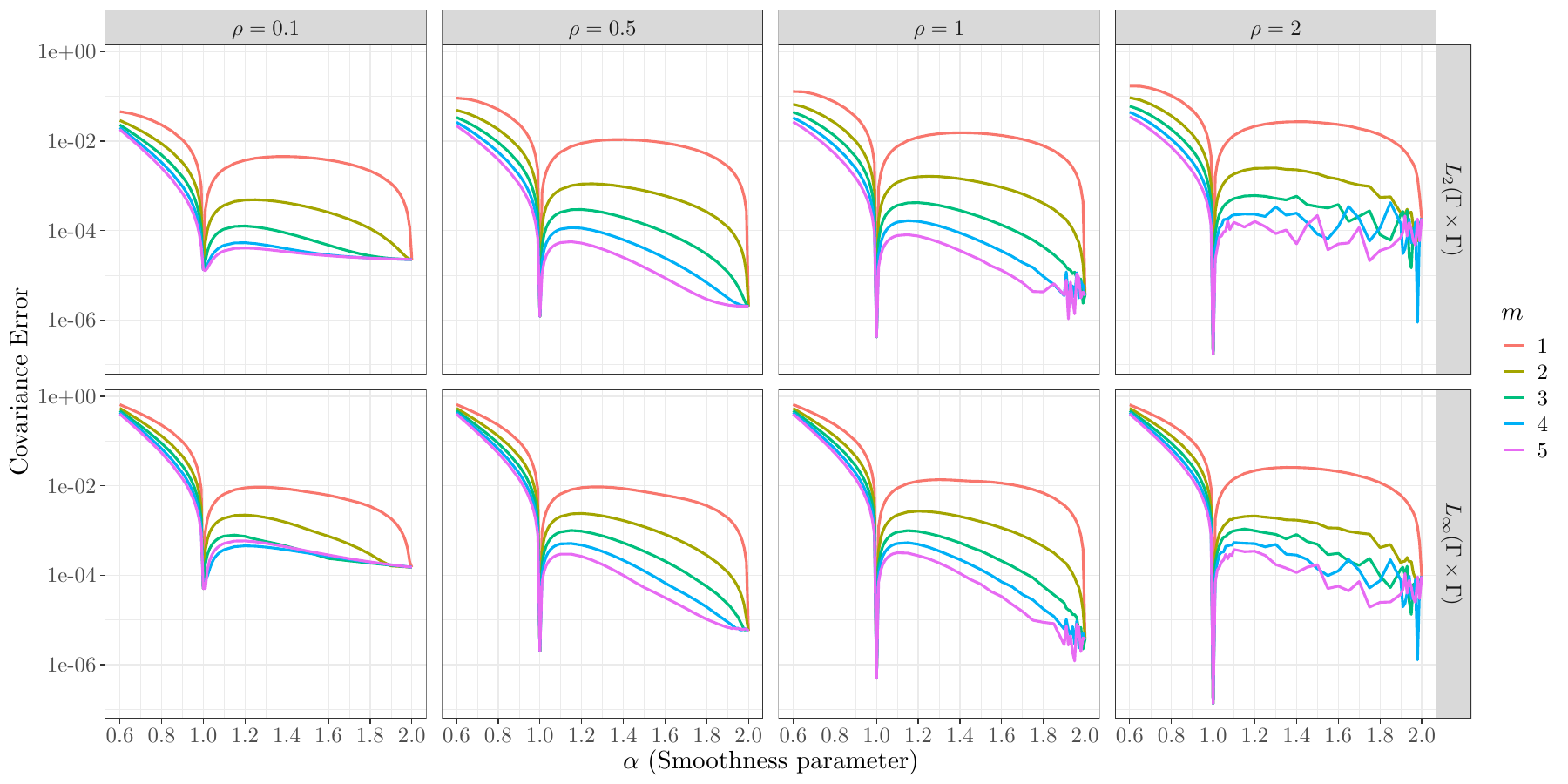}
\caption{Errors in $L_2(\Gamma\times\Gamma)$-norm (top row) and supremum norm (bottom row) on the circle graph for different combinations of parameters $\rho$ (practical range), $\nu$ (smoothness), and $m$ (rational order). In all cases, the FEM mesh contains 2000 equally spaced nodes.}
\label{circle_errors}
\end{figure}

\subsection{Covariates for the traffic speed application}
\phantomsection 
\label{app:covariates_for_applic}

We present illustrations of the covariates constructed in Subsection~\ref{ssec:const_covariates}. Figure~\ref{app:mean_better_3} displays the average speeds for the first part of the data. Points indicate measurements at each location, while colored edges correspond to the $\text{mean.cov}(s)$ covariate. Figure~\ref{stand_log_sigma} shows the standardized logarithm of the standard deviation over the same period, with points representing values at each location and colored edges indicating the $\text{std.cov}(s)$ covariate.

\begin{figure}[!t]
    \centering
    \includegraphics[width=0.95\textwidth]{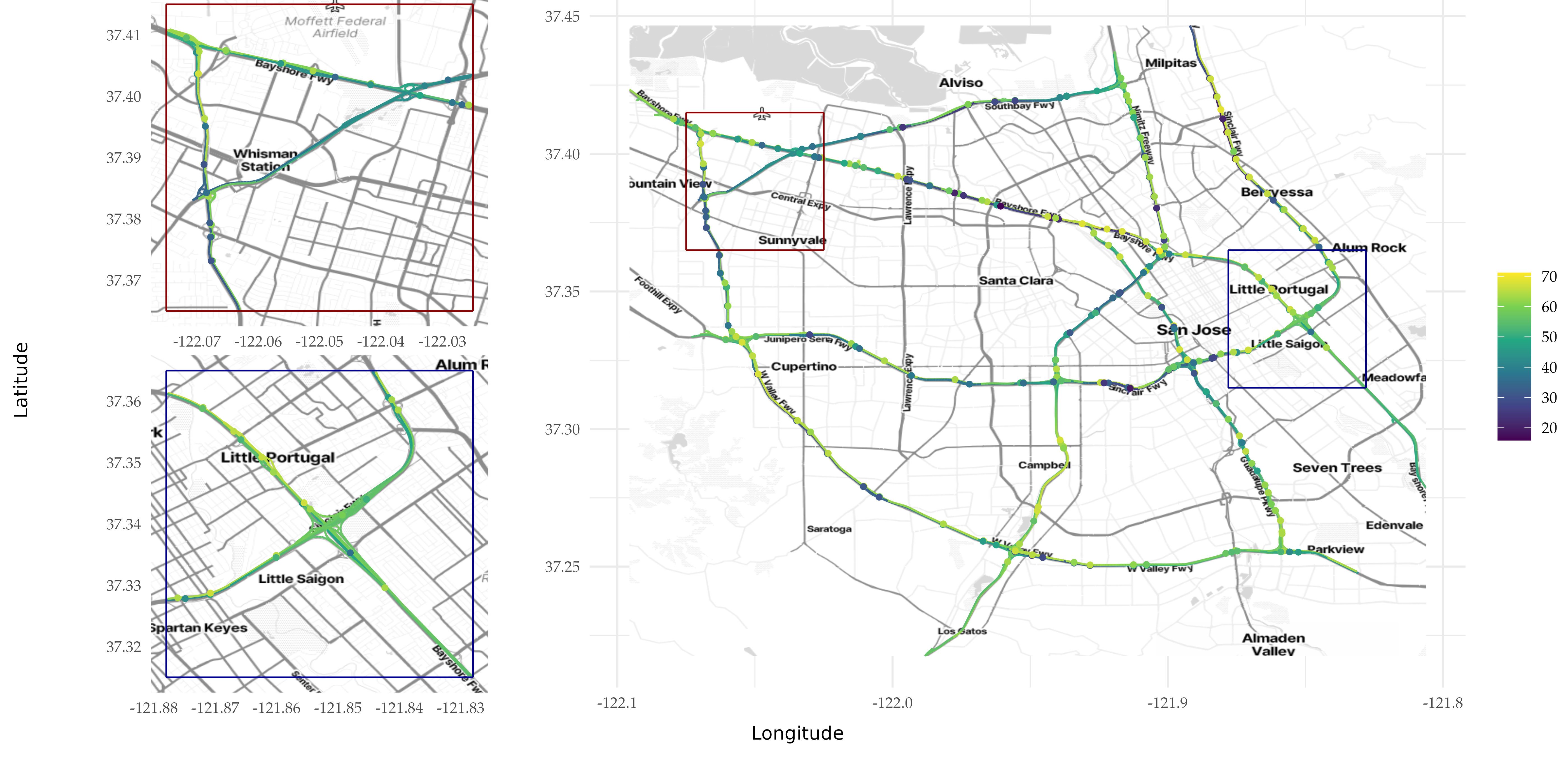}
    \caption{Average speeds of the first part of the data. Points represent the point value at each available location, and edges represent the $\text{mean.cov}(s)$ covariate.}
    \label{app:mean_better_3}
\end{figure}

\begin{figure}[!t]
    \centering
    \includegraphics[width=0.95\textwidth]{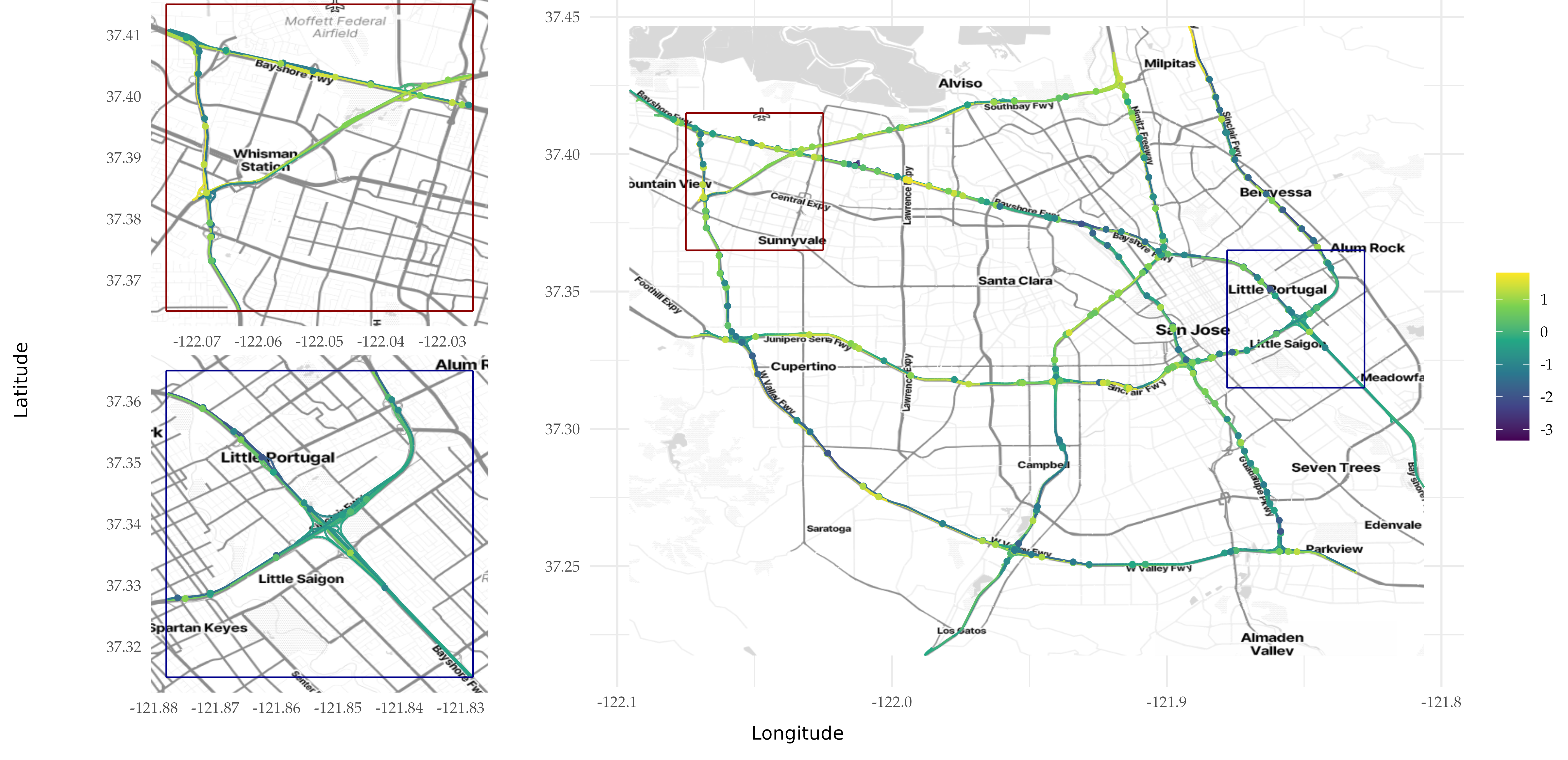}
    \caption{Standardized logarithm of the standard deviation for the first part of the data. Points show values at each location, and edges represent the $\text{std.cov}(s)$ covariate.}
    \label{stand_log_sigma}
\end{figure}

\section{Effects of smoothness misspecification}
\phantomsection 
\label{app:effects_miss}

We generate synthetic data $\{y_i\}_{i=1}^{200}$ from $10$ independent replicates according to\begin{equation}
\label{sim_model}
    y_i|u(\cdot)\sim N(\beta_0 + u(s_i),\sigma_\epsilon^2),\;i = 1,\dots, 200,
\end{equation} 
where $u(\cdot)$ is a Whittle--Mat\'ern field defined on the tadpole graph. The field is specified by the parameters $\sigma = 1.2$, $\rho = 0.2$, and $\nu = 0.8$ under the Mat\'ern parameterization, which correspond to $\tau = 0.06620297$, $\kappa = 12.64911$, and $\alpha = 1.3$ under the SPDE parameterization. The intercept is fixed at $\beta_0 = 1$, and the measurement noise $\sigma_e$ is set to $0.1$. We fit model \eqref{sim_model} to the synthetic data under three different specifications of the smoothness parameter $\alpha$: $\texttt{Model}_1$, in which $\alpha$ is fixed at $1$,
$\texttt{Model}_2$, in which $\alpha$ is fixed at $2$,
and $\texttt{Model}_{\texttt{est}}$, in which $\alpha$ estimated from the data.
Posterior mean estimates of the model parameters are reported in Table~\ref{table_simm}, together with the true values used to generate the data. Across all three model specifications, the posterior means of $\beta_0$ and $\sigma_\epsilon$ are stable and close to their true values, indicating that inference on the mean structure and measurement noise is relatively insensitive to the choice of $\alpha$. In contrast, substantial differences arise in the estimates of the SPDE parameters $\tau$ and $\kappa$ when $\alpha$ is misspecified. When $\alpha$ is fixed at either $1$ ($\texttt{Model}_1$) or $2$ ($\texttt{Model}_2$), the resulting estimates of $\tau$ and $\kappa$ significantly deviate from the true values. When $\alpha$ is estimated jointly with the remaining parameters ($\texttt{Model}_{\texttt{est}}$), the posterior mean $\hat{\alpha} = 1.322$ is close to the true value $\alpha = 1.3$, and the corresponding estimates of $\tau$ and $\kappa$ are also in close agreement with the true parameters. This demonstrates that allowing $\alpha$ to be estimated substantially improves recovery of the underlying spatial dependence structure and mitigates the bias induced by smoothness misspecification.

\begin{table}[!t]
    \caption{Posterior means of the parameters of model \eqref{sim_model}.}
    \label{table_simm}
    \centering
\begin{tabular}{lcccc}
\toprule
$\alpha$              & $\beta_0$ & $\tau$ & $\kappa$ & $\sigma_e$ \\ \hline 
$1$                    & $1.044$ & $0.241$ & $5.7002$ & $0.094$ \\ 
$2$                    & $1.039$ &  $0.003$ & $25.219$ & $0.125$ \\ 
$\texttt{est} = 1.322$ & $1.042$ &  $0.063$ & $12.388$ & $0.103$ \\ 
$\texttt{TRUE} = 1.300$ & $1.000$ &  $ 0.066$ & $12.649$ & $0.100$ \\
\bottomrule
\end{tabular}
\end{table}

As in Subsection~\ref{sec:comparison_crossval}, we assess the predictive performance of the three models using leave-group-out pseudo cross-validation. Predictions at each location are obtained by conditioning on the most informative observations, defined as those lying within a geodesic distance $R$ from the prediction location. Increasing $R$ therefore corresponds to incorporating progressively more data into the prediction. By varying $R$, we evaluate how predictive accuracy changes as a function of the amount of available local information. Model performance is summarized using the Mean Squared Error (MSE) and the Negative Log-Score (NLS). The results are shown in Figure~\ref{cv_sim}. Across the full range of $R$, models in which the smoothness parameter $\alpha$ is misspecified ($\texttt{Model}_1$ and $\texttt{Model}_2$) exhibit inferior predictive performance under both metrics, whereas allowing $\alpha$ to be estimated ($\texttt{Model}_{\texttt{est}}$) leads to consistently lower MSE and NLS values. This further illustrates the negative effect of smoothness misspecification on predictive accuracy, even when predictions are based on increasingly informative local neighborhoods.

\begin{figure}[!t]
    \centering
    \includegraphics[width=0.95\textwidth]{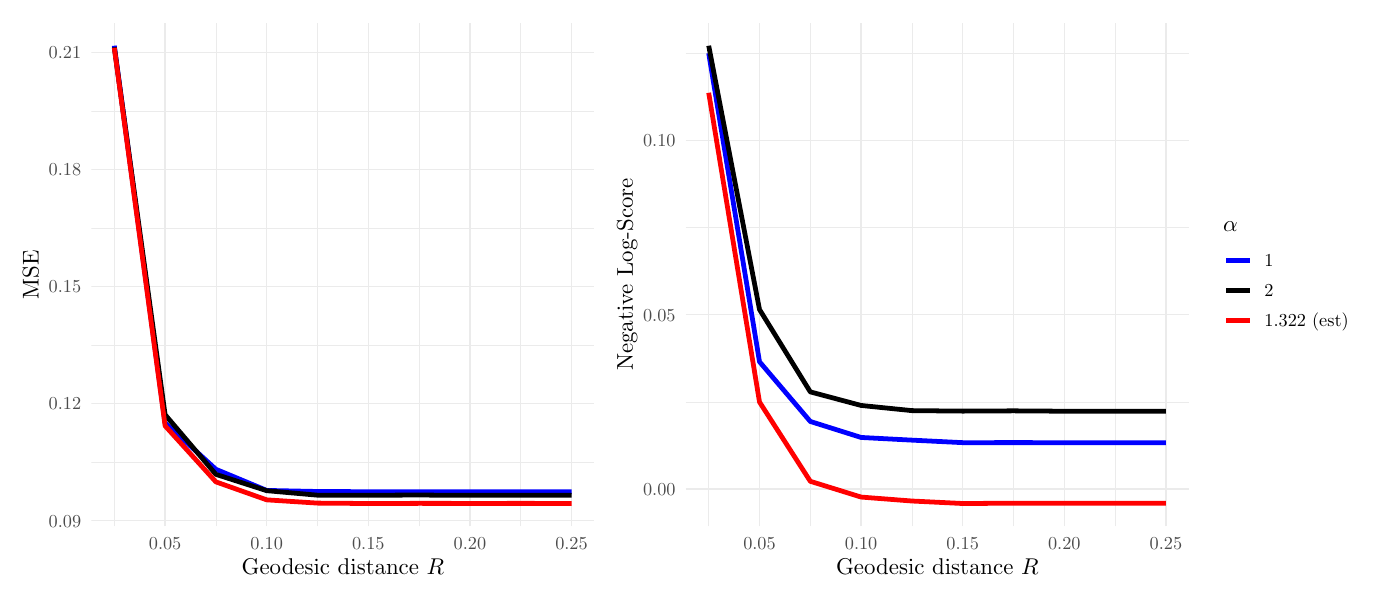}
    \caption{MSE and negative Log-Score as functions of distance for model \eqref{sim_model}  with $\alpha = 1$, $\alpha = 2$, and $\alpha$ estimated (est).}
    \label{cv_sim}
\end{figure}

\section{Additional notation}
\phantomsection 
\label{app:theoretical_details}

To support the technical details in the main text and the following appendices, we begin by introducing some additional notation. The Sobolev space $ H^1(\Gamma) $ is defined as 
\begin{align*}
    H^1(\Gamma) = \{f \in C(\Gamma) : \|f\|^2_{H^1(\Gamma)} < \infty \} = C(\Gamma) \cap \bigoplus_{e \in \mathcal{E}} H^1(e),
\end{align*}
with the norm
\begin{align*}
    \|f\|^2_{H^1(\Gamma)} = \sum_{e \in \mathcal{E}} \|f_{e}\|^2_{H^1(e)} = \sum_{e \in \mathcal{E}} \left( \int_{e} f_{e}'^2(s) \, ds + \int_{e} f_{e}^2(s) \, ds \right),
\end{align*}
where, for each $ e \in \mathcal{E} $, $ H^1(e) $ is the standard Sobolev space of order 1 on $[0, \ell_{e}]$, and $ f_{e}'$ is the weak derivative of $ f_{e} $, defined as the unique function in $L_2(e)$ satisfying  
\begin{equation}
\label{eq:weak_deriv}
    f_e(x) = f_e(0) + \int_0^x f_e'(s) \, ds \quad \text{for all } x \in [0, \ell_{e}].
\end{equation}
We define the decoupled Sobolev spaces of integer order $ k \geq 1 $,  by  
$\widetilde{H}^k(\Gamma) = \bigoplus_{e \in \mathcal{E}} H^k(e)$ with norm $\|f\|^2_{\widetilde{H}^k(\Gamma)} = \sum_{e \in \mathcal{E}} \|f_{e}\|^2_{H^k(e)}$, where $H^k(e)$ is the standard Sobolev space of order $ k $ on the edge $ e $. This space consists of functions in $ L^2(e) $ whose weak derivatives up to order $ k $ also belong to $L^2(e)$. Higher-order weak derivatives are defined inductively, with the $ j $-th derivative of a function $ f_e $ being the weak derivative of $ f_e^{(j-1)} $, as given by \eqref{eq:weak_deriv}. It follows that any $ f_e \in H^1(e) $ is almost everywhere equal to a continuous function. Further, if $ f \in \widetilde{H}^2(\Gamma) $, then for every $ e \in \mathcal{E} $, $ f_e' $ can be identified with a continuous function. At this point, we note that odd-order derivatives depend on the orientation (direction) of edge parameterizations, while even-order derivatives do not. Consequently, strictly speaking, one should indicate the chosen orientation when writing odd-order derivatives. However, doing so would burden the notation with limited benefit. We therefore choose not to make this dependence explicit. We refer the reader to \cite{Awadelkarim2025Fractional} for further details. To ensure continuity across $ \Gamma $, we define the spaces $\widetilde{H}_C^k(\Gamma) = \widetilde{H}^k(\Gamma) \cap C(\Gamma)$. Note that $ H^1(\Gamma) = \widetilde{H}^1(\Gamma) \cap C(\Gamma) = \widetilde{H}_C^1(\Gamma) $, so functions in $ H^1(\Gamma) $ are guaranteed to be continuous, unlike functions in the decoupled Sobolev space $ \widetilde{H}^1(\Gamma)$. The space $ K(\Gamma) $ is then defined as the space of functions in $ \widetilde{H}^2_C(\Gamma) $ that satisfy the Kirchhoff vertex conditions. Specifically, 
\begin{equation*}
\label{eq:kirchhoff_space}
K(\Gamma) = \mathcal{K}(\Gamma)\cap \widetilde{H}^2_C(\Gamma) =   \left\{ u \in \widetilde{H}^2_C(\Gamma) \;\Big|\; \forall v \in \mathcal{V}, \; \sum_{s\in v}\partial u(s)=0 \right\},
\end{equation*}
where $\mathcal{K}(\Gamma)$ is given by \eqref{eq:kirchhoff_cond} and $ \partial_e $ denotes the directional derivative along $ e $, taken away from $ v $. Further, $K(\Gamma)$ is well-defined since for functions $f\in\widetilde{H}^2(\Gamma)$, $f_e'$ can be identified with a continuous function, so the expression $\partial u(s)$ is meaningful.

Let $ (E, \|\cdot\|_E) $ and $ (F, \|\cdot\|_F) $ be two separable Hilbert spaces. We denote by $ \mathcal{L}(E, F) $ the Banach space of linear bounded operators from $ E $ to $ F $, equipped with the norm $\|\cdot\|_{\mathcal{L}(E, F)} = \sup_{\|x\|_E = 1} \|\cdot x\|_F$. The space $ \mathcal{L}_2(E, F) $ consists of Hilbert-Schmidt operators from $E$ to $F$, with the Hilbert-Schmidt norm defined by 
\begin{align*}
    \|\cdot\|^2_{\mathcal{L}_2(E, F)} = \sum_{j \in \mathbb{N}} \|\cdot e_j\|^2_F,
\end{align*}
where $ (e_j)_j $ is a complete orthonormal basis of $ E $. For the special case $ \mathcal{L}_2(E, E) $, we write $ \mathcal{L}_2(E) $ for simplicity. If $ E \subset F $ with continuous inclusion, i.e., there exists $ C > 0 $ such that $ \|f\|_F \leq C \|f\|_E $ for all $ f \in E $, we write $ E \hookrightarrow F $. If $ E \hookrightarrow F \hookrightarrow E $, we write $ E \cong F $.

Let $ (\Omega, \mathcal{F}, \mathbb{P}) $ be a complete probability space. For a real-valued random variable $ Z $, its expectation is given by $ \mathbb{E}(Z) = \int_\Omega Z(\omega) \, d\mathbb{P}(\omega) $. The space $ L_2(\Omega) $ denotes the Hilbert space of equivalence classes of real-valued random variables with finite second moments. For a separable Hilbert space $ (E, \|\cdot\|_E) $, the space $ L_2(\Omega; E) $ consists of $ E $-valued Bochner-measurable random variables with finite second moments, equipped with the norm 
\begin{align*}
    \|u\|_{L_2(\Omega; E)}^2 = \int_\Omega \|u(\omega)\|_E^2 \, d\mathbb{P}(\omega).
\end{align*}

If $ E \subset F $, we denote the real interpolation space of order $ s \in (0, 1) $ between $ F $ and $ E $ by $ (F, E)_s $. The fractional Sobolev spaces are defined as follows:
\begin{align*}
H^s(\Gamma) &:=
\begin{cases}
(L_2(\Gamma), H^1(\Gamma))_s, & \text{for } 0 < s < 1, \\
(H^1(\Gamma), \widetilde{H}_C^2(\Gamma))_{s-1}, & \text{for } 1 \leq s \leq 2,
\end{cases}\\
\widetilde{H}^s(\Gamma) &:=
\begin{cases}
(L_2(\Gamma), \widetilde{H}^1(\Gamma))_s, & \text{for } 0 < s < 1, \\
(\widetilde{H}^1(\Gamma), \widetilde{H}^2(\Gamma))_{s-1}, & \text{for } 1 \leq s \leq 2.
\end{cases}
\end{align*}

The essential supremum and infimum of a function $ f $ over $ \Gamma $ are defined as:
\begin{align*}
    \esssup_{s \in \Gamma} |f(s)| &:= \inf \{ M \geq 0 : \lambda(\{s \in \Gamma : |f(s)| > M\}) = 0 \},\\
\essinf_{s \in \Gamma} |f(s)| &:= \sup \{ m \geq 0 : \lambda(\{s \in \Gamma : |f(s)| < m\}) = 0 \},
\end{align*}
where $ \lambda $ is the Lebesgue measure on $ \Gamma $, given by $ \lambda(A) = \sum_{e \in \mathcal{E}} \lambda_e(A \cap e) $, with $ \lambda_e $ being the Lebesgue measure on each edge $ e $, identified with a compact interval.

The Gaussian linear space associated with a Gaussian random field $ u $ is defined as ${H_u = \overline{\mathrm{span}\{u(s) : s \in \Gamma\}}}$, where the closure is taken in $ L_2(\Omega) $. If $ e $ is an edge and ${f : e \to H_u}$ is a function, we say that $ f $ is weakly differentiable at $ s $ in the $ L_2(\Omega) $ sense if there exists $ f'(s) \in H_u $ such that, for every $ v \in H_u $ and every sequence $ s_n \to s $ with $ s_n \neq s $, we have 
\begin{align*}
    \mathbb{E}\left( v \frac{f(s_n) - f(s)}{s_n - s} \right) \to \mathbb{E}(w f'(s)).
\end{align*}

Finally, the Cameron-Martin space (also known as the reproducing kernel Hilbert space) associated with $ u $ is defined as $\mathcal{H}_u = \{h(s) = \mathbb{E}(vu(s)) : s \in \Gamma, v \in H_u\}$.

\section{Preliminaries on H\"older spaces}
\phantomsection 

In this section, we provide some general results on H\"older spaces which will be needed in what follows. The results are stated and proved for general metric spaces, as the proofs for metric graphs and general metric spaces are identical.

Let $(X, \tilde{d})$ be a metric space, and let $C(X)$ denote the space of real-valued continuous functions on $X$. For $\gamma \in (0, 1]$, the $\gamma$-H\"older seminorm is defined as
\begin{align*}
    [f]_{C^{0,\gamma}(X)} = \sup_{s, s’ \in X} \dfrac{|f(s) - f(s’)|}{\tilde{d}(s, s’)^\gamma},
\end{align*}
where $f \in C(X)$. The space of functions $f \in C(X)$ with a finite seminorm $[f]_{C^{0,\gamma}(X)}$ is denoted by $C^{0,\gamma}(X)$.
We equip the H\"older space $C^{0,\gamma}(X)$ with the norm
\begin{align*}
    |f|_{C^{0,\gamma}(X)} = |f|_{C(X)} + [f]_{C^{0,\gamma}(X)},
\end{align*}
for $f \in C^{0,\gamma}(X)$. A function $f \in C^{0,\gamma}(X)$ is referred to as $\gamma$-H\"older continuous.

We now state and prove some fundamental properties of the H\"older seminorm and, more generally, of H\"older spaces. All proofs are straightforward adaptations of those provided in \cite{Fiorenza2016Holder} for real-valued functions defined on subsets of $\mathbb{R}^d$, where $d \in \mathbb{N}$.
We start by showing the relation between H\"older spaces with different exponents. 

Throughout this section, $(X,\tilde{d})$ denotes a metric space.

\begin{prop}
\label{prop:Holder_exponents}
Let $0<\mu < \gamma \leq 1$ and define
\begin{align*}
    \delta = \operatorname{diam} X = \sup_{x', x'' \in X} \tilde{d}(x', x'').
\end{align*}
Then, for any $f\in C^{0,\gamma}(X)$, $[f]_{C^{0,\mu}(X)} \leq \delta^{\gamma - \mu}[f]_{C^{0,\gamma}(X)}$. In particular, if $\delta <\infty$ and $f$ is $\gamma$-H\"older continuous, then $f$ is also $\mu$-H\"older continuous.
\end{prop}

\begin{proof}
For any $x', x'' \in X$, we have
\begin{align*}
|f(x') - f(x'')| \leq [f]_{C^{0,\gamma}(X)} \tilde{d}(x', x'')^\gamma 
&= [f]_{C^{0,\gamma}(X)} \tilde{d}(x', x'')^{\gamma - \mu} \tilde{d}(x', x'')^\mu \\
&\leq [f]_{C^{0,\gamma}(X)} \delta^{\gamma - \mu} \tilde{d}(x', x'')^\mu.
\end{align*}
Thus, the inequality is satisfied and if $\delta<\infty$, then $f$ is $\mu$-H\"older continuous on $X$. \hfill
\end{proof}

We will now show properties related to the sums, products and compositions of H\"older continuous functions. We start with the sum:
\begin{prop}
\label{prop:Holder_sum}
If $f$ and $g$ are $\gamma$-H\"older continuous on $X$, then their sum $f + g$ is also $\gamma$-H\"older continuous. Moreover, the H\"older seminorm satisfies
\begin{equation}
\label{eq:holder1}
    [f + g]_{C^{0,\gamma}(X)} \leq [f]_{C^{0,\gamma}(X)} + [g]_{C^{0,\gamma}(X)}.
\end{equation}
\end{prop}

\begin{proof}
For any $x', x'' \in X$, we have
\begin{align*}
|(f + g)(x') - (f + g)(x'')| &\leq |f(x') - f(x'')| + |g(x') - g(x'')| \\
&\leq [f]_{C^{0,\gamma}(X)} \tilde{d}(x', x'')^\gamma + [g]_{C^{0,\gamma}(X)} \tilde{d}(x', x'')^\gamma \\
&= ([f]_{C^{0,\gamma}(X)} + [g]_{C^{0,\gamma}(X)}) \tilde{d}(x', x'')^\gamma.
\end{align*}
This shows that $f + g$ is $\gamma$-H\"older continuous, and that \eqref{eq:holder1} holds. \hfill
\end{proof}

\begin{prop}
\label{prop:Holder_product}
If $f$ and $g$ are bounded and $\gamma$-H\"older continuous on $X$, then their product $fg$ is also $\gamma$-H\"older continuous. Moreover, the H\"older seminorm satisfies
\begin{equation}
\label{eq:holder2}
[fg]_{C^{0,\gamma}(X)} \leq [f]_{C^{0,\gamma}(X)} \sup_{x \in X} |g(x)| + [g]_{C^{0,\gamma}(X)} \sup_{x \in X} |f(x)|.    
\end{equation}
\end{prop}

\begin{proof}
For any $x', x'' \in X$, we have
\begin{align*}
|f(x')g(x') - f(x'')g(x'')| &= |f(x')g(x') - f(x'')g(x'') + f(x'')g(x') - f(x'')g(x')| \\
&\leq |g(x')||f(x') - f(x'')| + |f(x'')||g(x') - g(x'')| \\
&\leq [f]_{C^{0,\gamma}(X)} \sup_{x \in X} |g(x)| \, \tilde{d}(x', x'')^\gamma\\
&\quad + [g]_{C^{0,\gamma}(X)} \sup_{x \in X} |f(x)| \, \tilde{d}(x', x'')^\gamma \\
&= \big([f]_{C^{0,\gamma}(X)} \sup_{x \in X} |g(x)| + [g]_{C^{0,\gamma}(X)} \sup_{x \in X} |f(x)|\big) \tilde{d}(x', x'')^\gamma.
\end{align*}
This shows that $fg$ is $\gamma$-H\"older continuous, and that \eqref{eq:holder2} holds. \hfill
\end{proof}

\begin{prop}
\label{prop:Holder_reciprocal}
If $g$ is $\gamma$-H\"older continuous function on $X$ such that $\inf_{x \in X} |g(x)| > 0$, then $1/g$ is also $\gamma$-H\"older continuous on $X$. Moreover, the H\"older seminorm satisfies
\begin{align*}
    \left[\frac{1}{g}\right]_{C^{0,\gamma}(X)} \leq \frac{[g]_{C^{0,\gamma}(X)}}{(\inf_{x \in X} |g(x)|)^2}.
\end{align*}
\end{prop}

\begin{proof}
For any $x', x'' \in X$, we have
\begin{align*}
\left|\frac{1}{g(x')} - \frac{1}{g(x'')}\right| 
&= \left|\frac{1}{g(x')} - \frac{1}{g(x'')}\right| 
= \frac{|g(x') - g(x'')|}{|g(x')||g(x'')|} 
\leq \frac{[g]_{C^{0,\gamma}(X)}}{(\inf_{x \in X} |g(x)|)^2} \tilde{d}(x', x'')^\gamma,
\end{align*}
which establishes the desired result. \hfill
\end{proof}

\begin{prop}
\label{prop:Holder_quotient}
Let $f$ and $g$ be $\gamma$-H\"older continuous functions on $X$. If $g$ satsifies ${\inf_{x \in X} |g(x)| > 0}$, then $f/g$ is also $\gamma$-H\"older continuous, and the H\"older seminorm satisfies
\begin{align*}
    \left[\frac{f}{g}\right]_{C^{0,\gamma}(X)} 
\leq \frac{[f]_{C^{0,\gamma}(X)} \sup_{x \in X} |g(x)| + [g]_{C^{0,\gamma}(X)} \sup_{x \in X} |f(x)|}{(\inf_{x \in X} |g(x)|)^2}.
\end{align*}
\end{prop}

\begin{proof}
The result follows directly from Proposition~\ref{prop:Holder_product}, with $g$ replaced by $1/g$, and using Proposition~\ref{prop:Holder_reciprocal}. \hfill
\end{proof}

\begin{prop}
\label{prop:Holder_composition}
If $f$ and $g$ are H\"older continuous functions, with $f$ being $\gamma$-H\"older continuous on $Y = g(X)$ and $g$ being $\mu$-H\"older continuous on $X$, then the composition $f \circ g$ is $(\gamma \mu)$-H\"older continuous. Moreover, the H\"older seminorm satisfies
\begin{equation}\label{eq:holder3}
[f \circ g]_{C^{0,\gamma \mu}(X)} \leq [f]_{C^{0,\gamma}(Y)} \big([g]_{C^{0,\mu}(X)}\big)^\gamma.    
\end{equation}
\end{prop}

\begin{proof}
For any $x', x'' \in X$, we have
\begin{align*}
|f(g(x')) - f(g(x''))| 
&\leq [f]_{C^{0,\gamma}(Y)} |g(x') - g(x'')|^\gamma 
\leq [f]_{C^{0,\gamma}(Y)} \big([g]_{C^{0,\mu}(X)} \tilde{d}(x', x'')^\mu\big)^\gamma \\
&= [f]_{C^{0,\gamma}(Y)} \big([g]_{C^{0,\mu}(X)}\big)^\gamma \tilde{d}(x', x'')^{\gamma \mu}.
\end{align*}
Thus, $f \circ g$ is $(\gamma \mu)$-H\"older continuous, and \eqref{eq:holder3} holds. \hfill
\end{proof}

\section{Uniqueness of the solution}
\phantomsection 
\label{app:uniqueness}

We now proceed with a rigorous definition of the differential operator $L = \kappa^2 - \Delta_{\Gamma}$ under Assumption~\ref{assumption1}, along with some of its fundamental properties.  Recall that $\Delta_{\Gamma}$ is the Kirchhoff--Laplacian, acting on functions in $K(\Gamma)$. Specifically, $\Delta_{\Gamma} $ is defined on $K(\Gamma)$ and acts on a function $f \in K(\Gamma)$ in such a way that, for every 
$e \in \mathcal{E}$, we have $ (\Delta_{\Gamma} f)|_e = f_e''$. Thus, the domain of $ \Delta_{\Gamma} $ consists of continuous functions that are twice differentiable on each edge and satisfy the Kirchhoff vertex conditions specified in $ K(\Gamma)$. These conditions enforce a sum-to-zero constraint on the directional derivatives at each vertex. In particular, this implies that the derivative is continuous at vertices of degree 2.

Under these assumptions, $L = \kappa^2 - \Delta_{\Gamma}$ is densely defined, selfadjoint, and positive definite with a compact inverse \citep{Bolin2024Regularity}. Consequently, it has a collection of eigenfunctions $(e_j)_j$, which form an orthonormal basis of $L_2(\Gamma)$, and corresponding eigenvalues $(\lambda_j)_j$.  Moreover, according to \citep[Cor. 2.13]{Bolin2024Regularity}, the eigenvalues $(\lambda_j)_j$ of $L$ satisfy Weyl's law, which states that there exist positive constants $C_1$ and $C_2$ for all $j\in\mathbb{N}$ such that $C_1j^2\leq\lambda_j\leq C_2j^2$. The fractional operator $L^{\sfrac{\alpha}{2}} : D(L^{\sfrac{\alpha}{2}}) \longmapsto L_2(\Gamma)$ is then defined in the spectral sense as
\begin{equation}
\label{lfractional}
        \phi \longmapsto L^{\sfrac{\alpha}{2}}\phi = \displaystyle\sum_{j\in\mathbb{N}}\lambda_j^{\sfrac{\alpha}{2}}(\phi, e_j)_{L_2(\Gamma)}e_j,
\end{equation}
where $ D(L^{\sfrac{\alpha}{2}}) = \dot{H}^{\alpha} = \{\phi\in L_2(\Gamma): \|\phi\|_\alpha<\infty\}$ is a Hilbert space  with inner product $(\phi,\psi)_\alpha = (L^{\sfrac{\alpha}{2}} \phi, L^{\sfrac{\alpha}{2}} \psi)_{L_2(\Gamma)}$ and norm $\|\phi\|_\alpha^2 = \|L^{\sfrac{\alpha}{2}} \phi\|^2_{L_2(\Gamma)} = \sum_{j\in\mathbb{N}}\lambda_j^{\alpha}(\phi, e_j)^2_{L_2(\Gamma)}$. 

Let, now, $M_g:L_2(\Gamma)\to L_2(\Gamma)$ be the multiplication operator by a function $g\in L_{\infty}(\Gamma)$ such that $\essinf_{s\in\Gamma} g(s) > 0$. Then, the operator $M_g$ is self-adjoint and coercive, and hence has a unique inverse $M_g^{-1}$. We, thus, define a solution to equation \eqref{eq:spde} as a centered Gaussian random field $u$ such that for all $h\in L_2(\Gamma)$,
\begin{align*}
    (u,h)_{L_2(\Gamma)} = \mathcal{W}((M_\tau^{-1} L^{-\sfrac{\alpha}{2}})^* h) = \mathcal{W}(L^{-\sfrac{\alpha}{2}}(M_\tau^{-1}h)),
\end{align*}
where $T^*$ denotes the adjoint of a linear operator $T$ and $\mathcal{W}$ is Gaussian white noise on $\Gamma$ introduced in Subsection~\ref{sec:model_class}.

We are now in a position to prove the uniqueness part of Proposition~\ref{prp:existence}.

\begin{proof}[Proof of Proposition~\ref{prp:existence} (uniqueness)]
    Let $\alpha>\sfrac{1}{2}$ and assume that Assumption~\ref{assumption1} holds. Let, $\tau_\infty := \esssup_{s\in\Gamma} \tau(s) <\infty$, since $\tau\in L_\infty(\Gamma)$, and recall that ${\tau_0 := \essinf_{s\in\Gamma} \tau(s) > 0}$. Now, consider the following auxiliary problem:
    \begin{equation}
    \label{eq:auxiliary_problem}
        (\kappa^2 - \Delta_\Gamma)^{\sfrac{\alpha}{2}}w = \mathcal{W}, \quad \hbox{on } \Gamma,
    \end{equation}
    where $\mathcal{W}$ is Gaussian white noise. This problem has a unique solution $w$ by \citet[Proposition 3.2]{Bolin2024Regularity}. Define $\widetilde{u} = \tau^{-1} w$, which is well-defined as ${0 < \tau_0 \leq \tau \leq \tau_\infty <\infty}$. Then, it follows that $\widetilde{u}$ is a solution to \eqref{eq:spde}. Conversely, in the same manner, if $u$ is a solution to \eqref{eq:spde}, then $\widetilde{w} = \tau u$  is a solution to \eqref{eq:auxiliary_problem}. This gives a one-to-one correspondence between solutions to \eqref{eq:spde} and \eqref{eq:auxiliary_problem}. 
    The uniqueness thus follows from the uniqueness of solutions to \eqref{eq:auxiliary_problem}. \hfill
\end{proof}

The next lemma shows that the reciprocal of a function satisfying the Kirchhoff vertex conditions, as well as the product of functions satisfying the Kirchhoff vertex conditions, also satisfy the Kirchhoff vertex conditions. Here, $C^1(e)$ denotes the set of functions that are continuously differentiable on $e$.

\begin{lemma}
\label{lem:kirchhof_product_stable}
    Let $f, g, h \in C(\Gamma)$ be functions such that, for every $e \in \mathcal{E}$, $f_e, g_e, h_e \in C^1(e)$, and they satisfy the Kirchhoff vertex conditions \eqref{eq:kirchhoff_cond}. Then, the product $fg$ also satisfies the Kirchhoff vertex conditions. Moreover, if $h$ satisfies $h(s) \geq h_0$ for all $s \in \Gamma$, where $h_0 > 0$, then $h^{-1}$ also satisfies the Kirchhoff vertex conditions.
\end{lemma}

\begin{proof}
Let $f,g,h$ be as in the statement. Then, $h^{-1}$ and $fg$ are continuous. Further,
\begin{align*}
    \sum_{s\in v} \partial h^{-1}(s) = -\sum_{s\in v} (h^{-1}(s))^2\partial h(s)  = -\frac{1}{(h(v))^2} \sum_{s \in v} \partial h(s) = 0,
\end{align*}
where we used the continuity of $h$ in the last equality. Therefore, $h^{-1}$ satisfies the Kirchhoff vertex conditions. Similarly, by using the continuity of $f$ and $g$,
\begin{align*}
    \sum_{s\in v} \partial (fg)(s) &= \sum_{s\in v} g(s) \partial f(s)  + \sum_{s\in v} f(s) \partial g(s)\\
    &= g(v)\sum_{s\in v} \partial f(s)  + f(v)\sum_{s\in v} \partial g(s) = 0.
\end{align*}
Therefore, $fg$ satisfies the Kirchhoff vertex conditions. \hfill
\end{proof}

In order to prove positive definiteness statement of Proposition~\ref{prp:existence}, we will need the following auxiliary result, which is an extension of \citet[Proposition 1]{Bolin2023Statistical}.

\begin{prop}
\label{prp:positive_definite_cov}
    Let $\alpha>\sfrac{1}{2}$ and let $w$ be the solution to \eqref{eq:auxiliary_problem}. Then, under Assumption~\ref{assumption1}, $w$ has a strictly positive-definite covariance function.
\end{prop}

\begin{proof}
Let $\rho(\cdot,\cdot)$ be the covariance function of $w$. Now, observe that by \citet[p. 44]{Bogachev1998Gaussian}, the Cameron-Martin space associated to $w$ is $\dot{H}^\alpha = D(L^{\sfrac{\alpha}{2}})$. Further, by \citet[Corollary 8.16]{Janson1997Gaussian}, the reproducing kernel Hilbert space associated to $\rho$ is the Cameron-Martin space associated to $w$, that is, $\dot{H}^\alpha$. By \citet{Sriperumbudur2011Universality}, a covariance kernel is strictly positive-definite if its reproducing kernel Hilbert space is dense in $C(\Gamma)$. Thus, it suffices to show that $\dot{H}^\alpha$ is dense in $C(\Gamma)$. To this end, we follow the proof of \citet[Proposition 1]{Bolin2023Statistical} and introduce a function class $\mathcal{A}(\Gamma)$. Using the Stone--Weierstrass theorem, $\mathcal{A}(\Gamma)$ can be shown to be dense in $C(\Gamma)$ (see \citet[Lemma 2]{Bolin2023Statistical}). It then remains to verify that $\mathcal{A}(\Gamma)\subset \dot{H}^\alpha$. Accordingly, we define
\begin{align*}
    \mathcal{A}(\Gamma) = \langle 1 \rangle + D(\Gamma) + \mathcal{S}_c(\Gamma) = \{c + f + g: c \in \mathbb{R}, f \in D(\Gamma), g \in \mathcal{S}_c(\Gamma)\}, 
\end{align*}
where $ \langle 1 \rangle = \mathrm{span}\{1\} $ is the space of constant functions on $ \Gamma $, $ D(\Gamma) = \bigoplus_{e \in \mathcal{E}} C^\infty_c(e) $ is the space of functions with support in the union of the interiors of edges, whose restrictions to the edges are infinitely differentiable, and $ \mathcal{S}_c(\Gamma) $ is a subspace defined as:
\begin{equation}
    \label{subspaceSc}
    \mathcal{S}_c(\Gamma) = \{f \in C(\Gamma) : \exists v \in \mathcal{V}, f \in \mathcal{S}_c(v, \Gamma) \cap \mathcal{N}(v, \Gamma)\},
\end{equation}
where $ \mathcal{S}_c(v, \Gamma) $ contains functions compactly supported within the star graph induced by the vertex $ v $ with the outer vertices removed, that is, supported on
\begin{align*}
    S(v, \Gamma) = \left\{ s\in \Gamma\;:\; s = (t,e), e\in\mathcal{E}_v, e = [0,\ell_e], t\in (0,\ell_e)\right\}\cup \{v\},
\end{align*}
and 
\begin{align*}
    \mathcal{N}(v, \Gamma) = \left\{ f \in C(\Gamma) \; : \; \forall e \in \mathcal{E}_v, f|_e \in C^\infty(e), \; f_e^{(k)}(v)=0, \; \forall k \in \mathbb{N}
\right\}.
\end{align*}
We note that although the subspace $\mathcal{S}_c(\Gamma)$ defined in \eqref{subspaceSc} differs slightly from that considered in the proof of \citet[Proposition 1]{Bolin2023Statistical}, the argument establishing the existence of nontrivial elements in that space remains valid in the present setting.

Let $\alpha\in(\sfrac{1}{2},2]$. Then it immediately follows from \citet[Theorem 4.1]{Bolin2024Regularity} that $\mathcal{A}(\Gamma) \subset \dot{H}^\alpha$. Let $\mathbb{A} = (2,\infty)\setminus\{k+\sfrac{1}{2}:k\in\mathbb{N}\setminus\{1\}\}$ and 
\begin{equation}
\label{eq:E_alpha_space}
E^\alpha=  \left\{f \in \bigoplus_{e \in \mathcal{E}} H^\alpha(e): L^{m_0} f \in C(\Gamma)\text{ and }\forall m \in\left\{0, \ldots,m_1\right\}, L^m f \in \mathcal{K}(\Gamma)\right\},
\end{equation}
where $m_0= \left\lfloor\sfrac{\alpha}{2}-\sfrac{1}{4}\right\rfloor$ and $m_1 = \left\lfloor\sfrac{\alpha}{2}-\sfrac{3}{4}\right\rfloor$, $\mathcal{K}(\Gamma)$  represents the Kirchhoff vertex conditions \eqref{eq:kirchhoff_cond}, and $H^\alpha(e)$ (see \citet[Definition 4]{Awadelkarim2025Fractional}) denotes the fractional Sobolev space of order $\alpha$ on the interval corresponding to the edge $e$. The space $E^\alpha$ is endowed with the norm $\|\cdot\|_{\bigoplus_{e \in \mathcal{E}} H^\alpha(e)}$. Let $\alpha\in \mathbb{A}$. Then we can use the characterization $E^\alpha\cong \dot{H}^\alpha$ in \citet[Theorem 14]{Awadelkarim2025Fractional} to show that $\mathcal{A}(\Gamma) \subset \dot{H}^\alpha$ as follows. Let $h\in \mathcal{A}(\Gamma)$. Then $h = c+f+g$, where $c\in\mathbb{R}$, $f\in D(\Gamma)$, and $g\in \mathcal{S}_c(\Gamma)$. Our goal is to prove that $h\in E^\alpha$, for which it suffices to show that $c\in E^\alpha$, $f\in E^\alpha$, and $g\in E^\alpha$ since $E^\alpha$ is a vector space. We proceed by treating each term separately.

\textbf{Step 1} (Proof that $c\in E^\alpha$): Clearly, $c\in \bigoplus_{e \in \mathcal{E}} H^\alpha(e)$ and $\Delta_\Gamma c = 0$, so $Lc = \kappa^2 c$ and thus $L^n c = c\, L^{n-1}(\kappa^2)$ for all $n=1,2,\ldots$. If $\alpha<2.5$, then $m_0=0$ and $m_1=0$, so the definition of $E^\alpha$ only requires $c\in C(\Gamma)\cap\mathcal{K}(\Gamma)$, which holds trivially. Now assume $\alpha>2.5$, so that $m_0\geq 1$. By Assumption~\ref{assumption2}, we have $L^{m_0-1}\kappa^2\in C(\Gamma)$, whence $L^{m_0}c = c\,L^{m_0-1}(\kappa^2)\in C(\Gamma)$. If additionally $\alpha>3.5$, then $m_1\geq 1$ and Assumption~\ref{assumption2} further yields $L^{m-1}\kappa^2\in\mathcal{K}(\Gamma)$ for $m=1,\ldots,m_1$. Hence, $L^{m}c = c\,L^{m-1}(\kappa^2)\in\mathcal{K}(\Gamma)$ for these $m$, while $c\in\mathcal{K}(\Gamma)$ for $m=0$ holds trivially. Thus $L^{m}c\in\mathcal{K}(\Gamma)$ for all $m\in\{0,\ldots,m_1\}$, and consequently $c\in E^\alpha$.

\begin{figure}[t]
\centering
\includegraphics[width=0.99\textwidth]{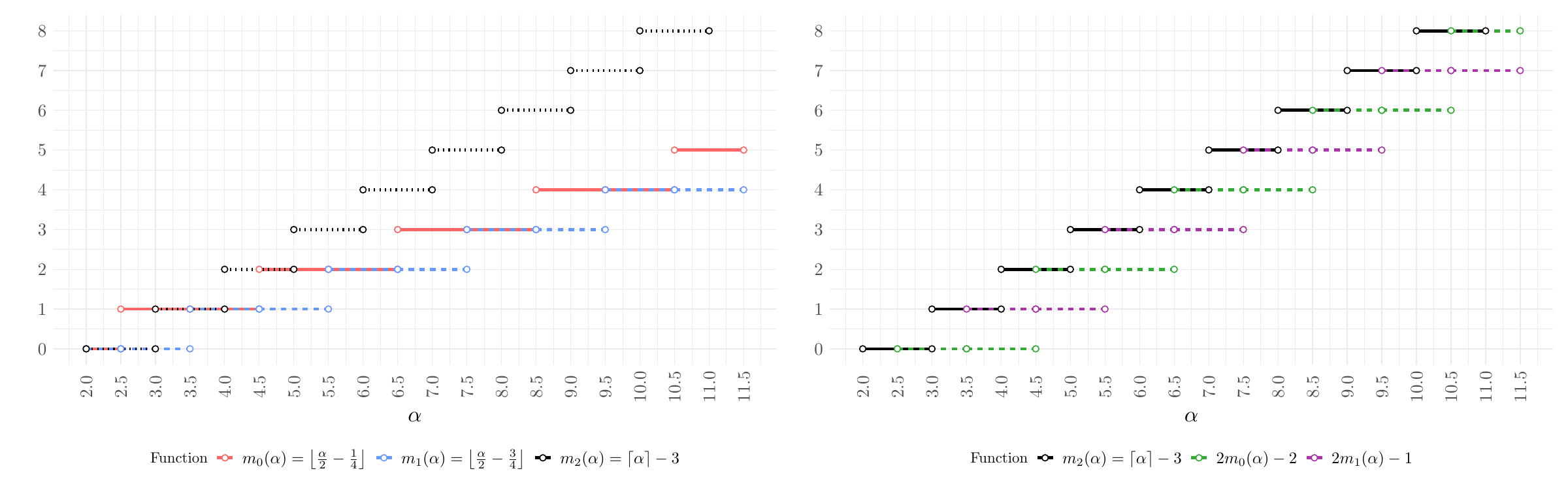}
\caption{Comparison of $m_0(\alpha) = \left\lfloor \sfrac{\alpha}{2} - \sfrac{1}{4} \right\rfloor$, 
$m_1(\alpha) = \left\lfloor \sfrac{\alpha}{2} - \sfrac{3}{4} \right\rfloor$, and 
$m_2(\alpha) = \lceil \alpha \rceil - 3$ (left) and comparison of $m_2(\alpha)$, $2m_0(\alpha)-2$, and $2m_1(\alpha)-1$ (right) as functions of $\alpha$.}
\label{mplots}
\end{figure}

\textbf{Step 2} (Proof that $f\in E^\alpha$):
Since $f|_e \in C_c^\infty(e)$ and $C_c^\infty(e)\subset H^\alpha(e)$, we immediately have that $f \in \bigoplus_{e \in \mathcal{E}} H^\alpha(e)$. Let $m=0$. Then $L^mf=f \in \mathcal{K}(\Gamma)$, trivially, since $f$ vanishes near vertices. Let $m=1$. Then on each edge, $L^mf = \kappa^2f-f^{''}\in  \mathcal{K}(\Gamma)$ since the Kirchhoff conditions are stable under multiplication of continuous functions (see Lemma~\ref{lem:kirchhof_product_stable}). Let $m=2$. Then on each edge, 
\begin{align*}
    L^mf = L(\kappa^2f-f^{''}) = \kappa^4f-2\kappa^2f''-2(\kappa'\kappa'+\kappa\kappa'')f-4\kappa\kappa'f'+f^{(4)}.
\end{align*}
Since $f|_e \in C_c^\infty(e)$ and $\kappa_e\in C^{2m-2,1}(e)$ for all $e\in\mathcal{E}$ we have $(L^mf)(v) = 0$ for all $v\in\mathcal{V}$, and consequently $L^mf\in C(\Gamma)$. The assumption $\kappa_e\in C^{2m-1,1}(e)$ for all $e\in \mathcal{E}$ implies that $\partial_e (L^mf)(v)$ is well-defined for all $v\in\mathcal{V}$. Again, $f|_e \in C_c^\infty(e)$ for all $e\in\mathcal{E}$ implies that $\partial_e (L^mf)(v)=0$ for all $v\in\mathcal{V}$, which in turn implies that $L^mf\in \mathcal{K}(\Gamma)$. Let $m_2 = \lceil \alpha \rceil - 3$. From the preceding discussion, we see that to ensure $L^{m_0} f \in C(\Gamma)$ it is sufficient that $\kappa_e\in C^{2m_0-2,1}(e)$ for all $e\in \mathcal{E}$. Moreover, to guarantee that $ L^m f \in \mathcal{K}(\Gamma)$ for all $m \in\left\{0, \ldots,m_1\right\}$ it is sufficient to require $\kappa_e\in C^{2m_1-1,1}(e)$ for all $e\in \mathcal{E}$. Both conditions are implied by Assumption~\ref{assumption1}, since $2m_0-2\leq m_2$ and $2m_1-1\leq m_2$, as illustrated in Figure~\ref{mplots}. Suppose now $L^mf \in \mathcal{K}(\Gamma)$. Then $L^{m+1}f=L(L^mf) = \kappa^2L^mf-(L^mf)^{''}$. Clearly, $\kappa^2L^mf\in  \mathcal{K}(\Gamma)$. The term $(L^mf)^{''}$ contains up to $2(m+1)-2$ derivatives of $\kappa$ and since $\kappa_e\in C^{2(m+1)-1,1}(e)$ for all $e\in \mathcal{E}$, $\partial_e(L^mf)^{''}(v)$ is well-defined for all $v\in\mathcal{V}$. Because of the compact support assumption on edges of $f$, $(L^mf)^{''}\in  \mathcal{K}(\Gamma)$. We conclude that $L^{m+1}f \in  \mathcal{K}(\Gamma)$. By induction, we have shown that $L^{m}f \in  \mathcal{K}(\Gamma)$ for all $m\in\mathbb{N}_0$, and consequently, $f\in E^\alpha$.

\textbf{Step 3} (Proof that $g\in E^\alpha$):
Since $g\in \mathcal{S}_c(\Gamma)$, there exists $v\in\mathcal{V}$ for which $g$ is compactly supported within the star graph induced by $v$, and for all $e\in\mathcal{E}_v$ we have that $g|_e \in C^\infty(e)$ and $ g_e^{(k)}(v)=0$, $\forall k \in \mathbb{N}$. Again, the vanishing derivative condition in $\mathcal{N}(v, \Gamma)$ and the smoothness condition $g|_e \in C^\infty(e)$ for all edges 
imply that $g\in E$.

Steps 1, 2, and 3 show that $h\in E^\alpha$, and by the characterization $E^\alpha\cong \dot{H}^\alpha$ we conclude that $\mathcal{A}(\Gamma) \subset \dot{H}^\alpha$ for $\alpha\in\mathbb{A}$. If $\alpha\in\{k+\sfrac{1}{2}:k\in\mathbb{N}\setminus\{1\}\}$, then we have that $\mathcal{A}(\Gamma) \subset \dot{H}^{\lceil\alpha\rceil}(\Gamma)$ as $\lceil\alpha\rceil\in\mathbb{A}$.  This and the fact that $\dot{H}^{\lceil\alpha\rceil}\subset  \dot{H}^\alpha$ imply that $\mathcal{A}(\Gamma) \subset \dot{H}^\alpha$. Altogether, we have shown that $\mathcal{A}(\Gamma) \subset \dot{H}^\alpha$ for $\alpha\in(\sfrac{1}{2},\infty)$. Now, to apply Stone--Weierstrass theorem, we must show that $\mathcal{A}(\Gamma)$ is a subalgebra of $C(\Gamma)$ that contains the constant functions and separates points of $\Gamma$. The proof that $\mathcal{A}(\Gamma)$ is a subalgebra of $C(\Gamma)$ and separates points is identical to the one in \cite[Lemma 2]{Bolin2023Statistical}, whereas we showed above that $\mathcal{A}(\Gamma)$ contains the constant function. Thus, we have that $\mathcal{A}(\Gamma)$ is dense in $C(\Gamma)$ with respect to the $\|\cdot\|_{C(\Gamma)}$ norm. Therefore, $\dot{H}^\alpha$ is dense in $C(\Gamma)$ with respect to the $\|\cdot\|_{C(\Gamma)}$ norm, which, by \citet[Proposition 10 and Remark~4]{Bolin2023Statistical}, implies that $\rho$ is strictly positive-definite. \hfill
\end{proof}

We are now ready to establish the positive definiteness statement in Proposition~\ref{prp:existence}.

\begin{proof}[Proof of Proposition~\ref{prp:existence} (positive definiteness)]
The strict positive definiteness follows directly from the assumptions on $\tau$ and Proposition~\ref{prp:positive_definite_cov}. Specifically, let $w$ be the solution to \eqref{eq:auxiliary_problem}, and note that $u = \tau^{-1}w$. The covariance function of $u$, denoted by $\varrho(\cdot, \cdot)$, is 
\begin{align*}
    \varrho(s, t) = \mathbb{E}[u(s)u(t)] = \mathbb{E}[\tau^{-1}(s) w(s) w(t) \tau^{-1}(t)] = \tau^{-1}(s) \rho(s, t) \tau^{-1}(t), \quad s, t \in \Gamma.
\end{align*}
By Proposition~\ref{prp:positive_definite_cov}, the function $\rho(\cdot, \cdot)$ is strictly positive-definite. Thus, for any $n \in \mathbb{N}$, any nonzero vector $\boldsymbol{a} = (a_1, \ldots, a_n) \in \mathbb{R}^n$, and any points $s_1, \ldots, s_n \in \Gamma$, we have:
\begin{align*}
    \sum_{i,j=1}^n a_i \rho(s_i, s_j) a_j > 0.
\end{align*}
Since $\tau^{-1}$ is strictly positive by assumption, the vector $\boldsymbol{b} = (a_1 \tau^{-1}(s_1), \ldots, a_n \tau^{-1}(s_n))$ is also nonzero. Thus, we can write:
\begin{align*}
    0 < \sum_{i,j=1}^n b_i \rho(s_i, s_j) b_j = \sum_{i,j=1}^n a_i \tau^{-1}(s_i) \rho(s_i, s_j) \tau^{-1}(s_j) a_j = \sum_{i,j=1}^n a_i \varrho(s_i, s_j) a_j.
\end{align*}
This shows that $\varrho(\cdot, \cdot)$ is strictly positive-definite, completing the proof. \hfill
\end{proof}

\section{The role of $\tau(\cdot)$ in the regularity of the field}
\phantomsection 

In this section, we provide proofs for Propositions~\ref{prp:regularity} and~\ref{prp:regularity-global}. We begin with the proof of statement (i) of Proposition~\ref{prp:regularity} and then proceed to prove statement (i) of Proposition~\ref{prp:regularity-global}. Our arguments rely on the auxiliary problem \eqref{eq:auxiliary_problem} and two preparatory lemmata. Lemma~\ref{lem:lem0003} provides Sobolev--H\"older embeddings on metric graphs, while Lemma~\ref{lem:lem99} gives H\"older regularity of the auxiliary solution.

\begin{lemma}
    \label{lem:lem0003}
    Let $\alpha\in(\sfrac{1}{2},\infty)$. If $\alpha\in(\sfrac{1}{2},\sfrac{3}{2}]$, then $\dot{H}^\alpha\hookrightarrow C^{0,\gamma}(\Gamma)$ for any $0<\gamma<\alpha-\sfrac{1}{2}$. If $\alpha\in(\sfrac{3}{2},\infty)$, then $\dot{H}^\alpha\hookrightarrow C^{0,\gamma}(\Gamma)$ for any $0<\gamma\leq 1$.
\end{lemma}

\begin{proof}
    Recall from \citet[Proposition 5.2]{Bolin2024Regularity} that $\dot{H}^\alpha \cong H^\alpha(\Gamma)$ whenever $\alpha\in(\sfrac{1}{2},2]$. Assumption~\ref{assumption1} allows to apply \citet[Theorem 6]{Awadelkarim2025Fractional}, which together 
    with the previous identification imply that $\dot{H}^\alpha\hookrightarrow C^{0,\gamma}(\Gamma)$ for any $0<\gamma<\alpha-\sfrac{1}{2}$ if $\alpha\in(\sfrac{1}{2},\sfrac{3}{2}]$. This proves the first part. \citet[Theorem 6]{Awadelkarim2025Fractional} also implies that $\dot{H}^\alpha\hookrightarrow C^{0,\gamma}(\Gamma)$ for any $0<\gamma\leq 1$ if $\alpha > 3/2$. 
    \hfill
\end{proof}
The following lemma is an adaptation of \cite[Lemma 6]{Bolin2024Regularity}.
\begin{lemma}
\label{lem:lem99}
Let $\alpha>\sfrac{1}{2}$, and let $w$ be the solution of \eqref{eq:auxiliary_problem}. Then $w$ admits a modification with $\gamma$-H\"older continuous sample paths for any $\gamma$ satisfying
\begin{align*}
    \begin{cases}
0<\gamma<\alpha-\sfrac{1}{2}, & \text{if } \alpha\in(\sfrac{1}{2},\sfrac{3}{2}],\\
0<\gamma\leq 1, & \text{if } \alpha\in(\sfrac{3}{2},\infty).
\end{cases}
\end{align*}
\end{lemma}
\begin{proof}
    As in the proof of \cite[Lemma 6]{Bolin2024Regularity}, we only need to stablish that $L^{-\sfrac{\alpha}{2}}:L_2(\Gamma)\to C^{0,\gamma}(\Gamma)$ is a bounded linear operator. This follows by combining Lemma~\ref{lem:lem0003} and the fact that $L^{-\sfrac{\alpha}{2}}:L_2(\Gamma)\to\dot{H}^\alpha$ is a bounded linear operator. \hfill
\end{proof}

\begin{proof}[Proof of statement (i) of Proposition~\ref{prp:regularity}]
    Let $\alpha\in(\sfrac{1}{2},\sfrac{3}{2})$, and let $w$ be a solution to \eqref{eq:auxiliary_problem}. Observe that $ u = \tau^{-1} w $ is a solution to \eqref{eq:spde}. By  Lemma~\ref{lem:lem99}, $w$ has a modification with $\gamma$-H\"older continuous sample paths on $\Gamma$ for $0<\gamma<\alpha-\sfrac{1}{2}$. For the remaining of the proof we will assume we are considering such a modification. This implies that for every $ e \in \mathcal{E} $, $ w_e $ has $\gamma$-H\"older continuous sample paths. 

    Under Assumption~\ref{assumption1}, note that $\tau^{-1}$ is bounded and bounded away from zero. Moreover, since $ w $ has continuous sample paths, these sample paths are bounded. Therefore, by Proposition~\ref{prop:Holder_quotient}, it follows that, for every $ e \in \mathcal{E} $, the sample paths of $ u_e = \tau^{-1} w_e $ are $\gamma$-H\"older continuous. \hfill
\end{proof}

\begin{proof}[Proof of statement (i) of Proposition~\ref{prp:regularity-global}]
    Let $\alpha>\sfrac{1}{2}$ and let $ w $ be as in the proof of statement (i) of Proposition~\ref{prp:regularity}, so that $ w $ has a modification with $\gamma$-H\"older continuous sample paths on $\Gamma$.  Assume that $ \tau \in C(\Gamma) $ and $ \tau_e \in C^{0,\gamma}(e) $ for all $ e \in \mathcal{E} $. Since $ w $ has a modification with $\gamma$-H\"older continuous sample paths, if $ \tau $ is continuous on $\Gamma$, then $ u = \tau^{-1} w $ has a modification with continuous sample paths. Conversely, assume that $ \tau_e \in C^{0,\gamma}(e) $ for all $ e \in \mathcal{E} $, and that $ u $ has a modification with continuous sample paths. Then, with probability 1, for every $ v \in \mathcal{V} $, we have
    $u_e(v) = u_{e'}(v)$, for all $e, e' \in \mathcal{E}_v$.
    Substituting $ u = \tau^{-1} w $, we thus have 
    $\tau_e^{-1}(v) w_e(v) = \tau_{e'}^{-1}(v) w_{e'}(v)$, for all $e, e' \in \mathcal{E}_v$.
    Since $ w $ has almost surely continuous sample paths, it follows that
    $w_e(v) = w_{e'}(v)$, for all $e, e' \in \mathcal{E}_v$.
    Furthermore, by Proposition~\ref{prp:positive_definite_cov}, $ w $ has a strictly positive-definite covariance function. Thus, for every $ v \in \mathcal{V} $ and every $ e \in \mathcal{E}_v $, we have $ \mathbb{P}(w_e(v) \neq 0) = 1 $. Consequently, 
    \begin{equation}
    \label{eq:tau_continuity}
        \tau_e^{-1}(v) = \tau_{e'}^{-1}(v), \quad \text{for all } e, e' \in \mathcal{E}_v.
    \end{equation}
    Because of the continuity of $ \tau_e $ on each edge $ e \in \mathcal{E} $, equation \eqref{eq:tau_continuity} ensures the continuity of $ \tau $ at all vertices $ v \in \mathcal{V} $. Therefore, $ \tau $ is continuous on $\Gamma$.

    Finally, as the sample paths of $ w $ are bounded,  Proposition~\ref{prop:Holder_quotient} implies that if ${\tau \in C^{0,\gamma}(\Gamma)}$ and $ \tau $ satisfies Assumption~\ref{assumption1}, then the sample paths of $ u = \tau^{-1} w $ are $\gamma$-H\"older continuous. \hfill
\end{proof}

Our goal now is to prove statement (ii) from Proposition~\ref{prp:regularity} and statement (ii) from Proposition~\ref{prp:regularity-global}. To such an end, we will need several auxiliary results. Let us first prove the following lemma, which is an extension of \citet[Lemma~1]{Bolin2024Regularity}.

\begin{lemma}
\label{lem:regularity_field_kolm_chent}
Fix $\alpha>\sfrac{1}{2}$ and let $\tilde{\alpha} = \min\{\alpha-\sfrac{1}{2}, 1\}$. Further, let $\widetilde{\Gamma}\subset \Gamma$ be a compact and connected metric graph contained in $\Gamma$. Assume that $T_\alpha : L_2(\Gamma) \to \dot{H}^{\alpha}(\widetilde{\Gamma})$ is a bounded linear operator. For a location $s\in\Gamma$, define $u_0 := \mathcal{W}(T_\alpha^\ast \delta_s)$, where $\delta_s$ is the Dirac measure concentrated at $s$ and $\mathcal{W}$ is Gaussian white noise on $L_2(\Gamma)$. Then,
$$\mathbb{E}\left( |u_0(s) - u_0(s')|^2 \right) \leq \|T_\alpha\|_{\mathcal{L}(L_2(\Gamma), C^{0,\tilde{\alpha}}(\widetilde{\Gamma}))}^{2} d(s,s')^{2\tilde{\alpha}}.$$
Thus, given $0<\gamma<\tilde{\alpha}$, $u_0$ has a modification with $\gamma$-H\"older continuous sample paths.
\end{lemma}

\begin{proof}
The proof is essentially the same as the proof of \citet[Lemma 1]{Bolin2024Gaussian}. Therefore, we will only provide the parts in which the proof differs. Begin by observing that by Lemma~\ref{lem:lem0003}, we have that ${T_\alpha : L_2(\Gamma) \to \dot{H}^{\alpha}(\widetilde{\Gamma}) \hookrightarrow C^{0,\tilde{\alpha}}(\widetilde{\Gamma})}$. Thus, $T_\alpha:L_2(\Gamma) \to C^{0,\tilde{\alpha}}(\widetilde{\Gamma})$ is a bounded operator. Let $T_\alpha^\ast : (C^{0,\tilde{\alpha}}(\widetilde{\Gamma}))^\ast \to (L_2(\Gamma))^\ast = L_2(\Gamma)$ be the adjoint of $T_\alpha$. Further, observe that $\delta_s \in (C^{0,\tilde{\alpha}}(\widetilde{\Gamma}))^*$, which yields $T_\alpha^\ast(\delta_s) \in L_2(\Gamma)$. Therefore, $u_0 := \mathcal{W}(T_\alpha^\ast (\delta_s))$, is well-defined. The remaining proof is the same as the proof of \citet[Lemma 1]{Bolin2024Gaussian}, with the replacement of $L^{-\sfrac{\alpha}{2}}$ by $T_\alpha$. \hfill
\end{proof}

The next lemma is an extension of \citet[Lemma 2]{Bolin2024Gaussian} and shows that $u_0$ defined in Lemma~\ref{lem:regularity_field_kolm_chent} is a centered Gaussian random field with covariance operator $T_\alpha T_\alpha^\ast$.

\begin{lemma}
\label{lem:regularity_field_kolm_chent_2}
    Under the assumptions of Lemma~\ref{lem:regularity_field_kolm_chent}, fix $0<\gamma<\tilde{\alpha}$ and let $u_1$ be any $\gamma$-H\"older continuous modification of $u_0$. Then, for every $h\in L_2(\widetilde{\Gamma})$,
    $(u_1, h)_{L_2(\widetilde{\Gamma})} = \mathcal{W}(T_\alpha^\ast h).$
\end{lemma}

\begin{proof}
    The proof follows the same steps as in \citet[Lemma 2]{Bolin2024Gaussian}, with the substitution of $ L^{-\sfrac{\alpha}{2}} $ by $ T_\alpha $. All arguments remain identical under this replacement. \hfill
\end{proof}

We now prove some additional results regarding the derivatives of the solution to the auxiliary problem \eqref{eq:auxiliary_problem}, when $\alpha > \sfrac{3}{2}$. To this end, we first recall some results from \cite{Bolin2024Regularity}. The following lemma is a special case of \citet[Lemma 5.5]{Bolin2024Regularity}:

\begin{lemma}
\label{lem:lemma_restriction_bounded}
    For any $0\leq s \leq 1$ and any edge $e\in\mathcal{E}$, the restriction operator $R_{e,s}$ is a bounded operator from $H^s(\Gamma)$ to $H^s(e)$ and from $\widetilde{H}^{s}(\Gamma)$ to $H^{s}(e)$.
\end{lemma}

To keep the notation simple, we will denote $R_{e,s}$ (i.e., the restriction operator from $H^s(\Gamma)$ to $H^s(e)$) by $R_{e}$. 
The next lemma is a special case of \citet[Lemma 5.4]{Bolin2024Regularity} when $\Gamma = e$:

\begin{lemma}
\label{lem:lemma_derivative_bounded}
    Let $1 < s < 2$. Then, for any $e\in\mathcal{E}$, with the identification $e = [0,\ell_e]$, the derivative operator $D$ such that $(D f_e)(x) = f_e'(x)$, for $f_e\in H^s(e)$, is a bounded operator from $H^s(e)$ to $H^{s-1}(e)$.
\end{lemma}

From the proof of \cite[Lemma 5.9]{Bolin2024Regularity}, we also have the following:
\begin{lemma}
\label{lem:lemma_secondderivative_bounded}
    Let $s\geq 2$. Then, for any $e\in\mathcal{E}$, with the identification $e = [0,\ell_e]$, the second derivative operator $D^2$ such that $(D^2 f_e)(x) = f_e''(x)$, for $f_e\in H^s(e)$, is a bounded operator from $H^s(e)$ to $H^{s-2}(e)$.
\end{lemma}

The following results are novel, even in the context of standard Whittle--Mat\'ern fields, where $\kappa$ is constant in the auxiliary problem.

\begin{prop}
\label{prp:regularity_derivatives}
    Let $\alpha > \sfrac{3}{2}$ and let $w$ be the solution to \eqref{eq:auxiliary_problem}. Then, for any $e\in\mathcal{E}$, 
    and $t,t'\in [0,\ell_e]$,
    \begin{equation}
    \label{eq:bound_deriv_holder}
        \mathbb{E}\left( |w_e'(t) - w_e'(t')|^2 \right) \leq \|D R_e L^{-\sfrac{\alpha}{2}}\|_{\mathcal{L}(L_2(\Gamma), C^{0,\alpha-\sfrac{3}{2}}(e))} |t-t'|^{2\alpha-3}, \quad t,t'\in e.
    \end{equation}
    Further, for any $0 < \gamma < \min\{\alpha-\sfrac{3}{2}, 1\}$, $w_e'$ has a modification with $\gamma$-H\"older continuous sample paths.
\end{prop}

\begin{proof}
    First, let $T_\alpha = D R_e L^{-\sfrac{\alpha}{2}}$. Then, $T_\alpha$ is a bounded operator from $L_2(\Gamma)$ to $C^{0,\alpha-\sfrac{3}{2}}(e)$. Indeed, $L^{-\sfrac{\alpha}{2}}$ is a bounded operator from $L_2(\Gamma)$ to $\dot{H}^\alpha$. By \citet[Theorem~4.1]{Bolin2024Regularity}, we have that $\dot{H}^\alpha\hookrightarrow \widetilde{H}^\alpha(\Gamma)$. Now, by Lemma~\ref{lem:lemma_restriction_bounded}, $R_e L^{-\sfrac{\alpha}{2}}$ is a bounded operator from $L_2(\Gamma)$ to $H^\alpha(e)$. Next, by Lemma~\ref{lem:lemma_derivative_bounded}, $D$ is a bounded operator from $H^\alpha(e)$ to $H^{\alpha-1}(e)$. Finally, by the Sobolev embedding \cite[Theorem 6]{Awadelkarim2025Fractional}, we have ${H^{\alpha-1}(e)\hookrightarrow C^{0,\gamma}(e)}$, where $0<\gamma<\min\{\alpha-\sfrac{3}{2},1\}$. Therefore, $T_\alpha$ is a bounded operator from $L_2(\Gamma)$ to $C^{0,\gamma}(e)$.
    
    Let $u_0 := \mathcal{W}(T_\alpha^\ast \delta_s)$, where $\delta_s$ is the Dirac measure concentrated at $s$. By Lemma~\ref{lem:regularity_field_kolm_chent}, for any $0<\gamma<\min\{\alpha-\sfrac{3}{2}, 1\}$, $u_0$ has a modification with $\gamma$-H\"older continuous sample paths. Next, by Lemma~\ref{lem:regularity_field_kolm_chent_2}, $u_0$  is a centered Gaussian random field with covariance operator $T_\alpha T_\alpha^\ast = D R_e L^{-\alpha} (D R_e)^\ast$.

    Now, let $w$ be the solution to \eqref{eq:auxiliary_problem}, so that $w$ has covariance operator $L^{-\alpha}$ and by \citet[Theorem 4.4]{Bolin2024Regularity}, the sample paths of $w$ belong to $H^1(\Gamma)$, so that the sample paths of $w$ are weakly differentiable. Then, $w_e' = D R_e w$, the derivative of $w$ restricted to $e$, is a centered Gaussian random field with covariance operator $D R_e L^{-\alpha} (D R_e)^\ast$. Therefore, by the above, $w_e'$ and $u_0$ have the same finite-dimensional distributions, which implies \eqref{eq:bound_deriv_holder} and also  that $w'_e$ has a modification with $\gamma$-H\"older continuous sample paths. \hfill
\end{proof}

We are now in a position to prove the statement (ii) from Proposition~\ref{prp:regularity}. 

\begin{proof}[Proof of statement (ii) of  Proposition~\ref{prp:regularity}]
    Let $w$ be the solution to \eqref{eq:auxiliary_problem}. Since $\alpha>\sfrac{3}{2}$, it follows from \citet[Theorem 4.4]{Bolin2024Regularity} that the sample paths of $w$ belong to $H^1(\Gamma)$. Furthermore, by Proposition~\ref{prp:regularity_derivatives}, $w_e'$ has a modification with $\gamma$-H\"older continuous sample paths for any $0<\gamma<\min\{\alpha-\sfrac{3}{2}, 1\}$. Fix $0 < \gamma < \min\{\alpha-\sfrac{3}{2}, 1\}$ and let $w$ be a modification of $w$ such that for every $e\in\mathcal{E}$, $w_e'$ has $\gamma$-H\"older continuous sample paths. Since the sample paths of $w$ belong to $H^1(\Gamma)$, it follows from \eqref{eq:weak_deriv} that for every $t \in e$,
    \begin{align*}
        w_e(t) = \int_0^t w_e'(s)ds + w_e(0).
    \end{align*}
    Since, with probability 1, $w_e'$ is continuous, it follows that $w_e$ is continuously differentiable. Finally, since $w_e'$ has $\gamma$-H\"older continuous sample paths, it follows that the sample paths of $w_e$ belong to $C^{1,\gamma}(e)$.

    Finally, if $\tau_e \in C^{1,\gamma}(e)$, and $\tau\geq \tau_0>0$, then $u_e = \tau_e^{-1} w_e$ is continuously differentiable. Further,
    $u_e' = \tau_e^{-1} w_e' - \tau_e^{-2} w_e \tau_e',$
    which is $\gamma$-H\"older continuous by Propositions~\ref{prop:Holder_exponents}, \ref{prop:Holder_sum}, \ref{prop:Holder_product}, and~\ref{prop:Holder_quotient}. Therefore, since $e$ is compact, $u$ has a modification such that for all $e\in\mathcal{E}$, $u_e$ has sample paths belonging to $C^{1,\gamma}(e)$. \hfill
\end{proof}

Next, we present a novel result showing that for $\alpha > \sfrac{3}{2}$, the solutions to \eqref{eq:auxiliary_problem} satisfy the Kirchhoff vertex conditions. This result improves upon \citet[Proposition~11]{Bolin2024Gaussian}, where it was established only for $\alpha \geq 2$ and limited to standard Whittle--Mat\'ern fields, i.e., the case where $\kappa$ is constant.

\begin{prop}
\label{prp:kirchhoff_vertex_conditions_solution_auxiliary}
    Let $\alpha > \sfrac{3}{2}$, fix some $0 < \gamma < \min\{\alpha-\sfrac{3}{2}, 1\}$, and let $w$ be the solution to \eqref{eq:auxiliary_problem}. By Proposition~\ref{prp:regularity_derivatives}, $w$ admits a modification such that, for every $e \in \mathcal{E}$, the sample paths of $w_e$ belong to $C^{1,\gamma}(e)$. This modification of $w$ satisfies the Kirchhoff vertex conditions \eqref{eq:kirchhoff_cond}; specifically, the sample paths of $w$ are continuous, and for all $v \in \mathcal{V}$, $\sum_{e \in \mathcal{E}_v} \partial_e w(v) = 0$.
\end{prop}

\begin{proof}
    The continuity of $w$ follows from \citet[Theorem 4.4]{Bolin2024Regularity}. Furthermore, by \citet[p. 44]{Bogachev1998Gaussian}, the Cameron-Martin space associated with $w$ is given by $\dot{H}^\alpha$. By \citet[Theorem 4.1]{Bolin2024Regularity}, we also have $\dot{H}^\alpha \cong \widetilde{H}^\alpha \cap C(\Gamma) \cap K_\alpha(\Gamma)$, where
    \begin{equation}
    \label{eq:K_alpha_definition}
    K_\alpha(\Gamma) = \left\{f \in \widetilde{H}^\alpha(\Gamma) : \forall v \in \mathcal{V}, \; \sum_{e \in \mathcal{E}_v} \partial_e f(v) = 0 \right\}.
    \end{equation}
    Therefore, by the definition of the Cameron-Martin space, along with the above identification, we have $\{h(s) = \mathbb{E}(w(s)g) : s \in \Gamma, \; g \in H_w\} = \widetilde{H}^\alpha \cap C(\Gamma) \cap K_\alpha(\Gamma)$. From Lemma~\ref{lem:lemma_restriction_bounded}, for every $h \in \widetilde{H}^\alpha$, $h|_e = R_e(h) \in H^\alpha(e)$. By the same arguments as in the proof of Proposition~\ref{prp:regularity_derivatives}, we have $h_e \in C^{1, \gamma}(e)$. 

    Given any $g \in H_w$, define $h(s) = \mathbb{E}(w(s)g)$. For every edge $e \in \mathcal{E}$, $h_e = h|_e$ is continuously differentiable on $[0, \ell_e]$. Hence, for any $t \in [0, \ell_e]$ and any sequence $t_n \to t$, $t_n \in [0, \ell_e]$ with $t_n \neq t$, the sequence 
    \begin{align*}
        \frac{w_e(t_n) - w_e(t)}{t_n - t}
    \end{align*}
    is weakly Cauchy in $H_w$. Since $H_w$ is a Hilbert space, it is weakly sequentially complete. Thus, there exists $w'_e(t) \in H_w$ such that
    \begin{align*}
        \frac{w_e(t_n) - w_e(t)}{t_n - t} \xrightarrow{w} w'_e(t),
    \end{align*}
    where $\xrightarrow{w}$ denotes weak convergence in $H_w$. Since the sample paths of $w_e$ belong to $C^{1, \gamma}(e)$, the limit $w'_e(t)$ coincides with the classical derivative of $w_e$.

    Because $h_e(t) = \mathbb{E}(w_e(t)g)$ is differentiable, the limit $w'_e(t)$ does not depend on the choice of the sequence $t_n \to t$. Moreover, the weak convergence implies $h_e'(t) = \mathbb{E}(w'_e(t)g)$, where we take lateral derivatives if $t$ lies on the boundary of $[0, \ell_e]$. This proves the existence of the weak derivative in the $L_2(\Omega)$-sense.

    Since $h \in \mathcal{H}_w$, we have $h \in K_\alpha(\Gamma)$. Therefore, for every $v \in \mathcal{V}$,
    \begin{align*}
        0 = \sum_{e \in \mathcal{E}_v} \partial_e h(v) = \sum_{e \in \mathcal{E}_v} \partial_e \mathbb{E}(w_e(v)g) = \sum_{e \in \mathcal{E}_v} \mathbb{E}(\partial_e w_e(v)g) = \mathbb{E}\Bigl[\Bigl(\sum_{e \in \mathcal{E}_v} \partial_e w_e(v)\Bigr) g\Bigr],
    \end{align*}
    where, for $e = [0, \ell_e]$, $\partial_e w_e(v) = w_e'(0)$ if $v = 0$, and $\partial_e w_e(v) = -w_e'(\ell_e)$ if $v = \ell_e$.

    Finally, take $g = \sum_{e \in \mathcal{E}_v} \partial_e w_e(v)$, which belongs to $H_w$, since, as shown above, $w_e' \in H_w$ for every $e \in \mathcal{E}$. Then, $\mathbb{E}\Bigl[\Bigl(\sum_{e \in \mathcal{E}_v} \partial_e w_e(v)\Bigr)^2\Bigr] = 0$, which implies $\sum_{e \in \mathcal{E}_v} \partial_e w_e(v) = 0$, completing the proof. \hfill
\end{proof}

We are now in a position to prove  statement (ii) of Proposition~\ref{prp:regularity-global}.

\begin{proof}[Proof of statement (ii) of Proposition~\ref{prp:regularity-global}]
Let $w$ be the solution to \eqref{eq:auxiliary_problem} so that $ w $ has a modification with $\gamma$-H\"older continuous sample paths on $\Gamma$. By Proposition~\ref{prp:kirchhoff_vertex_conditions_solution_auxiliary}, $w$ satisfies the Kirchhoff vertex conditions \eqref{eq:kirchhoff_cond}. Assume $\tau\in\mathcal{K}(\Gamma)$. Then it immediately follows that $u\in\mathcal{K}(\Gamma)$ from the fact that $u = \tau^{-1} w$ together with Lemma~\ref{lem:kirchhof_product_stable}. Conversely, assume that $u\in\mathcal{K}(\Gamma)$. Then $u \in C(\Gamma)$, and by statement (i) of Proposition~\ref{prp:regularity-global}, $\tau \in C(\Gamma)$. By the continuity of $\tau$ and $w$, we have
\begin{align*}
    0= \sum_{s\in v} \partial u(s) = \tau^{-1}(v) \sum_{s\in v} \partial w(s) - \dfrac{w(v)}{\tau^2(v)} \sum_{s\in v} \partial\tau(s) = - \dfrac{w(v)}{\tau^2(v)} \sum_{s\in v} \partial\tau(s).
\end{align*}
Proposition~\ref{prp:kirchhoff_vertex_conditions_solution_auxiliary} now implies that $0= w(v)\sum_{s\in v} \partial\tau(s)$. By Proposition~\ref{prp:positive_definite_cov}, the covariance function of $w$ is strictly positive definite. In particular, for every $v \in \mathcal{V}$, we have that $\mathbb{P}\big(w(v) \neq 0\big)=1$. Hence, it follows that $0=\sum_{s\in v} \partial\tau(s)$. Since $v$ was arbitrary, we conclude that $\tau\in\mathcal{K}(\Gamma)$. This completes the proof. \hfill
\end{proof}

Finally, we prove statement (iii) of Proposition~\ref{prp:regularity} and statement (iii) of Proposition~\ref{prp:regularity-global}.

\begin{proof}[Proof of statement (iii) of Proposition~\ref{prp:regularity}]
The claim follows by a straightforward adaptation of the proof of statement (ii) of Proposition~\ref{prp:regularity}. 
Specifically, one replaces the operator $D$ with $D^2$ throughout and applies Lemma~\ref{lem:lemma_secondderivative_bounded} to obtain the corresponding result for $w_e''$, analogous to Proposition~\ref{prp:regularity_derivatives}. \hfill
\end{proof}

To obtain the final result, statement (iii) of Proposition~\ref{prp:regularity-global}, we state and prove some auxiliary results.
    
\begin{prop}
\label{prop:local_holder_plus_cont_at_all_vertex}
    Let $\Gamma = (\mathcal{V}, \mathcal{E})$ be a finite metric graph endowed with the geodesic distance $d(\cdot,\cdot)$. Let $0<\gamma\leq1$ and let $f = (f_e)_{e \in \mathcal{E}}$ be a function on $\Gamma$. If $f_e \in C^{0, \gamma}(e)$ for all $e \in \mathcal{E}$ and $f$ is continuous at every vertex $v\in\mathcal{V}$, then $f\in C^{0, \gamma}(\Gamma)$.
\end{prop}
\begin{proof}
    Let $x,y\in\Gamma$ and let $P = (x=s_0, s_1,\dots,s_k = y)$ denote the shortest path in $\Gamma$ connecting $x$ and $y$, with $s_1,\dots,s_{k-1}$ being vertices of $\Gamma$ (whereas $s_0=x$ and $s_k=y$ are arbitrary points on $\Gamma$). Then $d(x,y) = \sum_{i=1}^k\ell_i$, where $\ell_i$ is the length of the segment between $s_{i-1}$ and $s_i$. Since $f_e \in C^{0, \gamma}(e)$ for all $e \in \mathcal{E}$, we have $|f(s_i)-f(s_{i-1})|\leq C_{e_i} \ell_i^\gamma$, where $e_i$ is the edge containing the segment between $s_{i-1}$ and $s_i$. By continuity of $f$ at every vertex $v\in\mathcal{V}$, the triangle inequality now yields
    \begin{align*}
        |f(x)-f(y)|\leq \sum_{i=1}^k|f(s_i)-f(s_{i-1})| \leq \max_{e_i\in\mathcal{E}}C_{e_i} \sum_{i=1}^k \ell_i^\gamma,
    \end{align*}
    and Jensen's inequality for concave functions gives
    \begin{align*}
        \sum_{i=1}^k \ell_i^\gamma \leq k^{1-\gamma}\left(\sum_{i=1}^k \ell_i\right)^\gamma \leq k^{1-\gamma}d(x,y)^\gamma.
    \end{align*}
    The combination of the preceding inequalities gives the desired result. \hfill
\end{proof}

\begin{prop}
\label{prop:cont_vs_kc_and_zero_der}
Let $f \in C(\Gamma)$ be such that $f'_e \in C^{0,\gamma}(e)$ for all $e \in \mathcal{E}$, for some $0 < \gamma \leq 1$. Then $f'\in C^{0,\gamma}(\Gamma)$ (recall the pathwise definition of continuous derivatives from Section~\ref{subsec:prelim}) if and only if $f$ satisfies the Kirchhoff conditions \eqref{eq:kirchhoff_cond} at all vertices $v$ with $\deg(v) =2$ and $f'_e(v) =0$ for all $e\in\mathcal{E}_v$ at vertices $v$ with $\deg(v)\ge 3$.
\end{prop}

\begin{proof}[Proof of Proposition~\ref{prop:cont_vs_kc_and_zero_der}]
    We divide the proof into two parts. We first prove that $f$ has a continuous derivative $f'$ at all vertices of degree $2$ if and only if $f$ satisfies the Kirchhoff conditions \eqref{eq:kirchhoff_cond} at all vertices of degree $2$. In the second part, we prove that $f$ has a continuous derivative $f'$ at vertices of degree greater than $2$ if and only if $f'_e(v) =0$ for all $e\in\mathcal{E}_v$ at vertices $v$ of degree greater than $2$.
    
    Let $v\in\mathcal{V}$ be a vertex of degree $2$ with incident edges $e_1$ and $e_2$. Recall from the definition of the directional derivative that 
    \begin{align*}
        \partial f(e,v) = 
        \begin{cases}
            f'_e(0), \quad &\text{if }v=0,\\
            -f'_e(\ell_e), \quad &\text{if }v=\ell_e.
        \end{cases}
    \end{align*}
    Take parameterizations of $e_1$ and $e_2$ that make compatible orientations, that is, $e_1 = [0,\ell_1]$ and $e_2=[0,\ell_2]$ with $v$ corresponding to $\ell_1$ on $e_1$ and $0$ on $e_2$. 
    In this case, $\partial f(e_1,v) = -f'_{e_1}(\ell_1)$ and  $\partial f(e_2,v) = f'_{e_2}(0)$. Because of compatible orientation, we can glue the two edges into a single interval $[e_1,e_2]\simeq [0,\ell_1+\ell_2]$ and view $f$ as a function on this interval so that the function values $f_{e_1}(\ell_1) = f(v) = f_{e_2}(0)$ and $f'_{e_1}(\ell_1) = f'(v)= f'_{e_2}(0)$ are well-defined. Observe that 
     \begin{align*}
        f'_{e_1}(\ell_1) = f'(v)= f'_{e_2}(0) \iff \sum_{s \in v} \partial f(s) = -f'_{e_1}(\ell_1) + f'_{e_2}(0) = 0,
    \end{align*}
    which shows that $f'$ is continuous at all vertices of degree $2$ if and only if $f$ satisfies the Kirchhoff conditions \eqref{eq:kirchhoff_cond} at all vertices of degree $2$. This proves the first part.

    Assume that $f'_e(v) =0$ for all $e\in\mathcal{E}_v$ at vertices $v$ of degree greater than $2$. This means that all derivatives vanish and no jumps occur. Then it is immediate that $f'$ is continuous at vertices of degree greater than $2$. Conversely, assume that $f'$ is continuous on $\Gamma$ and let $v$ be a vertex with $\deg(v)\ge 3$.  Then, for all $e,e'\in\mathcal E_v$, continuity of $f'$ implies that $f'_e(v) = f'(v) = f'_{e'}(v)$. By the definition of continuous derivatives, we may choose the parametrization and orientation of the incident edges. 
    Fix three distinct edges $e_1,e_2,e_3\in\mathcal E_v$ and parametrize them so that $v$ corresponds to $\ell_1$ on $e_1$, $\ell_2$ on $e_2$, and $0$ on $e_3$. Continuity of $f'$ along the three paths passing through $v$ yields
    \begin{align*}
        \partial f(e_1, v) &= -\partial f(e_3, v),\qquad
        \partial f(e_2, v) = -\partial f(e_3, v),\qquad
        \partial f(e_1, v) = -\partial f(e_2, v).
    \end{align*}
    Recalling that $\partial f(e_1, v) = -f'_{e_1}(v)$, $\partial f(e_2, v) = -f'_{e_2}(v)$, and $\partial f(e_3, v) = f'_{e_3}(v)$, we obtain
    \begin{align*}
        -f'_{e_1}(v) &= -f'_{e_3}(v),\qquad
        -f'_{e_2}(v) = -f'_{e_3}(v),\qquad
        -f'_{e_1}(v) = f'_{e_2}(v).
    \end{align*}
    Combining these identities gives $-f'_{e_1}(v) = f'_{e_1}(v)$, and therefore $f'_{e_1}(v)=0$. Since $f'$ is continuous at $v$, it follows that $f'_e(v)=0$ for all $e\in\mathcal E_v$. This proves that $f'_e(v)=0$ for all $e\in\mathcal E_v$ at vertices $v$ with $\deg(v)\ge 3$, completing the second part. Finally, since  $f'_e \in C^{0,\gamma}(e)$ for all $e \in \mathcal{E}$, Proposition~\ref{prop:local_holder_plus_cont_at_all_vertex} implies that $f'\in C^{0,\gamma}(\Gamma)$. \hfill
\end{proof}

\begin{remarkapp}
\label{rem:rem_unik}
    Proposition~\ref{prop:cont_vs_kc_and_zero_der} remains valid if the spaces $C^{0,\gamma}(e)$ and $C^{0,\gamma}(\Gamma)$ are replaced by $C(e)$ and $C(\Gamma)$, respectively. The statement and its proof carry over without modification in this setting.
\end{remarkapp}

\begin{prop}
    \label{prp:regularity_second_derivative}
    Let $\alpha > \sfrac{5}{2}$ and let $w$ be the solution to \eqref{eq:auxiliary_problem}. Then for $t,t'\in \Gamma$,
    \begin{equation}
    \label{eq:bound_second_derivative_holder}
        \mathbb{E}\left( |L w(t) - L w(t')|^2 \right) \leq \|L  L^{-\sfrac{\alpha}{2}}\|_{\mathcal{L}(L_2(\Gamma), C^{0,\alpha-\sfrac{5}{2}}(\Gamma))} d(t,t')^{2\alpha-5}, \quad t,t'\in \Gamma.
    \end{equation}
    Further, for any $0 < \gamma < \min\{\alpha-\sfrac{5}{2}, 1\}$, $L w$ has a modification with $\gamma$-H\"older continuous sample paths.
\end{prop}
\begin{proof}
        The proof is in the same spirit as that of Proposition~\ref{prp:regularity_derivatives}. Let $\alpha>\sfrac{5}{2}$ and define $T_\alpha = L L^{-\sfrac{\alpha}{2}}$. Again, $L^{-\sfrac{\alpha}{2}}$ is a bounded operator from $L_2(\Gamma)$ to $\dot{H}^\alpha$ and therefore $T_\alpha = L^{-\sfrac{(\alpha-2)}{2}}$ is a bounded operator from $L_2(\Gamma)$ to $\dot{H}^{\alpha-2}$. By Lemma~\ref{lem:lem0003}, we have $\dot{H}^{\alpha-2}\hookrightarrow C^{0,\alpha-\sfrac{5}{2}}(\Gamma)$, which yields the boundedness of $T_\alpha:L_2(\Gamma)\to C^{0,\alpha-\sfrac{5}{2}}(\Gamma)$. Let $u_0 := \mathcal{W}(T_\alpha^\ast \delta_s)$. Then, by Lemmas~\ref{lem:regularity_field_kolm_chent} and~\ref{lem:regularity_field_kolm_chent_2}, $u_0$ is a centered Gaussian random field with covariance operator $T_\alpha T_\alpha^\ast = L L^{-\alpha} L$ and admits a modification with $\gamma$-H\"older continuous sample paths for any $0 < \gamma < \min\{\alpha-\sfrac{5}{2}, 1\}$. Now, let $w$ be the solution to \eqref{eq:auxiliary_problem}, so that $w$ has covariance operator $L^{-\alpha}$. We define $L w$ to be the centered Gaussian random field with covariance operator $L L^{-\alpha} L$. By construction, $L w$ and $u_0$ have the same finite-dimensional distributions. This implies \eqref{eq:bound_second_derivative_holder} and that $L w$ has a modification with $\gamma$-H\"older continuous sample paths. \hfill
\end{proof}

We are now in a position to prove  statement (iii) of Proposition~\ref{prp:regularity-global}.

\begin{proof}[Proof of statement (iii) of Proposition~\ref{prp:regularity-global}]
Let $\alpha>\sfrac{5}{2}$, and let $w$ be a solution to \eqref{eq:auxiliary_problem}. By Proposition~\ref{prp:regularity_second_derivative}, $L w = \kappa^2w-w''$ has a modification with $\gamma$-H\"older continuous sample paths on $\Gamma$ for any $0 < \gamma < \min\{\alpha-\sfrac{5}{2}, 1\}$. By Lemma~\ref{lem:lem99}, $w$ has a modification with $\gamma$-H\"older continuous sample paths on $\Gamma$. Since $\kappa\in C^{0,\gamma}(\Gamma)$ by assumption, Proposition~\ref{prop:Holder_product} implies that $\kappa^2w$ has $\gamma$-H\"older continuous sample paths on $\Gamma$. This means that $w''=\kappa^2w-Lw$ has $\gamma$-H\"older continuous sample paths on $\Gamma$. \\
Assume now that $\tau,\tau'' \in C(\Gamma)$ and observe that 
\begin{align}
\label{upprime}
    u'' = (\tau^{-1}w)'' = w''\tau^{-1}-2\tau'w'\tau^{-2}-\tau''w\tau^{-2}+2(\tau')^2w\tau^{-3}. 
\end{align}
We now assume $\tau' \in C(\Gamma)$ and aim to prove that $u''\in C(\Gamma)$. By the assumption on $\tau,\tau',\tau''$ and the fact that $w$ and $w''$ have $\gamma$-H\"older continuous sample paths on $\Gamma$, we have that $w''\tau^{-1}-\tau''w\tau^{-2}+2(\tau')^2w\tau^{-3}\in C(\Gamma)$. By Proposition~\ref{prp:kirchhoff_vertex_conditions_solution_auxiliary}, $w$ satisfies the Kirchhoff conditions \eqref{eq:kirchhoff_cond}. In particular, $w$ satisfies the Kirchhoff conditions \eqref{eq:kirchhoff_cond} at all vertices of degree $2$, and therefore $w'$ is continuous at all vertices of degree $2$. Since $\tau' \in C(\Gamma)$, we have that $-2\tau'w'\tau^{-2}$ is continuous at all vertices of degree $2$, and by Proposition~\ref{prop:cont_vs_kc_and_zero_der} together with Remark~\ref{rem:rem_unik}, we have that $\tau'_e(v) =0$ for all $e\in\mathcal{E}_v$ at vertices $v$ with $\deg(v)\ge 3$. The latter implies that $(-2\tau'w'\tau^{-2})(v) =0$ for all $e\in\mathcal{E}_v$ at vertices $v$ with $\deg(v)\ge 3$. This means that $-2\tau'w'\tau^{-2}$ is continuous at vertices $v$ with $\deg(v)\ge 3$, and consequently, $-2\tau'w'\tau^{-2}\in C(\Gamma)$. Altogether, we have show that $u''\in C(\Gamma)$.

Conversely, assume that $u''\in C(\Gamma)$. Our goal is to show that $\tau' \in C(\Gamma)$. 
From \eqref{upprime}, we have that 
\begin{align}
\label{eq:eql1}
    \tau'w' - (\tau')^2w\tau^{-1} = \dfrac{w''\tau}{2}-\dfrac{u''\tau^2}{2}-\dfrac{\tau''w}{2}.
\end{align}
Since $u'',w,w'',\tau,\tau'' \in C(\Gamma)$, the right hand side of \eqref{eq:eql1} is continuous, and therefore $\tau'w' - (\tau')^2w\tau^{-1} \in C(\Gamma)$ almost surely. Consequently, for each vertex $v\in\mathcal{V}$, we have
\begin{align}
\label{eq:eqhh}
    \dfrac{\tau'_{e}(v)w_{e}'(v)\tau(v)}{w(v)} - (\tau_{e}'(v))^2 = \dfrac{\tau'_{e'}(v)w_{e'}'(v)\tau(v)}{w(v)} - (\tau_{e'}'(v))^2, \quad \text{for all } e, e' \in \mathcal{E}_v,
\end{align}
where division by $w(v)$ is justified almost surely, since $\mathbb{P}\big(w(v) \neq 0\big)=1$ by the strict positive-definiteness of the covariance of $w$. For notational convenience, let
\begin{align*}
    A = \tau'_{e}(v),\quad  B = \tau'_{e'}(v),\quad p = \dfrac{w_{e}'(v)\tau(v)}{w(v)},\quad q = \dfrac{w_{e'}'(v)\tau(v)}{w(v)}.
\end{align*}
Then \eqref{eq:eqhh} can be written as
\begin{align}
\label{eq:eqok}
    Ap-A^2 = Bq-B^2, \quad p,q\in\mathbb{R}.
\end{align}
This equality must hold for almost every realization of the Gaussian process $w$. Towards a contradiction, assume that $B \neq 0$. Then the random variable $Bq - B^2$ is non-degenerate, since $q$ is a non-degenerate. In particular, for any $\epsilon > 0$ and any $r > 1$, $p(\omega) \in (A - \epsilon, A + \epsilon)$ and $q(\omega) \in (r B - \epsilon, r B + \epsilon)$ with positive probability. For such a realization $\omega$, we have $A p(\omega) - A^2 \in (-A \epsilon, A \epsilon)$ and $B q(\omega) - B^2 \in ((r-1) B^2 - B \epsilon, (r-1) B^2 + B \epsilon)$. Letting $\epsilon \to 0$, we have $Ap(\omega)-A^2\to0$ and $Bq(\omega)-B^2\to(r-1)B^2\neq 0$. Therefore, for this realization, \eqref{eq:eqok} fails unless $B = 0$. Since the equality must hold almost surely, we conclude $B = \tau'_{e'}(v) = 0$. This implies that $\tau'_{e}(v)=0$ for all $v\in\mathcal{V}$ and for all $e\in\mathcal{E}_v$. This shows that $\tau' \in C(\Gamma)$ and the first part of the statement.\\
Now we prove the second part of the statement. Assume that $\tau\in C^{2,\gamma}(\Gamma)$ and $\tau' \in C^{0,\gamma}(\Gamma)$. This implies that $\tau,\tau'' \in C^{0,\gamma}(\Gamma)$. Since $w,w'' \in C^{0,\gamma}(\Gamma)$, Propositions~\ref{prop:Holder_exponents}, \ref{prop:Holder_sum}, \ref{prop:Holder_product}, and~\ref{prop:Holder_quotient}, imply that $w''\tau^{-1}-\tau''w\tau^{-2}+2(\tau')^2w\tau^{-3}\in C^{0,\gamma}(\Gamma)$. As in the first part of the proof, we also have $-2\tau'w'\tau^{-2}\in C(\Gamma)$. By Proposition~\ref{prp:regularity_derivatives}, $w'_e\in C^{0,\gamma}(e)$ for all $e\in\mathcal{E}$. Since $\tau,\tau'\in C^{0,\gamma}(\Gamma)$, we also have $\tau_e,\tau'_e\in C^{0,\gamma}(e)$ for all $e\in\mathcal{E}$. Propositions~\ref{prop:Holder_reciprocal} and~\ref{prop:Holder_product} imply that $\tau^{-2}_e\in C^{0,\gamma}(e)$ for all $e\in\mathcal{E}$. Proposition~\ref{prop:Holder_product} again implies that $-2\tau'_ew'_e\tau_e^{-2}\in C^{0,\gamma}(e)$ for all $e\in\mathcal{E}$. This now allows us to use Proposition~\ref{prop:local_holder_plus_cont_at_all_vertex} to conclude that $-2\tau'w'\tau^{-2}\in C^{0,\gamma}(\Gamma) $ and consequently, $u''\in C^{0,\gamma}(\Gamma)$. This finishes the proof. \hfill

\end{proof}

\begin{remarkapp}
In the preceding proof, we implicitly assumed that $(w(v), w'_e(v), w'_{e'}(v))$ constitutes a non-degenerate Gaussian vector for every vertex $v$ with $\deg(v)\geq 3$ and every pair of distinct incident edges $e, e'\in\mathcal{E}_v$. While this can be verified in full generality via an analysis of the Cameron--Martin space (a technically involved argument we omit here), we provide a rigorous justification for the case where $\alpha$ is an integer and $\kappa$ is constant.

By the conditional representation developed in \cite{Bolin2023Statistical}, the solution $w$ is obtained by conditioning independent stationary boundaryless Mat\'ern Gaussian processes $\{u_e\}_{e\in\mathcal{E}}$ on the edges to satisfy the Kirchhoff conditions at the vertices. For a fixed vertex $v$ of degree $d\geq 3$, let
\begin{align*}
    \mathbf{X}=\bigl(u_{e_1}(v),u'_{e_1}(v),\dots,u_{e_d}(v),u'_{e_d}(v)\bigr)^\top\in\mathbb{R}^{2d}
\end{align*}
be the vector of boundary values and inward derivatives of these independent edge processes. By \cite[Theorem 6]{Bolin2023Statistical}, we have $\mathbf{X}\sim\mathcal{N}(\mathbf{0},\mathbf{\Sigma})$ with $\mathbf{\Sigma}$ positive definite and block-diagonal across edges.

The Kirchhoff conditions at $v$ impose $d$ independent linear constraints on $\mathbf{X}$: the continuity conditions $u_{e_1}(v)=\cdots=u_{e_d}(v)$ (contributing $d-1$ constraints) and the flux condition $\sum_{i=1}^d u'_{e_i}(v)=0$ (contributing one constraint). Let $\mathbf{A}\in\mathbb{R}^{d\times 2d}$ denote the corresponding full-rank constraint matrix. Then the conditional distribution of $\mathbf{X}$ given $\mathbf{A}\mathbf{X}=\mathbf{0}$ is Gaussian with support $\ker \mathbf{A}$ and covariance $\mathbf{\Sigma}_{\mathrm{cond}}=\mathbf{\Sigma}-\mathbf{\Sigma} \mathbf{A}^\top(\mathbf{A}\mathbf{\Sigma} \mathbf{A}^\top)^{-1}\mathbf{A}\mathbf{\Sigma}$, which is positive definite on $\ker A$.

The vector of interest, $\mathbf{Y}=(w(v),w'_e(v),w'_{e'}(v))$, corresponds, without loss of generality, to $\mathbf{Y}=\bigl(u_{e_1}(v),u'_{e_1}(v),u'_{e_2}(v)\bigr)=\mathbf{B}\mathbf{X}$, where $\mathbf{B}\in\mathbb{R}^{3\times 2d}$ selects the indicated coordinates. On $\ker \mathbf{A}$, the continuity constraints identify all vertex values, so every element of $\ker \mathbf{A}$ is determined by a common value $c$ and derivatives $d_1,\dots,d_d$ satisfying $\sum_{i=1}^d d_i=0$. Thus the restriction of $\mathbf{B}$ to $\ker \mathbf{A}$ is the map $(c,d_1,\dots,d_d)\longmapsto (c,d_1,d_2).$ Since $d\geq 3$, this map is subjective: given any $(a,b_1,b_2)\in\mathbb{R}^3$, we may choose $c=a$, $d_1=b_1$, $d_2=b_2$, $d_3=-(b_1+b_2)$, and $d_4=\cdots=d_d=0$. Hence $\mathbf{B}|_{\ker \mathbf{A}}$ has rank $3$. Consequently, the conditional covariance $\mathbf{B}\mathbf{\Sigma}_{\mathrm{cond}}\mathbf{B}^\top$ is positive definite, and therefore $\mathbf{Y}$ is non-degenerate.
\end{remarkapp}

\section{Proofs of Propositions~\ref{prp:variance_stationary_regularity}, \ref{prp:covariates_regularity}, and~\ref{prp:kriging_predictor_regularity}}
\phantomsection 

\begin{proof}[Proof of Proposition~\ref{prp:variance_stationary_regularity}]
First, let $w$ be the solution to \eqref{eq:auxiliary_problem} and observe that by \citet[Lemma 4.6]{Bolin2024Regularity}, 
\begin{align*}
    \|w(s) - w(s')\|_{L_2(\Omega)} = \sqrt{\mathbb{E}\big(|w(s) - w(s')|^2\big)} \leq \|L^{-\sfrac{\alpha}{2}}\|_{\mathcal{L}(L_2(\Gamma), C^{0,\tilde{\alpha}}(\Gamma))} d(s, s')^{\tilde{\alpha}}.
\end{align*}
Thus, we have that
\begin{align*}
    | \|w(s)\|_{L_2(\Omega)} - \|w(s')\|_{L_2(\Omega)}| \leq \|w(s) - w(s')\|_{L_2(\Omega)} \leq \|L^{-\sfrac{\alpha}{2}}\|_{\mathcal{L}(L_2(\Gamma), C^{0,\tilde{\alpha}}(\Gamma))} d(s, s')^{\tilde{\alpha}}.
\end{align*}
Now, note that $\|w(s)\|_{L_2(\Omega)} = \sqrt{\mathbb{E}(w(s)^2)} = \sigma_\kappa(s)$. Therefore, this shows that $\sigma_\kappa \in C^{0,\tilde{\alpha}}(\Gamma)$. This proves part of statement (i). Further, this shows that $\sigma_\kappa$ is continuous in the compact set $\Gamma$. By Proposition~\ref{prp:positive_definite_cov}, we have that for every $s\in\Gamma$, $\sigma_\kappa(s) >0$. Thus, there exists $\sigma_{min}>0$ such that for every $s\in\Gamma$, $\sigma_\kappa(s)\geq \sigma_{min} >0,$ which shows that Assumption~\ref{assumption1} holds.
The proof of statement (i) is concluded by applying statement (i) of Proposition~\ref{prp:regularity-global}.

Let us now prove statement (ii). To this end, let $\alpha > \sfrac{3}{2}$ and start by observing that for every $e\in\mathcal{E}$, and every $t\in e$, we have that
\begin{align*}
    \tau_e'(t) = \frac{\sigma_0^{-1}}{2\sigma_{\kappa,e}(t)} \left(\partial_1 \rho_e(t,t) + \partial_2 \rho_e(t,t) \right),
\end{align*}
where $\rho(\cdot,\cdot)$ is the covariance function of $w$, $\sigma_{\kappa,e} = \sigma_\kappa|_e$, $\rho_e$ is the restriction of $\rho$ to $e\times e$, $\partial_j$ stands for the partial derivative with respect to the $j$-th argument. Further, by symmetry, we have that $\partial_1 \rho_e(t,t) = \partial_2 \rho_e(t,t)$. Thus, we have that
\begin{align*}
    \tau_e'(t) = \frac{\sigma_0^{-1}\partial_1 \rho_e(t,t)}{\sigma_{\kappa,e}(t)}.
\end{align*}

Now, from the first part of the proof, $\sigma_\kappa \in C^{0,\tilde{\alpha}}(\Gamma)$. Thus, to show that $\tau_e \in C^{1,\alpha-\sfrac{3}{2}}(e)$, it suffices to show that the function $t\mapsto \partial_1 \rho_e(t,t)$ is $(\alpha-\sfrac{3}{2})$-H\"older continuous.

To such an end, observe that by the same arguments as in the proof of Proposition~\ref{prp:kirchhoff_vertex_conditions_solution_auxiliary}, we have that for every $e\in\mathcal{E}$ and every $t\in e$, $\partial_1 \rho_e(t,t) = \mathbb{E}(w'_e(t)w_e(t))$, so that for $t,t'\in e$, we have that, by Cauchy-Schwarz, Proposition~\ref{prp:regularity_derivatives} and \cite[Lemma 4.6]{Bolin2024Regularity}, there exists $C_e>0$, that does not depend on $t,t'$, such that
\begin{align*}
\left|\partial_1 \rho_e(t,t) - \partial_1\rho_e(t,t') \right| &= \left|\mathbb{E}(w_e'(t)w_e(t)) - \mathbb{E}(w_e'(t')w_e(t')) \right| \\
&= \left|\mathbb{E}(w_e'(t)(w_e(t)-w_e(t'))) + \mathbb{E}(w_e(t')(w_e'(t) - w_e'(t'))) \right| \\
&\leq \left|\mathbb{E}(w_e'(t)(w_e(t)-w_e(t'))) \right| + \left|\mathbb{E}(w_e(t')(w_e'(t) - w_e'(t'))) \right| \\
&\leq \sqrt{\mathbb{E}(w_e'(t)^2)}\sqrt{\mathbb{E}((w_e(t)-w_e(t'))^2)}\\
&\quad+ \sqrt{\mathbb{E}(w_e(t')^2)}\sqrt{\mathbb{E}((w_e'(t) - w_e'(t'))^2)}\\
&\leq C_e \sqrt{\mathbb{E}(w_e'(t)^2)} |t-t'|^{\sfrac{1}{2}} + C_e\sqrt{\mathbb{E}(w_e(t')^2)} |t-t'|^{\alpha-\sfrac{3}{2}}\\
& = C_e \sqrt{\mathbb{E}(w_e'(t)^2)} |t-t'|^{\sfrac{1}{2}} + C_e \sigma_{\kappa,e}(t') |t-t'|^{\alpha-\sfrac{3}{2}}.
\end{align*}
Now, observe that $\sigma_\kappa(\cdot)$ is continuous in the compact set $\Gamma$, so that it is bounded. Let $C_\kappa>0$ be such that $\sigma_\kappa(s)\leq C_\kappa$ for every $s\in\Gamma$. Further, let $\widehat{\sigma_\kappa}(t) = \sqrt{\mathbb{E}(w_e'(t)^2)}$, and observe that by Proposition~\ref{prp:regularity_derivatives} and the same arguments as in the first part of the proof, we have that $\widehat{\sigma_\kappa} \in C^{0,\alpha-\sfrac{3}{2}}(\Gamma)$, and since $\Gamma$ is compact, it is bounded. By changing $C_\kappa$ if necessary, we may assume that $\widehat{\sigma_\kappa}(t)\leq C_\kappa$ for every $t\in\Gamma$. Thus, we have that
\begin{align*}
    \left|\partial_1 \rho_e(t,t) - \partial_1\rho_e(t,t') \right| \leq C_e C_\kappa (|t-t'|^{\sfrac{1}{2}} + |t-t'|^{\alpha-\sfrac{3}{2}}) \leq C_e C_\kappa (\ell_e^{2-\alpha}+1) |t-t'|^{\alpha-\sfrac{3}{2}}.
\end{align*}
Therefore, for every $e\in\mathcal{E}$, we have that $t\mapsto \partial_1 \rho_e(t,t)$ is $(\alpha-\sfrac{3}{2})$-H\"older continuous. Hence, in view of the previous results, this shows that $\tau_e \in C^{1,\alpha-\sfrac{3}{2}}(e)$.

It remains to be shown that $\tau$ satisfies the Kirchhoff conditions. By the first part of the proof, $\tau$ is continuous. Now, for every $v\in\mathcal{V}$, we have, by continuity of $\sigma_\kappa$, that
\begin{align*}
    \sum_{\substack{s\in v\\ s=(e,v)}} \partial_e \tau(v) = \sum_{\substack{s\in v\\ s=(e,v)}} \frac{\sigma_0^{-1}}{\sigma_{\kappa,e}(v)} \partial_1 \rho_e(v,v) = \frac{\sigma_0^{-1}}{\sigma_{\kappa}(v)}\sum_{\substack{s\in v\\ s=(e,v)}} \partial_1 \rho_e(v,v).
\end{align*}
Thus, to conclude the proof, we need to show that $\sum_{\substack{s\in v\\ s=(e,v)}} \partial_1 \rho_e(v,v) = 0$. In fact, we have that by continuity of $w$ and Proposition~\ref{prp:kirchhoff_vertex_conditions_solution_auxiliary}, that
\begin{align*}
    \sum_{\substack{s\in v\\ s=(e,v)}}\! \partial_1 \rho_e(v,v) = \!\sum_{s\in v}\! \mathbb{E}(w(s)\partial w(s)) = \mathbb{E}\Bigl(\sum_{s\in v}\! w(s) \partial w(s) \Bigr) = \mathbb{E}\Bigl(w(v)\!\sum_{s\in v}\! \partial w(s)\Bigr) = 0.
\end{align*}
This concludes the proof. \hfill
\end{proof}

\begin{proof}[Proof of Proposition~\ref{prp:covariates_regularity}]
First, let $h:\Gamma\to\mathbb{R}$ be defined as
\begin{align*}
    h(s) = \theta_0 + \sum_{j=1}^m \theta_j g_j(s),\quad s\in\Gamma.
\end{align*}
By Proposition~\ref{prop:Holder_sum} together with the fact that $g_1,\ldots,g_m\in C^{0,\gamma}(\Gamma)$, we have $h\in C^{0,\gamma}(\Gamma)$.

Now, observe that since $\Gamma$ is a compact metric space and $h$ is continuous, there exists $M>0$ such that $|h(s)|\leq M$ for all $s\in\Gamma$. Now, observe that, by the mean-value theorem, $t\mapsto \exp(t)$ is Lipschitz (that is, $1$-H\"older continuous) on the interval $[-M,M]$. Therefore, by using the fact that $h\in C^{0,\gamma}(\Gamma)$ and $\exp(\cdot)$ is Lipschitz in $[-M,M] \supset h(\Gamma)$, we can apply Proposition~\ref{prop:Holder_composition} to obtain that
\begin{align*}
    s\mapsto \exp(h(s)) = \exp\left(\theta_0 + \sum_{j=1}^m \theta_j g_j(s)\right)
\end{align*}
is $\gamma$-H\"older continuous, that is, $\tau \in C^{0,\gamma}(\Gamma)$. This proves the first  part of the proposition.

Now, suppose that for every $e\in\mathcal{E}$, $g_{1,e},\ldots, g_{m,e}\in C^{1,\gamma}(e)$, and that each $g_j$ satisfies the Kirchhoff conditions, that is, equation \eqref{eq:kirchhoff_cond}. Then, since for every $e\in\mathcal{E}$, the functions $\exp(\cdot)$ and $h(\cdot)$ are continuously differentiable on $e$, the function $\tau(\cdot) = \exp(h(\cdot))$ is continuously differentiable on $e$. 

Now, observe that, by the chain rule, $\tau_e'(s) = \exp(h_e(s)) h_e'(s)$, where $h_e(\cdot)$ is the restriction of $h(\cdot)$ to $e$. Further, since $h_e(\cdot)$ is continuously differentiable on $e$ and $e$ is compact, we have that $h_e'(s)$ is bounded on $e$. Therefore, there exists $M>0$ such that $|\tau_e'(s)|\leq M$ for all $s\in e$, and by the mean-value theorem, $\tau_e(\cdot) = \exp(h_e(\cdot))$ is Lipschitz on $e$ (or, $1$-H\"older continuous on $e$). In particular, by Proposition~\ref{prop:Holder_exponents}, $\tau_e(\cdot) \in C^{0,\gamma}(e)$. Furthermore, since $g_{1,e},\ldots, g_{m,e}\in C^{1,\gamma}(e)$, we have that $g_{1,e}',\ldots, g_{m,e}'\in C^{0,\gamma}(e)$. Hence, $h_e'\in C^{0,\gamma}(e)$. Therefore, by Proposition~\ref{prop:Holder_product}, $\tau_e'(\cdot) = \exp(h_e(\cdot)) h_e'(\cdot)$ is $\gamma$-H\"older continuous, that is, $\tau_e'\in C^{0,\gamma}(e)$, which implies that $\tau_e \in C^{1,\gamma}(e)$.

It remains to be shown that $\tau(\cdot)$ satisfies the Kirchhoff conditions. The continuity follows from the first part of the proof. Thus, since $\tau(\cdot) = \exp(h(\cdot))$ is continuous, we have that for every $v\in\mathcal{V}$,
\begin{align*}
\sum_{s\in v} \partial \tau(s) &= \sum_{s\in v} \partial \exp(h(s)) = \sum_{s\in v} \exp(h(s)) \partial h(s)
    = \exp(h(v)) \sum_{s\in v} \partial_e h(v) \\
    &= \exp(h(v)) \sum_{s\in v} \sum_{j=1}^m \theta_j \partial g_{j}(s)
    = \exp(h(v)) \sum_{j=1}^m \theta_j \sum_{s\in v} \partial g_{j}(s) = 0.
\end{align*}
Finally, assume that $g_1,\dots,g_m\in C^{2,\gamma}(\Gamma)$ and $g_1',\dots,g'_m \in C^{0,\gamma}(\Gamma)$. Then $g_1,g''_1 \ldots, g_m,g''_m \in C^{0,\gamma}(\Gamma)$. Observe that
\begin{align}
\label{hpandhpp}
    h'(s) = \sum_{j=1}^m \theta_j g'_j(s)\quad\text{ and }\quad h''(s) = \sum_{j=1}^m \theta_j g''_j(s)
\end{align}
and
\begin{align}
\label{taupandtaupp}
    \tau'(s) = \tau(s)h'(s) \quad\text{ and }\quad\tau''(s) = \tau(s)[(h'(s))^2+h''(s)].
\end{align}
By Proposition~\ref{prop:Holder_sum} and \eqref{hpandhpp}, $h',h''\in C^{0,\gamma}(\Gamma)$, and by the first part of the proof, $\tau \in C^{0,\gamma}(\Gamma)$. By \eqref{taupandtaupp}, and Propositions~\ref{prop:Holder_sum} and~\ref{prop:Holder_product}, we have $\tau',\tau''\in C^{0,\gamma}(\Gamma)$. \hfill
\end{proof}

We will now prove Proposition~\ref{prp:kriging_predictor_regularity}. But first, let us prove an auxiliary result that tells us that the kriging predictor belongs to the Cameron-Martin space associated to the solutions $w$ to equation \eqref{eq:auxiliary_problem_2}. First, recall, from Section~\ref{app:theoretical_details},that the Cameron-Martin space associated to $w$ is $\mathcal{H}_w = \left\{h(s) = \mathbb{E}(v u(s)): s\in\Gamma, v\in H_w\right\}$, where $H_w = \overline{\text{span}\{w(s): s\in\Gamma\}}$ is the Gaussian space associated to $w$.

\begin{lemma}\label{lem:kriging_predictor_regularity_CM}
Under the same assumptions, and using the same notation, as in Proposition~\ref{prp:kriging_predictor_regularity}, the centered kriging predictor $\widetilde{z}(\cdot) := z(\cdot) - \beta_0$ belongs to $\mathcal{H}_w$, where $w$ is the solution to equation \eqref{eq:auxiliary_problem_2}.
\end{lemma}

\begin{proof}
    First, let $\rho(\cdot, \cdot)$ denote the covariance function of $w$. By Proposition~\ref{prp:positive_definite_cov}, $\rho(\cdot, \cdot)$ is strictly positive definite. Define $\boldsymbol{\Sigma}_n = (\rho(s_i, s_j))_{i,j=1}^n$. The kriging predictor $z(\cdot)$ is 
\begin{equation}
\label{eq:kriging_formula}
    z(s) = \beta_0 + \boldsymbol{c}_n(s)^\top \widetilde{\boldsymbol{\Sigma}}_n^{-1} (\boldsymbol{z} - \beta_0),
\end{equation}
where $\boldsymbol{z} = (z_1, \ldots, z_n)$, $\boldsymbol{c}_n(s) = (\rho(s, s_1), \ldots, \rho(s, s_n))^\top$ for all $s \in \Gamma$. For direct observations, $\widetilde{\boldsymbol{\Sigma}}_n = \boldsymbol{\Sigma}_n$, while for noisy observations following \eqref{eq:z_1_z_n_data}, $\widetilde{\boldsymbol{\Sigma}}_n = \boldsymbol{\Sigma}_n + \sigma_\epsilon^2 \boldsymbol{I}_n$, where $\boldsymbol{I}_n$ is the $n \times n$ identity matrix. The matrix $\widetilde{\boldsymbol{\Sigma}}_n^{-1}$ exists because $\rho(\cdot, \cdot)$ is strictly positive definite.

Now, observe that the expression for $\boldsymbol{c}_n(\cdot)$ allows us to write
\begin{align*}
    \widetilde{z}(s) = \boldsymbol{c}_n(s)^\top \widetilde{\boldsymbol{\Sigma}}_n^{-1} (\boldsymbol{z} - \beta_0) = \sum_{j=1}^n \alpha_{j,n} \rho(s, s_j),
\end{align*}
where $\alpha_{j,n}\in\mathbb{R}$ do not depend on $s$. Therefore, by taking $v_j = w(s_j)$, we have that
\begin{align*}
    \rho(\cdot, s_j) = \mathbb{E}(w(s_j)w(\cdot)) = \mathbb{E}(v_j w(\cdot)) \in \mathcal{H}_w,
\end{align*}
for all $j=1,\ldots,n$. The proof is completed by observing that $\widetilde{z}(\cdot)$ is a linear combination of elements in $\mathcal{H}_w$, and that $\mathcal{H}_w$ is a vector space. \hfill
\end{proof}

We are now in a position to prove Proposition~\ref{prp:kriging_predictor_regularity}:

\begin{proof}[Proof of Proposition~\ref{prp:kriging_predictor_regularity}]
    Begin by observing that, by Lemma~\ref{lem:kriging_predictor_regularity_CM}, $\widetilde{z}(\cdot) \in \mathcal{H}_w$. Further, as in the proof of Proposition~\ref{prp:kirchhoff_vertex_conditions_solution_auxiliary}, we have that $\mathcal{H}_w$ is given by $\dot{H}^\alpha$. By Lemma~\ref{lem:lem0003}, we have that $\dot{H}^\alpha \hookrightarrow C^{0,\tilde{\alpha}}(\Gamma)$, where ${\tilde{\alpha} = \min\{\alpha - \sfrac{1}{2}, 1\}}$. This proves that $\widetilde{z}(\cdot) \in C^{0,\tilde{\alpha}}(\Gamma)$. Further, since the constant function equal to $\beta_0$ also belongs to $C^{0,\tilde{\alpha}}(\Gamma)$, we have that $z(\cdot) \in C^{0,\tilde{\alpha}}(\Gamma)$. By Proposition~\ref{prop:Holder_exponents}, and the fact that $\Gamma$ is compact, we have that for any $0 < \gamma \leq \tilde{\alpha}$,
    $z(\cdot) \in C^{0,\gamma}(\Gamma).$
    This proves the first claim. 
    
    Let us now prove the second claim. To such an end, let $\alpha > \sfrac{3}{2}$, and observe that by \citet[Theorem 4.1]{Bolin2024Regularity}, we have that $\dot{H}^\alpha \subset \widetilde{H}^\alpha(\Gamma) \subset \widetilde{H}^1(\Gamma)$. In particular, $z\in \widetilde{H}^1(\Gamma)$, so that for every $e\in\mathcal{E}$, with the identification $e = [0,\ell_e]$, we have 
    \begin{align*}
        z_{e}(x) = z_e(0) + \int_0^x z_e'(t)\, dt,
    \end{align*}
    where $z_e'(\cdot)$ is the weak derivative of $z_e(\cdot)$. Now, let $D$ be the derivative operator acting on $H^\alpha(e)$, and $R_e$ the restriction operator from $H^{\alpha}(\Gamma)$ to $H^{\alpha}(e)$. Then, we have by Lemmas~\ref{lem:lemma_restriction_bounded} and~\ref{lem:lemma_derivative_bounded} that $D R_e : \widetilde{H}^\alpha(\Gamma) \to H^{\alpha-1}(e)$, so that
    \begin{align*}
        z_e'(\cdot) = D R_e z(\cdot) \in H^{\alpha-1}(e).
    \end{align*}
    On the other hand, observe that, by the Sobolev embedding \cite[Corollary 6]{Bolin2024Gaussian}, we have $H^{\alpha-1}(e)\hookrightarrow C^{0,\alpha-\sfrac{3}{2}}(e)$. Thus, $z_{e}'(\cdot) \in C^{0,\alpha-\sfrac{3}{2}}(e)$. In particular, for every $e\in\mathcal{E}$, we have $z_{e}(\cdot) \in C^{1,\alpha-\sfrac{3}{2}}(e)$. By Proposition~\ref{prop:Holder_exponents}, and the fact that $\Gamma$ is compact, we have that for any $0 < \gamma \leq \alpha-\sfrac{1}{2}$, $z_{e}(\cdot) \in C^{1,\gamma}(e)$. To conclude the proof, note that by \citet[Theorem 4.1]{Bolin2024Regularity} and the fact that $\alpha>\sfrac{3}{2}$, we have that $\dot{H}^\alpha \subset K_\alpha(\Gamma)$, where $K_\alpha(\Gamma)$ is given in \eqref{eq:K_alpha_definition}. Thus, since $z(\cdot) \in \dot{H}^\alpha$, it follows that $z(\cdot) \in K_\alpha(\Gamma)$, which in turn implies that $z(\cdot)$ satisfies the Kirchhoff conditions given by \eqref{eq:kirchhoff_cond}. 
    This concludes the proof. \hfill
\end{proof}

\section{FEM details and proof of Proposition~\ref{conv_cov_func}}
\phantomsection 
\label{app:fem_and_conv_cov_f}

Following \cite{Arioli2018AFinite}, the finite element discretization is constructed by subdividing each edge $e\in\mathcal{E}$ into $n_{e}\geq 2$ regular segments of length $h_{e}$, which are delimited by the nodes 
\begin{align*}
    0 = x_0^{e}<x_1^{e}<\dots<x_{n_{e}-1}^{e}< x_{n_{e}}^{e} = \ell_{e}. 
\end{align*}
For each $j = 1,\dots,n_{e}-1$, we consider the following standard hat basis functions 
\begin{equation*}
    \varphi_j^{e}(x)=\begin{cases}
        1-\dfrac{|x_j^{e}-x|}{h_{e}},&\text{ if }x_{j-1}^{e}\leq x\leq x_{j+1}^{e},\\
        0,&\text{ otherwise}.
    \end{cases}
\end{equation*}
For each $e\in\mathcal{E}$, the set of hat functions $\{\varphi_1^{e},\dots,\varphi_{n_{e}-1}^{e}\}$ is a basis for the space
\begin{equation*}
    V_{h_{e}} = \left\{w\in H_0^1(e)\;\Big|\;\forall j = 0,1,\dots,n_{e}-1:w|_{[x_j^{e}, x_{j+1}^{e}]}\in\mathbb{P}^1\right\},
\end{equation*}
where $\mathbb{P}^1$ is the space of linear functions on $[0,\ell_{e}]$. For each vertex $v\in\mathcal{V}$, we define
\begin{equation*}
    \mathcal{N}_v = \left\{\bigcup_{e\in\left\{e\in\mathcal{E}_v: v = x_0^e\right\}}[v,x_1^e]\right\}\bigcup\left\{\bigcup_{e\in\left\{e\in\mathcal{E}_v: v = x^e_{n_e}\right\}}[x^e_{n_e-1},v]\right\},
\end{equation*}
which is a star-shaped set with center at $v$ and rays made of the segments contiguous to $v$. On $\mathcal{N}_v$, we define the hat functions as
\begin{equation*}
    \phi_v(x)=\begin{cases}
        1-\dfrac{|x_v^{e}-x|}{h_{e}},&\text{ if }x\in\mathcal{N}_v\cap e \text{ and }e\in\mathcal{E}_v,\\
        0,&\text{ otherwise},
    \end{cases}
\end{equation*}
where $x_v^e$ is either $x_0^e$ or $x_{n_e}^e$ depending on the edge direction and its parameterization.

Figure~\ref{app:basisfunctions} provides an illustration of the system of basis functions $\{\varphi_j^e, \phi_v\}$ (solid gray lines) on the tadpole graph. Note that for all $e_i\in\mathcal{E}$, $h_{e_i} = 1/4$. Corresponding to node $x_5^{e_2}$, we have plotted the basis function $\varphi_5^{e_2}$ in blue. The sets $\mathcal{N}_{v_1}$ and $\mathcal{N}_{v_2}$ are depicted in green and their corresponding basis functions $\phi_{v_1}$ and $\phi_{v_2}$ are shown in red.

\begin{figure}[!t]
    \centering
    \includegraphics[width=0.5\textwidth]{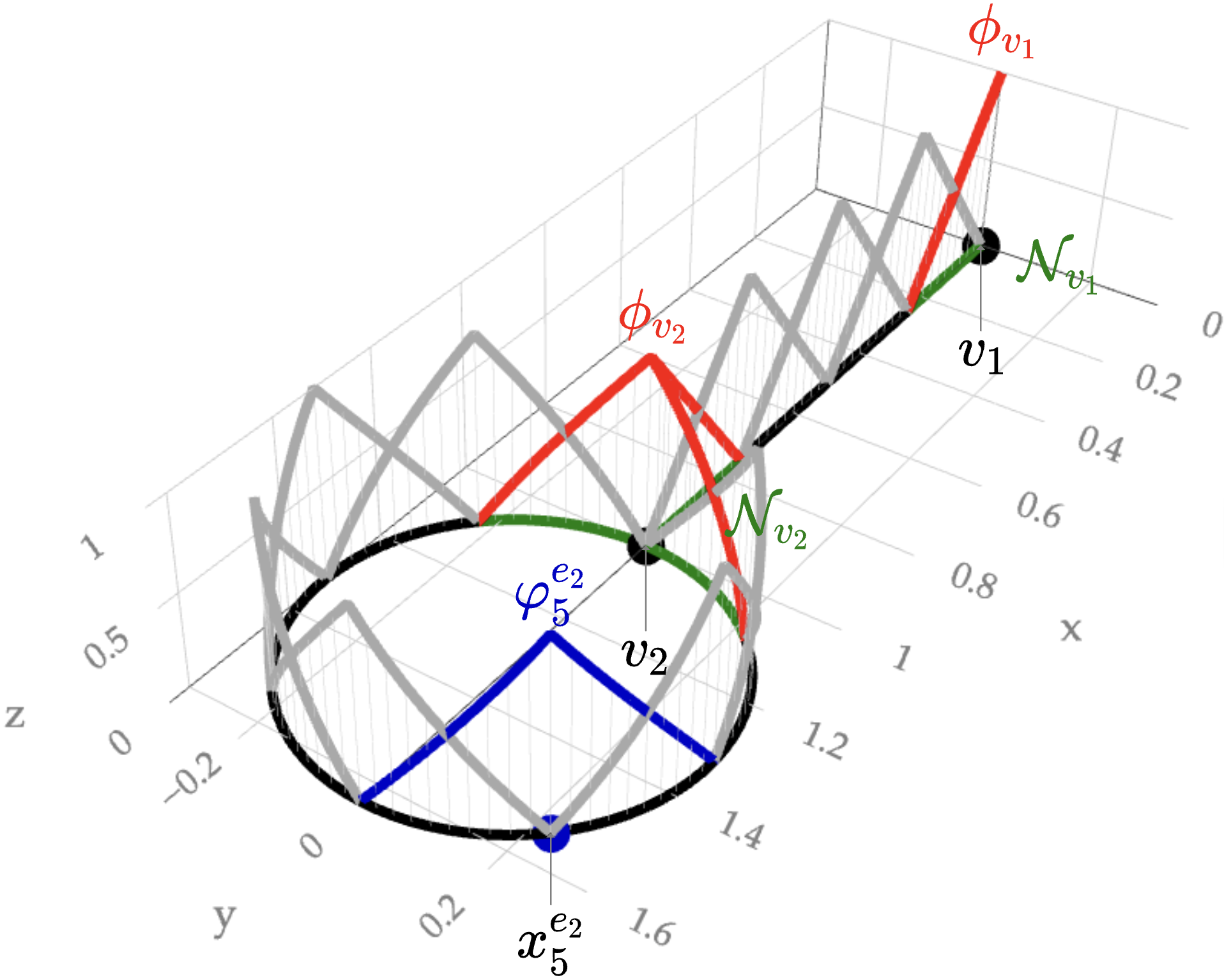}
    \caption{Illustration of the system of basis functions $\{\varphi_j^e, \phi_v\}$ on the tadpole graph.}
    \label{app:basisfunctions}
\end{figure}

Recall, from Subsection~\ref{fem_approx}, the definition of the discretized operator $L_h$. This operator $L_h$ is positive definite and has a collection of eigenvalues $(\lambda_{j,h})_{j=1}^{N_h}$ that can be arranged as $0<\lambda_{1,h}\leq \dots\leq \lambda_{N_h,h}$ and satisfy $\lambda_j\leq\lambda_{j,h},j\in\mathbb{N}$ \citep[Sec. 6.3]{Bolin2024Regularity}. The corresponding eigenfunctions, $(e_{j,h})_j$, are orthonormal in $L_2(\Gamma)$. 

Having introduced the system of basis functions $\{\varphi_j^e, \phi_v\}$ (in the main text, we do not distinguish between them and refer to them jointly as $\{\psi_j\}_{j=1}^{N_h}$), we can now define the finite element space $V_h\subset H^1(\Gamma)$ as $V_h = (\bigoplus_{e\in\mathcal{E}} V_{h_e})\bigoplus V_v$, where $V_v = \text{span}(\{\phi_v:v\in\mathcal{V}\})$ and $\dim(V_h)$ is given by $N_h = |\mathcal{V}| + \sum_{e\in\mathcal{E}}n_e$. 

In what follows, $P_h:L_2(\Gamma)\longrightarrow V_h$ represents the $L_2(\Gamma)$-orthogonal projection onto $V_h$, $h$ denotes $\max_{e\in\mathcal{E}}h_e$, and the expression $A\lesssim_{\alpha_1,\dots,\alpha_k}B$ indicates that $A\leq CB$ for some constant $C= C(\alpha_1,\dots,\alpha_k)$, where the parameters $\alpha_1,\dots,\alpha_k$ represent given data.
\begin{prop}
\label{propfemapprox}
    Let Assumption~\ref{assumption1} hold, let $\alpha>\sfrac{1}{2}$, and $\varrho^\alpha$ be the covariance function of the solution to the auxiliary problem \eqref{eq:auxiliary_problem}. Let, now, $\varrho^\alpha_h$ be the covariance function of the solution to the discretized auxiliary problem 
    $L_h^{\sfrac{\alpha}{2}}w_h = \mathcal{W}_h,$
    where is Gaussian white noise defined on $V_h$. Then,
    \begin{equation}
        \|\varrho^\alpha - \varrho^\alpha_h\|_{L_2(\Gamma\times\Gamma)}\lesssim_{\sigma,\alpha,\kappa,\Gamma}h^\sigma,
    \end{equation}
    where $\sigma < \min\{2\alpha-\sfrac{1}{2},2\}$.
\end{prop}

A detailed proof of this result can be found in \citet[Thm. 6.9]{Bolin2024Regularity}. Proposition~\ref{propfemapprox} provides an approximation of $\varrho^\alpha$ via $\varrho_h^\alpha$, which essentially translates to approximating its corresponding covariance operator $L^{-\alpha}$ with $L_h^{-\alpha}$. This is a consequence of the relationship between the norm of a Hilbert-Schmidt integral operator and its kernel. With this result at hand, we are ready to prove Proposition~\ref{conv_cov_func}.

\begin{proof}[Proof of Proposition~\ref{conv_cov_func}]
    By the relationship of a kernel operator and its kernel function, together with the triangle inequality, we have that
    \begin{align*}
        \|\varrho^\alpha - \varrho_{h,m}^\alpha\|_{L_2(\Gamma\times \Gamma)} &= \|M_{\tau^{-1}}L^{-\alpha}M_{\tau^{-1}}  - M_{\tau^{-1}}r_m(L_h^{-1})P_hM_{\tau^{-1}}\|_{\mathcal{L}_2(L_2(\Gamma))}\\
       &\leq \|M_{\tau^{-1}}L^{-\alpha}_hP_hM_{\tau^{-1}} - M_{\tau^{-1}}r_m(L_h^{-1})P_hM_{\tau^{-1}}\|_{\mathcal{L}_2(L_2(\Gamma))}\\
       &+ \|M_{\tau^{-1}}L^{-\alpha}M_{\tau^{-1}} - M_{\tau^{-1}}L^{-\alpha}_hP_hM_{\tau^{-1}}\|_{\mathcal{L}_2(L_2(\Gamma))}
    \end{align*}
    We can estimate each of the above terms as follows.
    \begin{align*}
        \big\|M_{\tau^{-1}}L^{-\alpha}_hP_hM_{\tau^{-1}} - M_{\tau^{-1}}&r_m(L_h^{-1})P_hM_{\tau^{-1}}\big\|_{\mathcal{L}_2(L_2(\Gamma))}\\
        &= \|M_{\tau^{-1}}\left(L_h^{-\alpha}P_h- r_m(L_h^{-1})P_h\right)M_{\tau^{-1}}\|_{\mathcal{L}_2(L_2(\Gamma))}\\
        & \leq \|M_{\tau^{-1}}\|_{\mathcal{L}(L_2(\Gamma))}\|\left(L_h^{-\alpha}P_h- r_m(L_h^{-1})P_h\right)M_{\tau^{-1}}\|_{\mathcal{L}_2(L_2(\Gamma))}\\
        & \leq \|M_{\tau^{-1}}\|^2_{\mathcal{L}(L_2(\Gamma))}\|L_h^{-\alpha}P_h- r_m(L_h^{-1})P_h\|_{\mathcal{L}_2(L_2(\Gamma))}\\
        & \leq \|\tau^{-1}\|^2_{L^\infty(\Gamma)}\|L_h^{-\alpha}P_h- r_m(L_h^{-1})P_h\|_{\mathcal{L}_2(L_2(\Gamma))},\\
        \big\|M_{\tau^{-1}}L^{-\alpha}M_{\tau^{-1}} - M_{\tau^{-1}}&L^{-\alpha}_hP_hM_{\tau^{-1}}\big\|_{\mathcal{L}_2(L_2(\Gamma))}\\
         &= \|M_{\tau^{-1}}\left(L^{-\alpha}-L^{-\alpha}_hP_h\right)M_{\tau^{-1}}\|_{\mathcal{L}_2(L_2(\Gamma))}\\
        & \leq \|M_{\tau^{-1}}\|_{\mathcal{L}(L_2(\Gamma))}\|\left(L^{-\alpha}-L^{-\alpha}_hP_h\right)M_{\tau^{-1}}\|_{\mathcal{L}_2(L_2(\Gamma))}\\
        & \leq \|M_{\tau^{-1}}\|^2_{\mathcal{L}(L_2(\Gamma))}\|L^{-\alpha}-L^{-\alpha}_hP_h\|_{\mathcal{L}_2(L_2(\Gamma))}\\
        & \leq \|\tau^{-1}\|^2_{L^\infty(\Gamma)}\|L^{-\alpha}-L^{-\alpha}_hP_h\|_{\mathcal{L}_2(L_2(\Gamma))}
    \end{align*}
    Therefore, 
    \begin{align*}
        \|\varrho^\alpha - \varrho_{h,m}^\alpha\|_{L_2(\Gamma\times \Gamma)} & \leq \|\tau^{-1}\|^2_{L^\infty(\Gamma)}\|L_h^{-\alpha}P_h- r_m(L_h^{-1})P_h\|_{\mathcal{L}_2(L_2(\Gamma))}\\
        &\quad+\|\tau^{-1}\|^2_{L^\infty(\Gamma)}\|L^{-\alpha}-L^{-\alpha}_hP_h\|_{\mathcal{L}_2(L_2(\Gamma))}.
    \end{align*}
    From Proposition~\ref{propfemapprox}, we have that 
    \begin{equation}
    \label{1bound}
        \|L^{-\alpha} - L^{-\alpha}_hP_h\|_{\mathcal{L}_2(L_2(\Gamma))}\lesssim_{\sigma,\alpha,\kappa,\Gamma} h^{\sigma}.
    \end{equation}
    This means that we only need to upper bound the term $\|L^{-\alpha}_hP_h - r_m(L_h^{-1})P_h\|_{\mathcal{L}_2(L_2(\Gamma))}$. Recall that the eigenvalues of $L_h$ are $0<\lambda_{1,h}\leq \dots\leq \lambda_{N_h,h}$ with corresponding eigenfunctions $(e_{j,h})_j$, which are orthonormal in $L_2(\Gamma)$. Since $\lambda_j\leq\lambda_{j,h},j\in\mathbb{N}$, we have that
    \begin{equation*}
        0<\dfrac{1}{\lambda_{N_h,h}}\leq\dots\leq\dfrac{1}{\lambda_{1,h}}\leq\dfrac{1}{\lambda_{1}}.
    \end{equation*}
    and therefore $J_h:=[\lambda^{-1}_{N_h,h},\lambda^{-1}_{1,h}]\subset [0,\lambda^{-1}_1] =: J$. Normalizing $L$, we get that $J_h\subset J\subset [0,1]$. Let $f(x) = x^{\alpha}$ and $\hat{f}(x) = x^{\{\alpha\}}$. Since $\alpha=  \lfloor \alpha \rfloor+\{\alpha\}$, we have that $f(x) = x^{\lfloor \alpha \rfloor}\hat{f}(x)$. Let $\hat{r}_m(x)=\dfrac{p(x)}{q(x)}$ be the best $L_\infty$-approximation of $\hat{f}(x)$ on $J_h$. Define $r_m(x) = x^{\lfloor \alpha \rfloor}\hat{r}_m(x)$. Recall that $r_m(L_h^{-1}) = L_h^{-{\lfloor \alpha \rfloor}} p(L_h^{-1})q(L_h^{-1})^{-1}$. Now
    \begin{align}
        \|L^{-\alpha}_hP_h - &r_m(L_h^{-1})P_h\|_{\mathcal{L}_2(L_2(\Gamma))}^2 \notag\\ 
        &= \sum_{j=1}^{N_h} \|L^{-\alpha}_he_{j,h} - r_m(L_h^{-1})e_{j,h}\|_{L_2(\Gamma)}^2
        = \sum_{j=1}^{N_h} \|\lambda^{-\alpha}_{j,h}e_{j,h} - r_m(\lambda_{j,h}^{-1}) e_{j,h}\|_{L_2(\Gamma)}^2\notag\\
        &= \sum_{j=1}^{N_h} (\lambda^{-\alpha}_{j,h}- r_m(\lambda_{j,h}^{-1}))^2\| e_{j,h}\|_{L_2(\Gamma)}^2
        = \sum_{j=1}^{N_h} (\lambda^{-\alpha}_{j,h}- r_m(\lambda_{j,h}^{-1}))^2
        \notag\\
        &\leq \sum_{j=1}^{N_h} \max\big|\lambda^{-\alpha}_{j,h}- r_m(\lambda_{j,h}^{-1})\big|^2
         = {N_h} \max_{1\leq j \leq N_h}\big|\lambda^{-\alpha}_{j,h}- r_m(\lambda_{j,h}^{-1})\big|^2. \label{eqproof1}
    \end{align}
    Note that $x^{\lfloor \alpha \rfloor}\leq1$ on $J_h\subset[0,1]$ because $\lfloor \alpha \rfloor>0$. Note also that $f(x) = x^{\alpha}\leq x^{\{\alpha\}} = \hat{f}(x)$ on $[0,1]$. Hence,
    \begin{equation}
        \max_{1\leq j \leq N_h}\big|(\lambda^{-1}_{j,h})^{\alpha}- r_m(\lambda_{j,h}^{-1})\big| \leq \sup_{x\in J_h}\big|f(x)- r_m(x)\big| \leq \sup_{x\in [0,1]}\big|\hat{f}(x)- \hat{r}_m(x)\big|.\label{eqproof2}
    \end{equation}
    From \citet[Theorem 1]{Stahl2003Best}, we have that 
    \begin{equation}
        \sup_{x\in [0,1]}\big|\hat{f}(x)- \hat{r}_m(x)\big| \lesssim e^{-2\pi\sqrt{\{\alpha\}m}}.\label{eqproof3}
    \end{equation}
    Combining \eqref{eqproof1}, \eqref{eqproof2}, and \eqref{eqproof3}, we obtain that
    \begin{equation*}
        \|L^{-\alpha}_hP_h - r_m(L_h^{-1})P_h\|_{\mathcal{L}_2(L_2(\Gamma))}\lesssim N_h^{\sfrac{1}{2}}e^{-2\pi\sqrt{\{\alpha\}m}}.
    \end{equation*}
    Using that $N_h\lesssim h^{-1}$, 
    $\|L^{-\alpha}_hP_h - r_m(L_h^{-1})P_h\|_{\mathcal{L}_2(L_2(\Gamma))}\lesssim h^{-\sfrac{1}{2}}e^{-2\pi\sqrt{\{\alpha\}m}}$.
    Since this source of error only occurs when we employ the rational approximation, 
    \begin{equation}
    \label{2bound}
        \|L^{-\alpha}_hP_h - r_m(L_h^{-1})P_h\|_{\mathcal{L}_2(L_2(\Gamma))}\lesssim 1_{\alpha\not\in\mathbb{N}} h^{-\sfrac{1}{2}}e^{-2\pi\sqrt{\{\alpha\}m}}.
    \end{equation}
    Combining \eqref{1bound} and \eqref{2bound}, we obtain the desired result.
    \begin{align*} \|\varrho^\alpha - \varrho_{h,m}^\alpha\|_{L_2(\Gamma\times \Gamma)}
    &\lesssim_{\sigma,\alpha,\kappa,\Gamma} \|\tau^{-1}\|^2_{L^\infty(\Gamma)} \left( h^{\sigma} + 1_{\alpha\not\in\mathbb{N}}\cdot h^{-\sfrac{1}{2}}e^{-2\pi\sqrt{\{\alpha\}m}}\right)
    \\
    &\lesssim_{\sigma,\alpha,\kappa,\tau,\Gamma} h^{\sigma} + 1_{\alpha\not\in\mathbb{N}}\cdot h^{-\sfrac{1}{2}}e^{-2\pi\sqrt{\{\alpha\}m}}.
    \end{align*}
    \hfill
\end{proof}

\bibliographystyle{chicago}
\bibliography{reference}
\end{document}